\documentclass[12pt]{article}
\usepackage{amsmath,amsthm,amssymb,mathtools,bm,hyperref,color,multirow}
\usepackage{graphicx,psfrag,epsf}
\usepackage{enumerate}
\usepackage{natbib}
\usepackage{url} 
\usepackage{setspace}

\newcommand{\blind}{0}

\newtheorem{theorem}{Theorem}
 
\newtheorem{lemma}{Lemma}
\newtheorem{example}{Example}

    \usepackage[compact]{titlesec}
    \titlespacing{\section}{0pt}{1ex}{0.75ex}
    \titlespacing{\subsection}{0pt}{1ex}{0.75ex}
    \titlespacing{\subsubsection}{0pt}{1ex}{0.75ex}

\begin{document}

\def\spacingset#1{\renewcommand{\baselinestretch}%
{#1}\small\normalsize} \spacingset{1}


  \title{\bf 
    Facilitating heterogeneous effect estimation via statistically efficient categorical modifiers
  }
  \if0\blind
    {   
  \author{Daniel R. Kowal\thanks{Associate Professor, Department of Statistics and Data Science, Cornell University and Department of Statistics, Rice University (\href{mailto:dan.kowal@cornell.edu}{dan.kowal@cornell.edu}). 
    Research  was sponsored by the National Institute of Environmental Health Sciences (R01ES028819) and the National Science Foundation (SES-2214726). The findings and conclusions in this publication are those of the author(s) and do not necessarily represent the views of the NIH, the U.S. government, or the North Carolina Department of Health and Human Services, Division of Public Health.}\hspace{.2cm}
    } 
  \maketitle
} \fi

\bigskip
\begin{abstract}
Categorical covariates such as race, sex, or group are ubiquitous in regression analysis. While main-only (or ANCOVA) linear models are predominant, \emph{cat-modified} linear models that include categorical-continuous or categorical-categorical interactions are increasingly important and allow heterogeneous, group-specific effects. However, with standard approaches, the addition of cat-modifiers fundamentally alters the estimates and interpretations of the main effects, often inflates their standard errors, and introduces significant concerns about group (e.g., racial) biases. We advocate an alternative parametrization and estimation scheme using \emph{abundance-based constraints} (ABCs). ABCs induce a model parametrization that is both interpretable and equitable. Crucially, we show that with ABCs, the addition of cat-modifiers 1) leaves main effect estimates unchanged and 2) enhances their statistical power, under reasonable conditions. Thus, analysts can, and arguably \emph{should} include cat-modifiers in linear regression models to discover potential heterogeneous effects---without compromising estimation, inference, and interpretability for the main effects. Using simulated data, we verify these invariance properties for estimation and inference and showcase the capabilities of ABCs to increase statistical power. We apply these tools to study demographic heterogeneities among the effects of social and environmental factors on STEM educational outcomes for children in North Carolina. \if0\blind An \texttt{R} package \texttt{lmabc} is available. \fi
\end{abstract}

\noindent%
\vfill

\newpage
\spacingset{1.45} 
\section{Introduction}
\label{sec-intro}
Interactions are remarkably valuable  in linear regression analysis. 
In particular, interactions between a categorical (or nominal) variable and either a continuous or categorical variable---referred to here as \emph{cat-modifiers}---are crucial for discovering and quantifying heterogeneous effects. A prominent example is  race: due to structural racism and discrimination, the effects of many important variables on health and life outcomes vary by race  
\citep{Williams2019},  
with race often interacting with sex or socioeconomic status \citep{Schoendorf1992,Bauer2014}. 
Cat-modifiers are also highly relevant for studying gene-environment interactions \citep{Miao2024} and appear broadly in the social and behavioral sciences \citep{Krefeld2024}. Within statistics, the urgency of cat-modifiers is perhaps best known by Simpson's paradox \citep{Simpson1951}, where the omission of cat-modifiers produces entirely misleading associations.  
 



Yet there are significant obstacles to the inclusion of cat-modifiers in linear regression analysis. Broadly, cat-modifiers alter the interpretation of the main effects, introduce concerns about equity across categorical groups (e.g., for race, sex, and other protected groups), 
change the main effect estimates, and typically inflate the main effect standard errors (SEs). Consequently, cat-modifiers are often omitted or misreported \citep{Knol2009}, which falsely suppresses heterogeneity. 

We argue that, \emph{with the right parametrization},  cat-modifiers can readily, and arguably \emph{should} be included in linear  regression models with categorical covariates. 
To establish ideas, suppose we have $p$ continuous covariates $\bm x = (x_1,\ldots, x_p)^\top$ and $K$ categorical variables $\bm C = (C_1,\ldots, C_K)^\top$ with $L_k$ levels for each categorical variable $k=1,\ldots,K$. We consider regression models for data $\{(\bm x_i, \bm c_i, y_i)\}_{i=1}^n$ parameterized by a linear regression function $\mu(\bm x, \bm c)$ which typically models the conditional expectation $\mathbb{E}(Y \mid \bm x, \bm c)$ or a transformed version for generalized linear models. We distinguish between two classes of linear models: those that do not include cat-modifiers and those that do.  First, the \emph{main-only} model includes multiple continuous and categorical variables, but no interactions:
\begin{equation}
    \label{reg-main}
    \mu^M(\bm x, \bm c) = \alpha_0^M + \bm x^\top \bm \alpha^M + \sum_{k=1}^K \beta_{k, c_k}^M
\end{equation}
or in Wilkinson notation, \texttt{y $\sim$ x$_1$ + \ldots + x$_p$ + c$_1$ + \ldots + c$_K$}. Second, the \emph{cat-modified} model expands \eqref{reg-main} to allow categorical-continuous and categorical-categorical interactions:
\begin{equation}\label{reg-cm}
\mu(\bm x, \bm c) = \alpha_0 + \bm x^\top \bm \alpha + \sum_{k=1}^K \beta_{k, c_k} +  \sum_{k=1}^K \bm x^\top  \bm\gamma_{k, c_k} + \sum_{k=1}^{K-1} \sum_{k'=k+1}^K \gamma_{k,k',c_k, c_{k'}}
\end{equation}
or equivalently, \texttt{y $\sim$ (x$_1$ + \ldots + x$_p$)*(c$_1$ + \ldots + c$_K$) + c$_1$*c$_2$ + \ldots + c$_{K-1}$*c$_K$}, using pairwise interactions for convenience. Our notation emphasizes that the parameters in \eqref{reg-main} and \eqref{reg-cm} are fundamentally distinct, even though these models are nested. 

The advantage of the cat-modified model is the ability to estimate heterogeneous, group-specific effects for each $x_j$. While both models specify \emph{group-specific intercepts} (consider \eqref{reg-main} and \eqref{reg-cm} with $\bm x =\bm 0$), 
only the cat-modified model features \emph{group-specific slopes}: 
\begin{equation} \label{joint-slopes}
\mu_{x_j}'(\bm c) \coloneqq \mu(x_j+1, \bm x_{-j}, \bm c) - \mu(x_j, \bm x_{-j}, \bm c) = \alpha_j + \sum_{k=1}^K \gamma_{j, k, c_k}.
\end{equation}
By comparison, the slopes in the main-only model do not depend on $\bm c$:  ${\mu_{x_j}^M}' \coloneqq \mu^M(x_j+1, \bm x_{-j}, \bm c) - \mu^M(x_j, \bm x_{-j}, \bm c) = \alpha_j^M$.

For concreteness, we consider two popular cases. Empirical examples are given in Tables~\ref{tab:ex-x}~and~\ref{tab:ex-cat}, respectively, and these cases are revisited subsequently. 
\begin{example}[ANCOVA]\label{ex:cts}
    Suppose we have $p=1$ continuous variable $x \in \mathbb{R}$ and $K=1$ categorical variable \texttt{race} with $L_R$ groups. The main-only model \eqref{reg-main} is then 
    \begin{equation}\label{reg-main-x}
        \mu^M(x, r) = \alpha_0^M + x\alpha_1^M + \beta_{r}^M
    \end{equation}
    or equivalently, \texttt{y $\sim$ x + race}, with group-specific intercepts, $\mu^M(0, r) = \alpha_0^M + \beta_r^M$ for each \texttt{race} group $r$, but a global (race-invariant) slope, ${\mu_{x}^M}' =  \alpha_1^M$. Thus, \eqref{reg-main-x} produces parallel lines with race-specific vertical shifts. By comparison, the cat-modified model \eqref{reg-cm} is 
    \begin{equation}\label{reg-cm-x}
    \mu(x, r) = \alpha_0 + x \alpha_1 + \beta_{r} + x \gamma_{r}
    \end{equation}
    or equivalently, \texttt{y $\sim$ x + race + x:race}, with group-specific intercepts $\mu(0, r) = \alpha_0 + \beta_r$ \emph{and} group-specific slopes  $\mu_x'(r) = \alpha_1 + \gamma_r$ for each \texttt{race} group $r$. 
\end{example}

\begin{example}[Two-way ANOVA]\label{ex:cat}
    Suppose we have $K=2$ categorical variables \texttt{race} and \texttt{sex} with $L_R$ and $L_S$ groups, respectively. The main-only  model \eqref{reg-main} is then
    \begin{equation}\label{reg-main-cat}
        \mu^M(r,s) = \alpha_0^M + \beta_{1,r}^M + \beta_{2,s}^M
    \end{equation}
    or equivalently, \texttt{y $\sim$ race + sex}, while the cat-modified model \eqref{reg-cm} is 
    \begin{equation}\label{reg-cm-cat}
    \mu(r, s) = \alpha_0 + \beta_{1,r} + \beta_{2,s} + \gamma_{rs}
    \end{equation}
or equivalently, \texttt{y $\sim$ race + sex + race:sex}. 
\end{example}

The central challenge is that expanding from the main-only model to the cat-modified model  alters the interpretations, estimates, and inference for the \emph{main effects}, i.e., the parameters $\{\alpha_0^M, \bm \alpha^M, \beta_{k,c_k}^M\}$ in \eqref{reg-main} or the analogous terms $\{\alpha_0, \bm \alpha, \beta_{k,c_k}\}$ in \eqref{reg-cm}. 
If these impacts are detrimental, then a quantitative modeler may be reluctant to include cat-modifiers. The key determinant is the model parametrization or \emph{identification} strategy used for the categorical variable coefficients. Specifically, both models \eqref{reg-main} and \eqref{reg-cm} require additional constraints to interpret and estimate the model parameters:  the main-only intercepts $\{\alpha_0^M, \beta_{k,c_k}^M\}$ are overparametrized, while the cat-modified intercepts $\{\alpha_0, \beta_{k,c_k}, \gamma_{k,k', c_k,c_{k'}}\}$ and slopes $\{\bm \alpha, \bm \gamma_{k,c_k}\}$ are overparametrized. The identifications determine the interpretations of all main and interaction parameters and the statistical properties of their estimators.

The most popular identification strategies are  problematic for cat-modified models. \emph{Reference group encoding} (RGE)  is the overwhelming default, including for all major software implementations of  generalized linear regression (\texttt{R}, SAS, Python,  MATLAB, Stata, etc.).  With RGE, a reference group is selected for each categorical variable $C_k$ and removed: $\beta_{k,1}^M = 0$ for all $k$ in \eqref{reg-main} and $\beta_{k,1} = 0$, $\bm \gamma_{k,1} = \bm 0$, $\gamma_{k,k', 1,c_{k'}} = \gamma_{k,k', c_k, 1} =0$ for all $(c_k,c_{k'})$ in \eqref{reg-cm} (using $1$ for each reference group without loss of generality). This is equivalent to using $L_k -1$ ``dummy variables" to encode each $C_k$. 
Despite the simplicity of RGE, the implied notion of ``main effects" in \eqref{reg-cm} significantly impedes the use of cat-modifiers.  For the main-only model, the $j$th main effect is a \emph{global} slope, $\alpha_j^M = {\mu_{x_j}^M}'$, invariant of $\bm c$; yet for the cat-modified model, RGE fixes $ \gamma_{j,k,1} =  0$ for all $j$ so that $\alpha_j = \mu_{x_j}'(\bm 1)$ is the group-specific $x_j$-effect with \emph{all} categorical variables set to their reference groups ($\bm c = \bm 1$). 


First, this main effect parametrization is statistically inefficient: SEs for $\hat\alpha_j$ are typically larger than those for $\hat\alpha_j^M$---the intersection of all reference groups is a subset of the data with a much smaller effective sample size---while the group-specific $x_j$-effect $\mu_{x_j}'(\bm c)$ may be smaller for the reference groups ($\bm c= \bm 1$) than for other groups or globally (i.e., $\alpha_j^M$). We illustrate this effect in Table~\ref{tab:ex-x}: with RGE, the main effects for the cat-modified model are attenuated and sacrifice power compared to those for the main-only model (see also Sections~\ref{sec-sims}~and~\ref{sec-app}). Similar effects occur for categorical-categorical interactions (see Table~\ref{tab:ex-cat}). Of course, these parameters refer to different functionals of $\mu(\bm x, \bm c)$; yet crucially,  they are presented \emph{identically} as ``main effects" in statistical software output and manuscript tables  \citep{Knol2009}. In fact, fewer than half of recent social science publications even reported the reference category \citep{Johfre2021}. This leads to misleading conclusions about effect magnitudes, directions, and heterogeneity \citep{Kowal2024}.


\begin{table}[h]
\centering \scriptsize

\begin{tabular}{llrr} 
\multicolumn{4}{c}{{\bf Reference group encoding (RGE)}} \\ 
Variable & Model & Estimate (SE)   & $p$-value \\
\hline
\multirow{2}{*}{\color{red} \texttt{RI}} & \color{red} Main-only &  \color{red} -0.036 (0.007) &  \color{red} $<$0.001 \\
& \color{red} Cat-modified &  \color{red} -0.022 (0.011) &   \color{red} 0.047 \\
\texttt{RI:White} & Cat-modified& ref &  ref \\
\texttt{RI:Black} & Cat-modified & -0.030 (0.015) & 0.036\\
\texttt{RI:Hispanic} & Cat-modified& 0.038 (0.028) & 0.163\\
\hline
\end{tabular} 
\begin{tabular}{llrr} 
\multicolumn{4}{c}{{\bf Abundance-based constraints (ABCs)}} \\ 
Variable & Model & Estimate (SE)   & $p$-value \\
\hline
\multirow{2}{*}{\color{blue} \texttt{RI}} & \color{blue} Main-only & \color{blue}   -0.036 (0.007) &  \color{blue}  $<$0.001 \\
&\color{blue}  Cat-modified &  \color{blue}  -0.030 (0.007) & \color{blue}   $<$0.001 \\
\texttt{RI:White} & Cat-modified & 0.008 (0.006) &  0.157 \\
\texttt{RI:Black} & Cat-modified & -0.022 (0.009) &  0.014 \\
\texttt{RI:Hispanic} & Cat-modified & 0.047 (0.025) &  0.059 \\
\hline
\end{tabular} 
\caption{\footnotesize Abbreviated regression output for the main-only model  \eqref{reg-main-x} 
and the cat-modified model  \eqref{reg-cm-x} 
for North Carolina end-of-4th-grade reading scores $y$ (see Section~\ref{sec-app}) with $x = $ racial residential isolation (\texttt{RI}) and (mother's) race. 
With RGE (left), the cat-modifier attenuates the \texttt{RI} main effect (red), inflates its SE, and suppresses race-specific \texttt{RI} effects. With ABCs (right), the \texttt{RI} main effect (blue) estimates and SEs are nearly invariant to the cat-modifier (
see Section~\ref{sec-inv}) 
and the output clearly shows that the \texttt{RI} effect is significantly negative and much worse for Black students.  
\label{tab:ex-x}
}
\end{table}

Second, RGE is inequitable: the main effect elevates a single (reference) group above the others. In Table~\ref{tab:ex-x},  White is the reference group: the main effect $\alpha_1 = \mu_x'(\texttt{White})$ is  the $x$-effect \emph{for the  White group}, while the interaction effect $\gamma_r = \mu_x'(r) - \mu_x'(\texttt{White})$ is the difference between the $x$-effect for race $r$ and that for the White group. RGE presents the reference groups (main effects) as ``normal" while all the other groups (interaction effects) are ``deviations from normal"---almost always without any explicit labeling. This framing biases the interpretations of results \citep{Chestnut2018}. 
The problem is compounded for regularized regression: when coefficient estimates are regularized toward zero, the group-specific slopes are \emph{statistically biased} toward the reference group slope ($\gamma_r \to 0$ implies $\mu_x'(r) \to \mu_x'(\texttt{White})$). Beyond the obvious inequities---including (racial, gender, etc.) bias in the estimators---this shrinkage obscures potential differences between the $x$-effects for dominant (e.g., White, Male, etc.) and nondominant groups. 
Thus, RGE undermines progress toward statistical methods that  promote equity \citep{Chen2021}.

Finally, RGE is difficult to interpret: each main effect and interaction in \eqref{reg-cm} must be traced back to \emph{all} reference groups. Consider Example~\ref{ex:cat} (and Table~\ref{tab:ex-cat}), using White and Male for the reference groups:  the main effects in the cat-modified model \eqref{reg-cm-cat} are  $\alpha_0 = \mu(\texttt{White}, \texttt{Male})$, $\beta_{1,r} = \mu(r, \texttt{Male}) - \mu(\texttt{White}, \texttt{Male})$ for each race $r$, and $\beta_{2,s} = \mu(\texttt{White}, s) - \mu(\texttt{White}, \texttt{Male})$ for each sex $s$. Each main effect is anchored at both reference groups, which then affects the interpretations of the interaction effects $\gamma_{rs}$. These challenges are accentuated with multiple categorical covariates and interactions as in \eqref{reg-cm}.

An alternative identification strategy uses \emph{sum-to-zero} constraints (STZ). STZ identifies the parameters by restricting the group-specific coefficients to sum to zero:  $\sum_{\ell=1}^{L_k} \beta_{k,\ell}^M = 0$ for all $k$ in the main-only model and $\sum_{\ell=1}^{L_k} \beta_{k,\ell} = 0$, $\sum_{\ell=1}^{L_k} \gamma_{j,k,\ell} = 0$ for $j=1,\ldots, p$, and  $\sum_{\ell=1}^{L_k} \gamma_{k,k', \ell,c_{k'}} = 0 $ and $\sum_{\ell=1}^{L_k} \gamma_{k,k', c_k,\ell} =0$ for all $(c_k, c_{k'})$ in the cat-modified model. STZ is common for ANOVA models \citep{Scheffe1999,Fujikoshi1993} and has been incorporated into regularized regression \citep{Lim2015}. STZ eliminates the need for a reference group, and thus resolves the inequities of RGE. However, STZ does not offer any special statistical properties for estimation of \eqref{reg-main} or \eqref{reg-cm}, nor does it establish a clear connection between the ``main effects" in \eqref{reg-main} and \eqref{reg-cm}. As a result, it is difficult to interpret the parameters under STZ, while the addition of cat-modifiers may have unpredictable or detrimental effects on the main effect estimates and inferences (see Section~\ref{sec-sims}). 



To address these limitations, we advocate and analyze   \emph{abundance-based constraints} (ABCs) for identification and estimation with cat-modified models. Broadly, ABCs identify parameters using group abundances (see Section~\ref{sec-abcs}). 
ABCs are sufficiently general and may be combined with ordinary least squares (OLS), maximum likelihood, and modern regularized estimation techniques. 
The benefits are summarized by ``EEI":
\begin{enumerate}
    \item {\bf \underline{E}fficiency:} ABCs permit the inclusion of cat-modifiers 1) \emph{without} altering the main effect OLS estimates and 2) either maintaining or \emph{increasing} their statistical power, under reasonable conditions; 
    \item {\bf \underline{E}quity:} ABCs do \emph{not} require a reference group and thus eliminate the alarming inequities under default approaches (RGE); and 
    \item {\bf \underline{I}nterpretability:} main effects  are identified as group-averaged  parameters, interaction effects are group-specific deviations from these group-averaged parameters, and both sets of parameters inherit meaningful notions of sparsity. 
\end{enumerate}
ABCs effectively remove the impediments to cat-modifiers, thus facilitating richer regression analyses of heterogeneous effects. Of course, ABCs cannot guarantee that cat-modifiers will be practically or statistically significant, especially when the effective sample sizes for interactions are small. Rather, with ABCs, there is virtually nothing to lose by expanding from the main-only model \eqref{reg-main} to the cat-modified model \eqref{reg-cm}; yet the potential gains include greater statistical power for the main effects and discovery of  heterogeneous effects.

We emphasize that ABCs, in various forms and by other names, have deep historical roots, but have lacked sufficient motivation to encourage widespread adoption.  \cite{Scheffe1999} and \cite{Fujikoshi1993} considered identification strategies for ANOVA models based on arbitrary group-specific weights. 
Ultimately, both adopted STZ. \cite{Sweeney1972} suggested ABCs for the simple ANCOVA \eqref{reg-main-x} so that the estimated intercept  would  equal the sample mean  (see also Theorem~\ref{thm-intercept}). However, there was no consideration of cat-modifiers or multiple covariates and no case made for any of EEI. Among nonlinear models, \cite{Park2021} and \cite{Park2023} used an ABC-like approach to \emph{avoid} estimating main effects, instead focusing exclusively on interactions to optimize individual treatment rules. More subtly, they required independence between the cat-modifier (treatment) and any modified covariates, which is not usually satisfied for observational data and \emph{not} required for our results. For the two-way ANOVA \eqref{reg-cm-cat}, \cite{Wang2024} briefly mentioned ABCs only to dismiss them, claiming  they ``complicate the interpretation of the model parameters and make it difficult to fit the model...especially when other covariates are present." Here, we forcefully argue the opposite, embodied by EEI---each of which applies 
with multiple covariates present. 
ABCs are considered concurrently in  \cite{Kowal2024}, which focuses on issues of equity with race as a single cat-modifier.  

Contrasts provide an alternative perspective on identification of  \eqref{reg-main} and \eqref{reg-cm}: dummy coding, effects coding, and weighted effects coding (WEC) are respectively linked to RGE, STZ, and ABCs. Although WEC has garnered recent support \citep{TeGrotenhuis2017a,TeGrotenhuis2017},
this work did not consider general cat-modifiers or any EEI. Instead, WEC has been mainly limited to two-way ANOVAs \eqref{reg-main-cat} or \eqref{reg-cm-cat} and only advocated in restrictive settings with 
``certain types of unbalanced data that are missing not at random'' \citep{Brehm2022} or ``categories of different sizes, and if these differences are considered relevant''  \citep{TeGrotenhuis2017}. Our case is much broader and more direct: ABCs are ideal to identify coefficients on any categorical variables and  cat-modifiers should be included in many, if not all linear models. Further, we enforce ABCs using linearly-constrained optimization, which---unlike contrasts---is well-suited for regularized regression (see Section~\ref{sec-est}).


Lastly, we acknowledge additional perspectives on the cat-modified model \eqref{reg-cm}. A widely-used approach is \emph{subgroup analysis}, which subsets the data into groups (for all combinations of $\bm c$) and then fits separate regression models (e.g., \citealp{Pocock2004}). The appeal is that it estimates group-specific slopes without the complicated interpretations of the parameters in \eqref{reg-cm} under default approaches (RGE). However, subgroup analysis does not provide  estimates or inference for the main effects, cannot incorporate regularization or borrow information across groups, and does not allow direct testing for interaction effects. Notably, ABCs offer the same (and more) benefits without any of these drawbacks. 
Related, \cite{Searle1980} advocated for marginal means. These quantities, like group-specific slopes and fitted values, will be identical for all (minimally sufficient) identification strategies under maximum likelihood estimation. Thus, it does not distinguish among identification strategies. 
However, the identification strategy remains key for 1) \emph{parameter} interpretation, estimation, and inference and 2) regularized regression and variable selection. 





The paper is organized as follows. We introduce ABCs  in Section~\ref{sec-abcs}, both for parameter identification and statistical estimation. Our main results on theory for estimation and inference with ABCs are in Section~\ref{sec-inv}. Simulation studies are in Section~\ref{sec-sims} and a real data example is in Section~\ref{sec-app}. We conclude in Section~\ref{sec-conc}. Supplementary material includes proofs of all results, details for generalized linear models, additional simulation results, and supporting data information and analysis.  
\if0\blind An \texttt{R} package \texttt{lmabc} is available. \fi

\section{Identification, estimation, and inference with ABCs} 
\label{sec-abcs}



The goal of ABCs is to enforce model identifiability while maintaining EEI. We first describe the model parametrization and interpretation, and then show how to compute regularized regression estimators and inference  using linearly-constrained optimization. The main properties for estimation and inference are in Section~\ref{sec-inv}.

\subsection{Parameter identification with ABCs}
For motivation, consider \eqref{reg-cm-x} from Example~\ref{ex:cts}: identifiability is obtained by constraining $\sum_{r=1}^{L_R} \pi_r \beta_r = 0$ and $\sum_{r=1}^{L_R} \pi_r \gamma_r = 0$ for some chosen nonnegative weights $\{\pi_r\}_{r=1}^{L_R}$. RGE sets $\pi_1 = 1$ and $\pi_r = 0$ for $r>0$, while STZ sets all $\pi_r = 1$. Instead, suppose we view each constraint as an expectation: $\mathbb{E}_{\pi}(\beta_R) = 0$ and $\mathbb{E}_{\pi}(\gamma_R) = 0$, where  $R$ is a categorical random variable with probabilities $\{\pi_r\}_{r=1}^{L_R}$. Now, 
the main $x$-effect $\alpha_1$ is equivalently the \emph{average} of the group-specific slopes: $\mathbb{E}_{\pi}\{\mu_x'(R)\} = \mathbb{E}_{\pi}(\alpha_1 + \gamma_R) =\alpha_1$. Of course, this notion of ``average"---as well as the  accompanying properties for statistical estimation (Section~\ref{sec-inv})---depends entirely on the supplied probabilities $\{\pi_r\}$. 
 ABCs adopt a natural choice: the (population or sample) abundances by group. For instance, in 
Table~\ref{tab:ex-x}, the ABCs specify
$(\pi_{\texttt{White}}, \pi_{\texttt{Black}}, \pi_{\texttt{Hisp}}) = (0.587, 0.351, 0.062)$ using  sample proportions. 

More broadly, we define ABCs for the cat-modified model \eqref{reg-cm}; special cases such as  \eqref{reg-main} simply omit the constraints for the omitted parameters. We express the ABCs in terms of $\bm{\hat \pi}$, which is the joint proportions across all categorical variables $\bm C = (C_1,\ldots, C_K)^\top$ in the data $\{\bm c_i\}_{i=1}^n$. ABCs may be defined using population or sample proportions; we prefer the latter because they are always available  and estimation properties are tractable and favorable (Section~\ref{sec-inv}). First, ABCs for categorical main effects and categorical-continuous interactions are
\begin{equation}\label{abcs-gen}
    \begin{split}
        \mathbb{E}_{\bm{\hat \pi}} (\beta_{k, C_k}) &= 0, \quad k=1,\ldots, K \\
        \mathbb{E}_{\bm{\hat \pi}} (\gamma_{j,k, C_k}) &= 0, \quad k=1,\ldots, K, \quad j=1,\ldots,p.
    \end{split}
\end{equation}
Equivalently, ABCs may be expressed marginally and with summations: $\sum_{\ell = 1}^{L_k} \hat \pi_{k, \ell} \beta_{k, \ell} = 0$, where $\{\hat \pi_{k,\ell}\}_{\ell=1}^{L_k}$ are the sample proportions for each categorical variable $C_k$, $k=1,\ldots, K$,  and then similarly for each $\{\gamma_{j,k,\ell}\}_{\ell=1}^{L_k}$. The key implication is that, while the cat-modified model \eqref{reg-cm} incorporates heterogeneity via mutual, group-specific slopes \eqref{joint-slopes},  ABCs  concisely identify each main $x_j$-effect as the average of this group-specific slope:
\begin{equation}\label{joint-slopes-avg}
\alpha_j = \mathbb{E}_{\bm{\hat \pi}} \{\mu_{x_j}'(\bm C)\}.
\end{equation}
ABCs parameterize each main $x_j$-effect  by aggregating the group-specific slopes \eqref{joint-slopes}, each weighted by its respective abundance in the data. Unlike RGE, ABCs do not elevate any single (reference) group, and thus avoid the accompanying inequities. 

The identification in \eqref{joint-slopes-avg} also guides interpretation of the group-specific slope parameters $\{\gamma_{j,k,\ell}\}_{\ell=1}^{L_k}$. Consider \eqref{reg-cm-x} from Example~\ref{ex:cts}:  $\gamma_r = \mu_x'(r) - \mathbb{E}_\pi\{\mu_x'(R)\}$ is the difference between the group-specific slope for group $r$ and the group-averaged slope. 
For the general cat-modified model \eqref{reg-cm}, isolating $\gamma_{j,k,c_k}$ requires averaging over the remaining categorical variables $\bm C_{-k}$ with joint proportions $\bm{\hat \pi}_{-k}$ with $C_k=c_k$  fixed: $\gamma_{j,k,c_k} = \mathbb{E}_{\bm{\hat \pi}_{-k}} \{\mu_{x_j}'(\bm C_{-k}, c_k)\} - \mathbb{E}_{\bm{\hat \pi}} \{\mu_{x_j}'(\bm C)\}$. 
Further simplifications are often available, since these averages only must include the categorical variables that act as cat-modifiers for $x_j$. In contrast with RGE, these group-specific coefficients are parameterized relative to a global main effect term \eqref{joint-slopes-avg}, rather than a single (reference) group (White, Male, etc.).

For categorical-categorical interactions, ABCs identify $\{\gamma_{k,k', c_k, c_{k'}}\}$ by requiring
\begin{equation}\label{abcs-cat-cat-gen}
    \begin{split}
    \mathbb{E}_{\bm{\hat \pi}_{C_k \mid C_{k'}=\ell}}(\gamma_{k,k', C_k, \ell}) &= 0, \quad \ell=1,\ldots,L_{k'}\\
    \mathbb{E}_{\bm{\hat \pi}_{C_{k'} \mid C_{k}=\ell}}(\gamma_{k,k', \ell, C_{k'}}) &= 0, \quad \ell=1,\ldots,L_{k}
    \end{split}
\end{equation}
for all $(C_k,C_{k'})$ interactions based on the \emph{conditional} proportions for categorical variable $C_k$ given that the interacting  variable $C_{k'}$ belongs to group $\ell$ (and vice versa). In conjunction,  \eqref{abcs-gen} and \eqref{abcs-cat-cat-gen} constitute ABCs. We illustrate \eqref{abcs-cat-cat-gen} using model  \eqref{reg-cm-cat} from Example~\ref{ex:cat}: 
 $\mathbb{E}_{\bm{\hat \pi}_{S \mid R=r}}(\gamma_{rS}) = 0$ for  $r=1,\ldots,L_R$ and $\mathbb{E}_{\bm{\hat \pi}_{R \mid S=s}}(\gamma_{Rs}) = 0$ for $s=1,\ldots,L_S$, where $\bm{\hat{\pi}}_{S \mid R=r} = \{\hat \pi_{rs}/\pi_r\}_{s=1}^{L_S}$ is the conditional probability for each \texttt{sex} given $\texttt{race} =r$ (similarly for $\bm{\hat{\pi}}_{R \mid S=s}$). Equivalently, \eqref{abcs-cat-cat-gen} may be expressed using the joint proportions $\bm{\hat\pi}$: for Example~\ref{ex:cat}, this is $\sum_{s=1}^{L_S} \hat\pi_{rs}\gamma_{rs} = 0$ for $r=1,\ldots,L_R$ and  $\sum_{r=1}^{L_R} \hat\pi_{rs}\gamma_{rs} = 0$ for  $s=1,\ldots,L_S$. Thus, all ABCs \eqref{abcs-gen} and \eqref{abcs-cat-cat-gen} can be written in terms of the joint probabilities $\bm{\hat \pi}$.  

%



There are several compelling reasons to identify the categorical-categorical interactions with \eqref{abcs-cat-cat-gen}.  First,  it guarantees a global, group-averaged identification for the intercept: 
\begin{lemma}\label{lemma-int}
    Under ABCs, 
    the intercept parameter in  \eqref{reg-cm} satisfies $\mathbb{E}_{\bm{\hat\pi}}\{\mu(\bm 0, \bm C)\} = \alpha_0$.
\end{lemma}
ABCs produce clean expressions and simple interpretations for these main effects: while cat-modifiers induce group-specific intercepts and slopes, ABCs identify suitably global, group-averaged quantities $\alpha_0 = \mathbb{E}_{\bm{\hat\pi}}\{\mu(\bm 0, \bm C)\} $ and  $\alpha_j = \mathbb{E}_{\bm{\hat \pi}} \{\mu_{x_j}'(\bm C)\}$ for $j=1,\ldots,p$. This cannot occur for RGE and only occurs for STZ if the probabilities $\bm{\hat\pi}$ are exactly uniform. 
Second, \eqref{abcs-cat-cat-gen} orthogonalizes the main and interaction categorical effects: in fact, the OLS estimates of the main categorical effects $\{\beta_{k, c_k}\}$ are identical between models that do \eqref{reg-cm-cat} or do not \eqref{reg-main-cat} include cat-modifiers (Theorem~\ref{thm-cat-cat}). Finally, \eqref{abcs-cat-cat-gen} offers the interesting result that, if we were to instead combine the interacted covariates $(C_k, C_{k'})$ into a single categorical variable (e.g., race-sex) with $L_kL_{k'}$ levels, the main effect ABCs \eqref{abcs-gen} would be satisfied for this new categorical variable. Of course, doing so would sacrifice the ability to estimate the main effects $\{\beta_{k, c_k}\}$, but this internal consistency is reassuring. 

For implementation, it is sufficient to enforce $L_k + L_{k'} -1$ of the $L_k + L_{k'}$ constraints in \eqref{abcs-cat-cat-gen}. The choice of omitted constraint is arbitrary, since all constraints \eqref{abcs-cat-cat-gen} hold regardless:
\begin{lemma} \label{lemma-cat-cat}
    Suppose we apply \eqref{abcs-cat-cat-gen} to all but one interaction term: 
    $\mathbb{E}_{\bm{\hat \pi}_{C_k \mid C_{k'}=\ell}}(\gamma_{k,k', C_k, \ell}) = 0$ for $\ell=1,\ldots,L_{k'}$ and 
    $\mathbb{E}_{\bm{\hat \pi}_{C_{k'} \mid C_{k}=\ell}}(\gamma_{k,k', \ell, C_{k'}}) = 0$ for  $\ell=2,\ldots,L_{k}$. Then the same constraint holds for $\ell=1$: $\mathbb{E}_{\bm{\hat \pi}_{C_{k'} \mid C_{k}=1}}(\gamma_{k,k', 1, C_{k'}}) = 0$. 
\end{lemma}

Finally, we emphasize that ABCs \eqref{abcs-gen} and \eqref{abcs-cat-cat-gen} are designed for parameter identification in the general cat-modified model \eqref{reg-cm}, which may be featured in generalized linear models (see the supplementary material, Section~\ref{sec-a-glm}) and includes numerous important special cases, such as  main-only models \eqref{reg-main}, ANCOVA models (Example~\ref{ex:cts}), and two-way ANOVA models (Example~\ref{ex:cat}), among many others.

\subsection{Estimation, inference, and sparsity with ABCs}\label{sec-est}
ABCs are linear constraints and thus readily compatible with regularized regression. First, we consolidate the cat-modified model \eqref{reg-cm} into a traditional regression structure: $ (\mu(\bm x_1, \bm c_1), \ldots, \mu(\bm x_n, \bm c_n))^\top = \bm X\bm\theta$, where  $\bm X$ is the $n \times P$ matrix that includes an intercept, all (centered) continuous covariates, indicator variables for all levels of each categorical variable, and all specified interactions, and $\bm \theta$ include all unknown regression coefficients. In the presence of at least one categorical covariate, $\bm X$ is rank deficient, say $\mbox{rank}(\bm X) = P-m$. We represent all ABCs \eqref{abcs-gen} and \eqref{abcs-cat-cat-gen} generically as $\bm A_{\bm{\hat\pi}} \bm \theta = \bm 0$, where is the $m \times P$ matrix of constraints with $\mbox{rank}( \bm A_{\bm{\hat\pi}}) = m$. 
Then, for a loss function $\mathcal{L}(\bm y, \bm X \bm \theta)$ for data $\bm y = (y_1,\ldots,y_n)^\top$ and a coefficient penalty $\mathcal{P}(\bm\theta)$, we aim to solve 
\begin{equation}
    \label{est-pen}
        \bm{\hat \theta}(\lambda) = \arg\min_{\bm \theta} 
    \mathcal{L}(\bm y, \bm X \bm \theta)  + \lambda\mathcal{P}(\bm\theta)  \quad \mbox{subject to }  \bm A_{\bm{\hat\pi}} \bm \theta = \bm 0
\end{equation}
and $\lambda \ge 0$ is a tuning parameter. We primarily focus on squared error loss  $\mathcal{L}(\bm y, \bm X \bm \theta) = \Vert \bm y -  \bm X \bm \theta \Vert^2$ and either unpenalized estimation ($\lambda=0$) or (group) lasso and ridge regression with $\lambda$ selected by cross-validation. 
One way to solve \eqref{est-pen} is to reparametrize to an unconstrained space with only $P-m$ parameters. 
Let $\bm A_{\bm{\hat\pi}}^\top = \bm Q \bm R$ be the QR-decomposition 
with  columnwise partitioning of the $P \times P$ orthogonal matrix $\bm Q = (\bm Q_{1:m} : \bm Q_{\bm{\hat\pi}} ) $ and similarly, $\bm R^\top = (\bm R_{1:m, 1:m} : \bm 0)$. 
By construction, $\bm A_{\bm{\hat\pi}} \bm Q_{\bm{\hat\pi}} = \bm 0$, so that for any $(P-m)$-dimensional vector $\bm \theta_Q$, the vector $\bm \theta = \bm Q_{\bm{\hat\pi}} \bm \theta_Q$ satisfies  $ \bm A_{\bm{\hat\pi}} \bm \theta = \bm 0$. Then, letting $\bm X_Q \coloneqq \bm X  \bm Q_{\bm{\hat\pi}}$,  \eqref{est-pen} is equivalently
\begin{equation}
    \label{pen-abc}
    \bm{\hat \theta}(\lambda) = \bm Q_{\bm{\hat\pi}} \bm{\hat \theta}_Q(\lambda), \quad \bm{\hat \theta}_Q(\lambda) = \arg\min_{\bm \theta} 
    \mathcal{L}(\bm y, \bm X_Q \bm \theta_Q) + \lambda \mathcal{P}(\bm Q_{\bm{\hat\pi}} \bm{\theta}_Q).
\end{equation}
Regularized regression with ABCs simply requires 1) computing the QR decomposition of $\bm A_{\bm{\hat\pi}}^\top$ and 2)  solving an unconstrained regularized regression problem. 
When $\mathcal{L}$ is a negative log-likelihood,  $\bm{\hat \theta}_Q \coloneqq \bm{\hat \theta}_Q(0)$ is a maximum likelihood estimator (MLE) and so is  $\bm{\hat \theta}$. Hence, usual properties for MLEs apply to estimators with ABCs. Under standard regularity conditions,  \eqref{pen-abc} satisfies 
    $\sqrt{n}(\bm{\hat \theta} - \bm \theta) \stackrel{d}{\rightarrow} N_{P}(\bm 0, \bm Q_{\bm{\pi}}\mathcal{I}(\bm \theta_Q)^{-1}\bm Q_{\bm{\pi}}^\top) 
$ 
where $\mathcal{I}$ is the Fisher information associated with $\bm{\theta}_Q$ and $\bm\pi$ is the joint population probabilities for the categorical covariates $\bm C$. Thus, it is straightforward to construct confidence intervals and conduct hypothesis tests for the coefficients $\bm \theta$.  When the model errors $y_i - \mu(\bm x_i, \bm c_i)$ are Gaussian, uncorrelated, and homoskedastic, 
the OLS estimator under ABCs satisfies 
$\bm{\hat \theta} \sim N_{P}\{\bm{\theta}, \sigma^2 \bm Q_{\bm{\hat\pi}}(\bm X_Q ^\top\bm X_Q)^{-1}\bm Q_{\bm{\hat\pi}}^\top \}$, even in finite samples. Although this distribution does not account for the sampling variability in $\bm{\hat\pi}$, this is typically quite small relative to the variability in $\bm{\hat \theta}$. Our empirical analyses suggest that no further adjustments are needed (see Section~\ref{sec-sims}). 


Finally, we emphasize the unique challenges of regularization and selection  for cat-modified models. Selection of interaction effects has primarily focused on high-dimensional, continuous-continuous interactions \citep{Bien2013,Lim2015}. For cat-modified models with RGE,  coefficient  shrinkage introduces (racial, gender, etc.) biases: $\gamma_{j,k,c_k} \to 0$ implies $\mu_{x_j}'(\bm c) \to \mu_{x_j}'(\bm 1)$, so group-specific effects are pulled toward those for the reference (White, Male, etc.) groups.  With ABCs, no such biases occur: $\gamma_{j,k,c_k} \to 0$ implies $\mu_{x_j}'(\bm c) \to \mathbb{E}_{\bm{\hat \pi}} \{\mu_{x_j}'(\bm C)\}$ collapses to the group-averaged $x_j$-effect, which produces a reasonable notion of parameter sparsity.  

When $\lambda > 0$, it is possible to omit constraints and still obtain unique estimators. However, these estimators do not target identifiable parameters and thus are difficult to interpret. For lasso estimation, \cite{Kowal2024} observes that such ``overparametrized" estimation tends to reproduce RGE by implicitly selecting a reference group, and thus inherits the same limitations as RGE.



\section{Theory for estimation and inference with ABCs}\label{sec-inv}
A central nuisance with interactions is that they change the main effect estimates and SEs. Here, we show that ABCs circumvent these challenges for cat-modifiers. The main point is that, with ABCs, the \emph{addition} of cat-modifiers is either 1) harmless, since it has little to no impact on main effects estimates and inference, or 2) beneficial, since it can reveal heterogeneity and improve statistical power for the main effects.



\subsection{Estimation invariance with ABCs}\label{sec-est-inv}
We establish conditions under which main effect OLS estimates are \emph{invariant} to the addition of cat-modifiers under ABCs. These results make  minimal  assumptions about the true data-generating process and do not apply for other identifications (RGE, STZ, etc.). 

First, consider OLS estimation of the intercept. For an enormous class of linear models---with arbitrarily many continuous covariates, categorical covariates, and categorical-categorical interactions---ABCs ensure that the OLS-estimated intercept is always \emph{exactly} equal to the sample mean, $\hat \alpha_0 = \bar y \coloneqq n^{-1}\sum_{i=1}^n y_i$.
\begin{theorem}\label{thm-intercept}
For \emph{any} linear model of the form \eqref{reg-cm} with 1) centered continuous covariates ($\bm{\bar x} = \bm 0$),  2) no categorical-continuous interactions (all $\bm\gamma_{k,c_k} = \bm 0$), and 3) ABCs \eqref{abcs-gen} and \eqref{abcs-cat-cat-gen}, 
the OLS estimate of the intercept is $\hat \alpha_0 = \bar y$. 
\end{theorem}
Simple models such as \texttt{y $\sim$ race} yield the same intercept estimate as more complex models like \texttt{y $\sim$ x$_1$ + \ldots + x$_p$ + race + sex + race:sex}.  This reaffirms the global interpretation of the intercept: under ABCs, $\hat \alpha_0$ targets the global intercept $\alpha_0 = \mathbb{E}_{\bm{\hat \pi}} \{\mu(\bm 0, \bm C)\}$ (Lemma~\ref{lemma-int}). Of course, $\bar y$ is a good estimator for the marginal expectation of $Y$, so  $\hat \alpha_0$ is appropriately global---even in the presence of categorical variables and their interactions. For models with at least one categorical variable, this result cannot occur for any other identification (RGE, STZ, etc.).  Theorem~\ref{thm-intercept} extends   \cite{Sweeney1972} to allow for categorical-categorical interactions and arbitrarily many continuous covariates. 

Next, consider the impact of adding categorical-categorical interactions on estimation of the main effects. For concreteness, we consider a two-way ANOVA (Example~\ref{ex:cat}).
\begin{theorem}\label{thm-cat-cat}
    Under ABCs \eqref{abcs-gen} and \eqref{abcs-cat-cat-gen}, the OLS estimates of \emph{all} main effects are identical under the main-only model \eqref{reg-main-cat} and the cat-modified model \eqref{reg-cm-cat}: $\hat \alpha_0^M = \hat \alpha_0$, $\hat \beta_{1,r}^M = \hat \beta_{1,r}$ for all $r=1,\ldots,L_R$, and $\hat \beta_{2,s}^M = \hat \beta_{2,s}$ for all $s=1,\ldots,L_S$.
\end{theorem}
This estimation invariance applies to \emph{all} ($1 + L_R + L_S$) main effects in \eqref{reg-cm-cat}. Implicitly, we assume that the OLS estimates exist and are unique (i.e., empty categories are not permitted), but otherwise there are no requirements on the data-generating process. In particular, there are no assumptions of independence or uncorrelateness between the categorical covariates and no assumptions about their relationship with $Y$.  ABCs deliver a natural interpretation of the categorical variable coefficients: the main effects are deviations from the global mean (Theorem~\ref{thm-intercept}), while the interaction effects are deviations from the main effects in the main-only model (Theorem~\ref{thm-cat-cat}).  This result validates our choice of ABCs \eqref{abcs-gen} and especially \eqref{abcs-cat-cat-gen}. Again, such invariance does not occur for other identifications.


Finally, we consider the addition of categorical-continuous interactions to main-only models. For clarity, we focus on a single categorical variable ($K=1$) with levels $r=1,\ldots, L_R$, but showcase these principles empirically with multiple categorical variables (Section~\ref{sec-app}).  Following Example~\ref{ex:cts}, we begin with a single continuous covariate $x$ ($p=1$). Let $\hat \sigma_{x[r]}^2 \coloneqq n_r^{-1} s_{x[r]}^2 - \bar x_r^2$ be the (scaled) sample variance of $\{x_i\}_{i=1}^n$ for each group $r$, where $n_r = n\hat\pi_r$, $s_{x[r]}^2 = \sum_{r_i=r} x_i^2$ and $\bar x_r = n_r^{-1} \sum_{r_i=r} x_i$. If the continuous covariate has  the same scale for each group, then the OLS estimate of the coefficient on $x$ is the same whether or not the cat-modifier is included.
\begin{theorem}\label{thm-cts}
    Under ABCs \eqref{abcs-gen} and the equal-variance condition
    \begin{equation} \label{eq-v}
    \hat \sigma_{x[r]}^2 = \hat \sigma_{x[1]}^2 \quad \mbox{for all } r=1,\ldots,L_R,
\end{equation}
    the OLS estimates for the main-only model \eqref{reg-main-x} and the cat-modified model \eqref{reg-cm-x} satisfy  \emph{estimation invariance}, $\hat \alpha_1^M = \hat \alpha_1.$
\end{theorem}
We apply Theorem~\ref{thm-cts} as an approximation, $\hat \alpha_1 \approx \hat \alpha_1^M$ whenever $\hat \sigma_{x[r]}^2 \approx \hat \sigma_{x[1]}^2$ for all $r$, which is reasonably robust to deviations from  \eqref{eq-v} (see Table~\ref{tab:ex-x} and Section~\ref{sims-cts-cat}). This condition makes no requirements on the true associations between $Y$ and $(x,r)$ and generally allows the distribution of $x$ to vary by $r$---as long as the scale is approximately constant. In particular, strong dependencies between $x$ and $r$ are permissible.

It is clarifying to consider violations of the equal-variance condition \eqref{eq-v}, so that $x$ varies substantially for some groups, but varies little for others. This scenario does \emph{not} invalidate estimation with ABCs; rather, it decouples the coefficients on $x$ (i.e., the main $x$-effects) from the models that do \eqref{reg-cm-x} or do not \eqref{reg-main-x} include a cat-modifier. Arguably, the main-only model \eqref{reg-main-x} is no longer appropriate in this setting. The $x$-effect in \eqref{reg-main-x} is $\alpha_1^M = \mu^M(x+1) - \mu^M(x)$, which considers a one-unit change in $x$ \emph{regardless of the group $r$}. But when \eqref{eq-v} is violated, the scale of $x$---and a ``one-unit change in $x$"---is no longer comparable across  groups. Thus,  \emph{group-specific} slopes $\mu(x+1, r) - \mu(x, r) = \alpha_1 + \gamma_r$ are appropriate, which mandates the cat-modified model. As the global $x$-effect $\alpha_1^M$ from the main-only model is no longer appealing, the cat-modified model with ABCs instead identifies a global $x$-effect via the group-averaged quantity $\alpha_1 = \mathbb{E}_{\bm{\hat \pi}}\{\mu(x+1, R) - \mu(x, R)\}$ as in \eqref{joint-slopes-avg}. In this setting, distinctness between $\alpha_1$ and $\alpha_1^M$ is appropriate.

This result can be extended  for $p$ continuous covariates, each of which is cat-modified:  \texttt{y $\sim$ x$_1$ + \ldots + x$_p$ + c$_1$ + x$_1$:c$_1$ + \ldots + x$_p$:c$_1$}. Here, the equal-variance condition \eqref{eq-v}  instead uses the (scaled) sample covariance between $x_j$ and $x_h$ in group $r$, $\widehat{\mbox{Cov}}_r(\bm x_j, \bm{x}_h) \coloneqq \ n_r^{-1} \sum_{r_i = r} (x_{ij} - \bar x_j)(x_{ih} - \bar x_h)$. 
\begin{theorem}
    \label{thm-int-full}
    Consider the main-only model \eqref{reg-main} and the cat-modified model 
    \eqref{reg-cm}, each with $K=1$ categorical variable. 
    Under  ABCs \eqref{abcs-gen} and the equal-covariance condition 
    $\widehat{\mbox{Cov}}_r(\bm x_j, \bm{x}_h) = \widehat{\mbox{Cov}}_1(\bm x_j, \bm{x}_h)$
    for all $r=1,\ldots,L_R$ and each $j,h=1,\ldots,p$, 
    the OLS estimates satisfy $\bm{\hat \alpha} = \bm{\hat \alpha}^M $. 
\end{theorem}
Theorem~\ref{thm-int-full} ensures estimation invariance for \emph{all} $p$ continuous main effects, each of which is cat-modified. Thus, the equal-covariance condition is stricter than \eqref{eq-v}. As with Theorem~\ref{thm-cts}, we apply Theorem~\ref{thm-int-full} as an approximation, so that $\bm{\hat \alpha} \approx \bm{\hat \alpha}^M $ when equal-covariance approximately holds. 

Lastly, we establish a middle ground: \texttt{y $\sim$ x$_1$ + \ldots + x$_p$ + c$_1$ + x$_1$:c$_1$}, which is a cat-modified model with $p$ continuous covariates and $K=1$ categorical variable, but now only $x_1$ is cat-modified. Instead of covariances between all pairs of covariates, the equal-variance condition now involves only $x_1$ and the residuals $\bm{\hat e}_{1}$ from regressing $x_1$ on all other variables, \texttt{x$_1$ $\sim$ x$_2$ + \ldots + x$_p$ + c$_1$}. 
\begin{theorem}\label{thm-int-local}
    Consider the main-only model \eqref{reg-main} with $K=1$ and the  cat-modified model 
    \eqref{reg-cm} with $K=1$ and interactions only with $x_1$ (fix $\gamma_{j,r} = 0$ for all $j > 1$ and $r=1,\ldots,L_R$). 
    Under ABCs \eqref{abcs-gen}  and the equal-variance condition 
    $\widehat{\mbox{Cov}}_r(\bm x_1, \bm{\hat e}_1) = \widehat{\mbox{Cov}}_1(\bm x_1, \bm{\hat e}_1)$
    for all $r=1,\ldots,L_R$, the OLS estimates satisfy $\hat \alpha_1^M = \hat \alpha_1$. 
\end{theorem}
To understand this modified equal-variance condition, we can equivalently express $\widehat{\mbox{Cov}}_r(\bm{\hat e}_1, \bm{x}_1) = n_r^{-1}\sum_{r_i = r} (x_{i1}^2 - x_{i1}\hat x_{i1}) = \hat \sigma_{x_1[r]}^2 - \widehat{\mbox{Cov}}_r(\bm{\hat x}_1, \bm{x}_1)$, where $\bm{\hat x}_1$ are the fitted values from \texttt{x$_1$ $\sim$ x$_2$ + \ldots + x$_p$ + c$_1$}. Theorem~\ref{thm-int-local} requires that the variability in $x_1$ explained by the remaining (continuous and categorical) covariates is the same within each group. When this condition is violated, a one-unit change in $x_1$ holding \emph{all else equal} among $x_2,\ldots,x_p$ is no longer comparable across groups $r$. As with Theorem~\ref{thm-cts}, the main-only model $x$-effect $\alpha_1^M$ is no longer appropriate; group-specific $x_1$-effects $\mu_{x_1}'(r) = \alpha_1 + \gamma_{1,r}$ are preferred; and ABCs offer a substitute for the global slope parameter via the group-averaged $x_1$-effect, $\alpha_1 = \mathbb{E}_{\bm{\hat \pi}}\{\mu_{x_1}'(R)\}$.

\subsection{Powerful inference with ABCs}
A primary reason for the unpopularity of cat-modifiers is the loss of statistical power for the main effects. With RGE, cat-modifiers relegate the main effects to a single reference group, which shrinks the effective sample size and often attenuates global effects. Thus,  quantitative modelers may be reluctant to include cat-modifiers for fear of larger $p$-values, wider confidence intervals, and less power to identify important effects. Consequently, potential race-, sex-, or other group-specific effects may remain hidden.

ABCs directly and uniquely address this challenge. With the addition of cat-modifier effects, we show that ABCs may actually \emph{reduce} SEs for the main effects. The magnitude of this reduction increases with the effect size of the cat-modifier. Crucially, when the cat-modifier effect is unnecessary, then the main effect SEs match, but do not inflate, those for a (correct) main-only model.

Consider two nested models, a main-only model and a cat-modified model. Our general result is that the cat-modified model with ABCs has smaller SEs for the main effects whenever the estimated \emph{residual} variance is smaller for the cat-modified model,
\begin{equation} \label{rv}
    \hat S^2 \le \hat S_M^2.
\end{equation}
For the maximum likelihood estimators  $\hat S^2 =  \Vert \bm{\hat e} \Vert^2/n$ and $\hat S_M^2 =  \Vert \bm{\hat e}_M \Vert^2/n $, where $\bm{\hat e}$ and $\bm{\hat e}_M$ are the residuals from the cat-modified and main-only models,  respectively, \eqref{rv} is guaranteed: $\Vert \bm{\hat e} \Vert^2 \le \Vert \bm{\hat e}_M \Vert^2$, typically with strict inequality. More commonly, the unbiased estimators  $\hat S^2 =  \Vert \bm{\hat e} \Vert^2/(n - d_M - d)$ and $\hat S_M^2 =  \Vert \bm{\hat e}_M \Vert^2/(n - d_M) $ are used, where $d_M+d$ and $d_M$ are the number of identified parameters for the cat-modified and main-only models, respectively. In that case, \eqref{rv} 
requires that the \emph{adjusted-$R^2$} for the cat-modified model exceeds that for the main-only model, or equivalently,  
$ (\Vert \bm{\hat e}_M \Vert^2  - \Vert \bm{\hat e} \Vert^2) /\Vert \bm{\hat e}_M \Vert^2 \ge  d/(n-d_M)$, 
so that the (guaranteed) reduction in sum-squared-residuals from main-only to cat-modified must be  large enough to justify the addition of $d$ parameters. This requirement is modest:  adjusted-$R^2$ is well-known to prefer overparametrized models, and thus \eqref{rv} is likely to hold even when the cat-modifiers are extraneous (see Section~\ref{sec-sims-est}). When cat-modifiers are indeed necessary, the reduction from $\hat S_M^2$ to $\hat S^2$ can be substantial.

We revisit each case from Section~\ref{sec-est-inv}, beginning with a two-way ANOVA (Example~\ref{ex:cat}). 
\begin{theorem}\label{thm-cat-cat-se}
    Under ABCs \eqref{abcs-gen} and \eqref{abcs-cat-cat-gen} and \eqref{rv}, the OLS SEs of \emph{all} main effects under the cat-modified model  \eqref{reg-cm-cat} are less than or equal to those under the main-only model \eqref{reg-main-cat}: $\mbox{SE}(\hat \alpha_0) \le \mbox{SE}(\hat \alpha_0^M)$, $\mbox{SE}(\hat \beta_{1,r}) \le \mbox{SE}(\hat \beta_{1,r}^M)$ for all $r=1,\ldots,L_R$, and $\mbox{SE}(\hat \beta_{2,s}) \le \mbox{SE}(\hat \beta_{2,s}^M)$ for all $s=1,\ldots,L_S$.
\end{theorem}
Remarkably, Theorems~\ref{thm-cat-cat}~and~\ref{thm-cat-cat-se} confirm that ABCs deliver the best possible result: adding cat-modifiers to the main-only model \eqref{reg-main-cat} does not change the main effect estimates, but potentially decreases their SEs. Thus, analysts may include cat-modifiers ``for free"---with no negative consequences for the main effects---while acquiring the ability to infer possibly heterogeneous, group-specific effects. The same occurs for categorical-continuous interactions, again with the equal-variance condition: 
\begin{theorem} \label{thm-se}
     Under ABCs \eqref{abcs-gen}, equal-variance \eqref{eq-v}, and \eqref{rv}, the OLS SE for the main $x$-effect under the cat-modified model \eqref{reg-cm-x} is less than or equal to that under the main-only model \eqref{reg-main-x}:  $\mbox{SE}(\hat \alpha_1) \le  \mbox{SE}(\hat \alpha_1^M)$. 
\end{theorem}
This result applies in the context of Theorem~\ref{thm-cts}, but analogous extensions are available for Theorems~\ref{thm-int-full}~and~\ref{thm-int-local}; only the condition \eqref{rv} must be added.

These results make minimal assumptions about the true data-generating process and do \emph{not} require independence or uncorrelatedness among the covariates. However, the OLS SEs are defined as usual, which implicitly refers to uncorrelated and homoskedastic error assumptions for both the main-only and cat-modified models. Thus, while Theorems~\ref{thm-cat-cat-se}~and~\ref{thm-se} are direct statements about the SEs as statistics and do not require any assumptions on the error distributions, the utility of these results is clearly linked to these assumptions.

\section{Simulations}\label{sec-sims}


\subsection{Validating invariance for estimation and inference}
The first objective is to verify the theory of ABCs for estimation and inference invariance,  focusing on the conditions in Section~\ref{sec-inv}. We consider both categorical-categorical interactions (Section~\ref{sims-cat-cat}) and categorical-continuous interactions (Section~\ref{sims-cts-cat}). 


\subsubsection{Categorical-categorical interactions}\label{sims-cat-cat}
Given two categorical variables, say \texttt{race} and \texttt{sex}, what is the effect of including the \texttt{race:sex} interaction term on the estimates and SEs for the \texttt{race} and \texttt{sex} \emph{main} effects? The theory of ABCs (Section~\ref{sec-inv}) predicts that the estimates will be exactly the same, while the SEs may decrease if the interaction effect is sufficiently large. These results make no requirements on the data-generating process. Thus, we simulate data such that  1) \texttt{race} and \texttt{sex} are dependent, 2) the errors are non-Gaussian, and 3)  ABCs are not satisfied. 

Let \texttt{race} and \texttt{sex} be categorical variables with groups $\{\texttt{A}, \texttt{B}, \texttt{C}, \texttt{D}\}$ and $\{\texttt{uu}, \texttt{vv}\}$, respectively; we use arbitrary labeling here to remain agnostic about particular race- or sex-specific effects in our synthetic data-generating process. For each of $n=500$ observations, we  draw each \texttt{race} assignment with $( \pi_a,  \pi_b,  \pi_c,  \pi_d) = (0.4, 0.3, 0.2, 0.1)$, and then draw the \texttt{sex} assignment conditional on \texttt{race} with 
 $( \pi_{uu},  \pi_{vv})_{\cdot \mid r = \texttt{A}} = (0.4, 0.6)$, 
 $( \pi_{uu},  \pi_{vv})_{\cdot \mid r = \texttt{B}} = (0.6, 0.4)$, 
 $( \pi_{uu},  \pi_{vv})_{\cdot \mid r = \texttt{C}} = (0.7, 0.3)$,  and 
 $( \pi_{uu},  \pi_{vv})_{\cdot \mid r = \texttt{D}} = (0.2, 0.8)$. Thus, \texttt{race} and \texttt{sex} are dependent, and marginally $\pi_{uu} = \pi_{vv} = 0.5$.  The response variable $y$ is simulated with expectation \eqref{reg-cm-cat} with $\alpha_0 = 1$, $\beta_c = -1$,  $\gamma_{b,vv} = \gamma$, and all other coefficients zero, or equivalently, 
$
\mu(r,s) = 1   - \mathbb{I}\{r = \texttt{C}\} +  \gamma \mathbb{I}\{r = \texttt{B}, s = \texttt{vv}\} 
$ 
plus $t_4(0,1)$-distributed errors. 
Crucially, $\gamma$ controls the  magnitude of the  \texttt{race:sex} effect: we consider $\gamma = 0$ (no interaction effect), $\gamma = 0.5$ (moderate interaction effect; see the supplementary material),  and $\gamma = 1.5$ (large interaction effect).  
We repeat this process to create 500 synthetic datasets.

For each simulated dataset, we fit the main-only model \eqref{reg-main-cat} 
and the cat-modified model \eqref{reg-cm-cat} 
and compare the estimates and SEs for each main effect between the two models. These models are fit using ABCs, RGE (references  $r = \texttt{A}$, $s = \texttt{uu}$), and STZ. This setting is favorable for RGE: the data-generating process satisfies RGE ($\beta_a = 0, \beta_{uu} = 0, \gamma_{a,uu} = 0$), but not ABCs, and the reference groups are the most abundant groups for both \texttt{race} and \texttt{sex}. To aid comparisons, we omit the main effects from the RGE reference groups, resulting in four main effects ($\beta_b, \beta_c, \beta_d, \beta_{vv}$) to compare between the main-only and cat-modified models for ABCs, RGE, and STZ. 

The estimates are in Figure~\ref{fig:sim-race-sex-est}. Under ABCs, \emph{all} \texttt{race} and \texttt{sex} main effects are \emph{exactly} identical between the models that do and do not include the \texttt{race:sex} interaction, confirming Theorem~\ref{thm-cat-cat}. This result persists regardless of the true data-generating process, including the magnitude of the interaction. No such invariance occurs for RGE or STZ: the inclusion of the interaction completely changes the estimates (and the interpretations) of the main effects. 

\begin{figure}
\centering
\includegraphics[width=.49\linewidth]{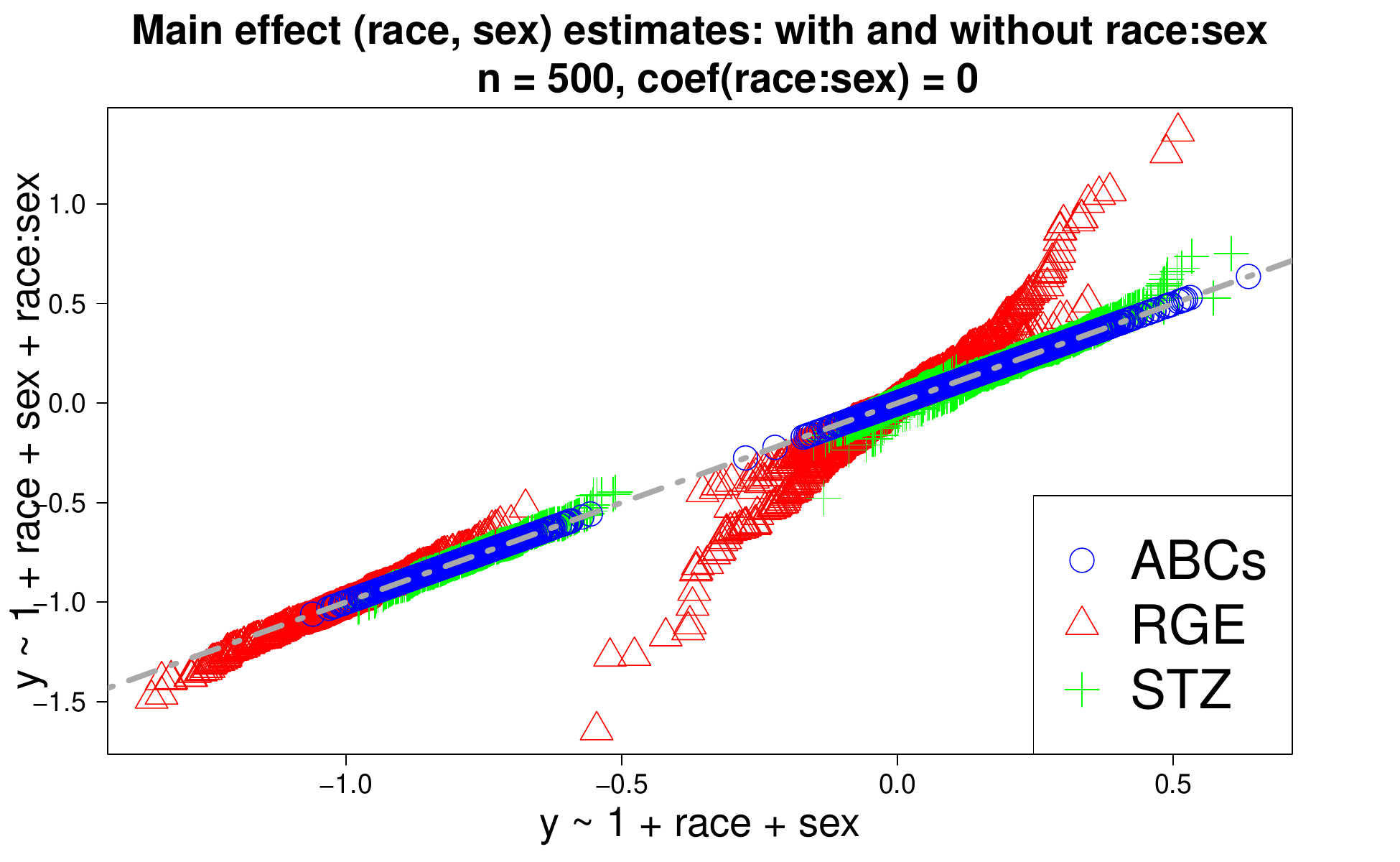}
\includegraphics[width=.49\linewidth]{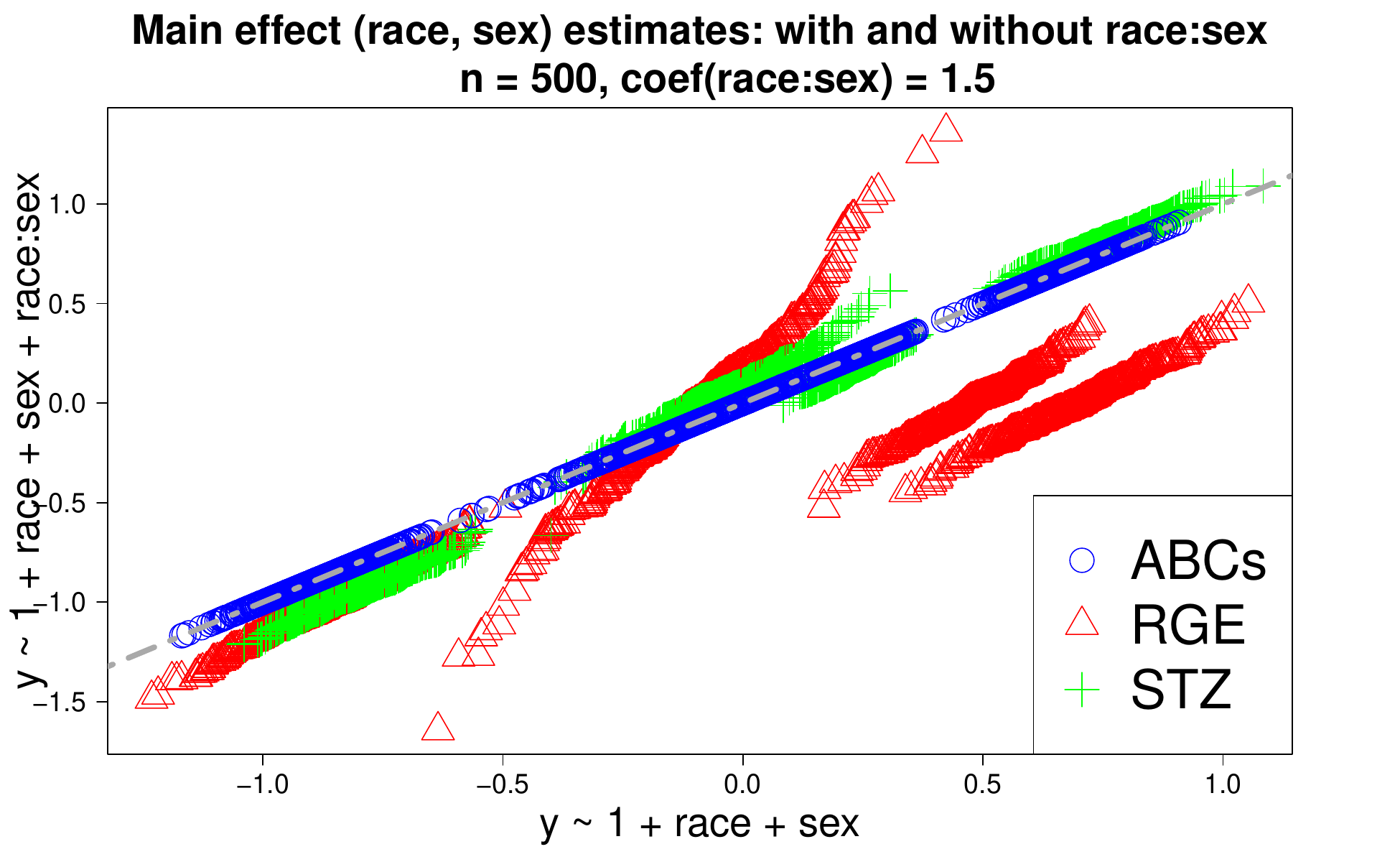}
\caption{\small Estimates for all \texttt{race} and \texttt{sex} main effects 
for models that do (y-axis) and do not (x-axis) include the \texttt{race:sex} interaction across 500 simulated datasets. Under ABCs, all main effect estimates are \emph{exactly} identical between the two models ($45^\circ$ line), 
regardless of whether the interaction effect is zero ($\gamma = 0$, left) or large ($\gamma = 1.5$, right). Such invariance does not hold for other identifications (RGE or STZ). 
}
\label{fig:sim-race-sex-est}
\end{figure}

The SEs are in 
Figure~\ref{fig:sim-race-sex-se}. 
Under ABCs,  the addition of the \texttt{race:sex} interaction  has virtually no impact on the main effect SEs. For larger $\gamma$, the main effect SEs are slightly smaller (about a 5\% reduction) for cat-modified model, as expected. More substantial SE reductions occur for larger interactions ($\gamma \ge 5$), although such large interaction effects are not usually expected in practice. 
Again, no such results occur for RGE: the SEs are much larger for the model that includes  the \texttt{race:sex} interaction, regardless of $\gamma$.

These results must be interpreted carefully: the ``main effects" under ABCs, RGE, or STZ target different functionals of $\mu(r,s)$. In fact, the OLS fitted values for $\hat \mu(r,s)$ are identical under each identification (this is not the case for regularized regression). However, each identification puts forth ``main effects" in both the main-only and cat-modified model. We argue that the main effects under ABCs are superior: the estimates are \emph{exactly} invariant to the inclusion of (\texttt{race:sex}) interactions and the SEs may decrease slightly. Uniquely, ABCs circumvent the traditional roadblocks to including interactions: the interpretations remain simple (and equitable) and there is no loss of statistical power for the main effects.

\begin{figure}
\centering
\includegraphics[width=.49\linewidth]{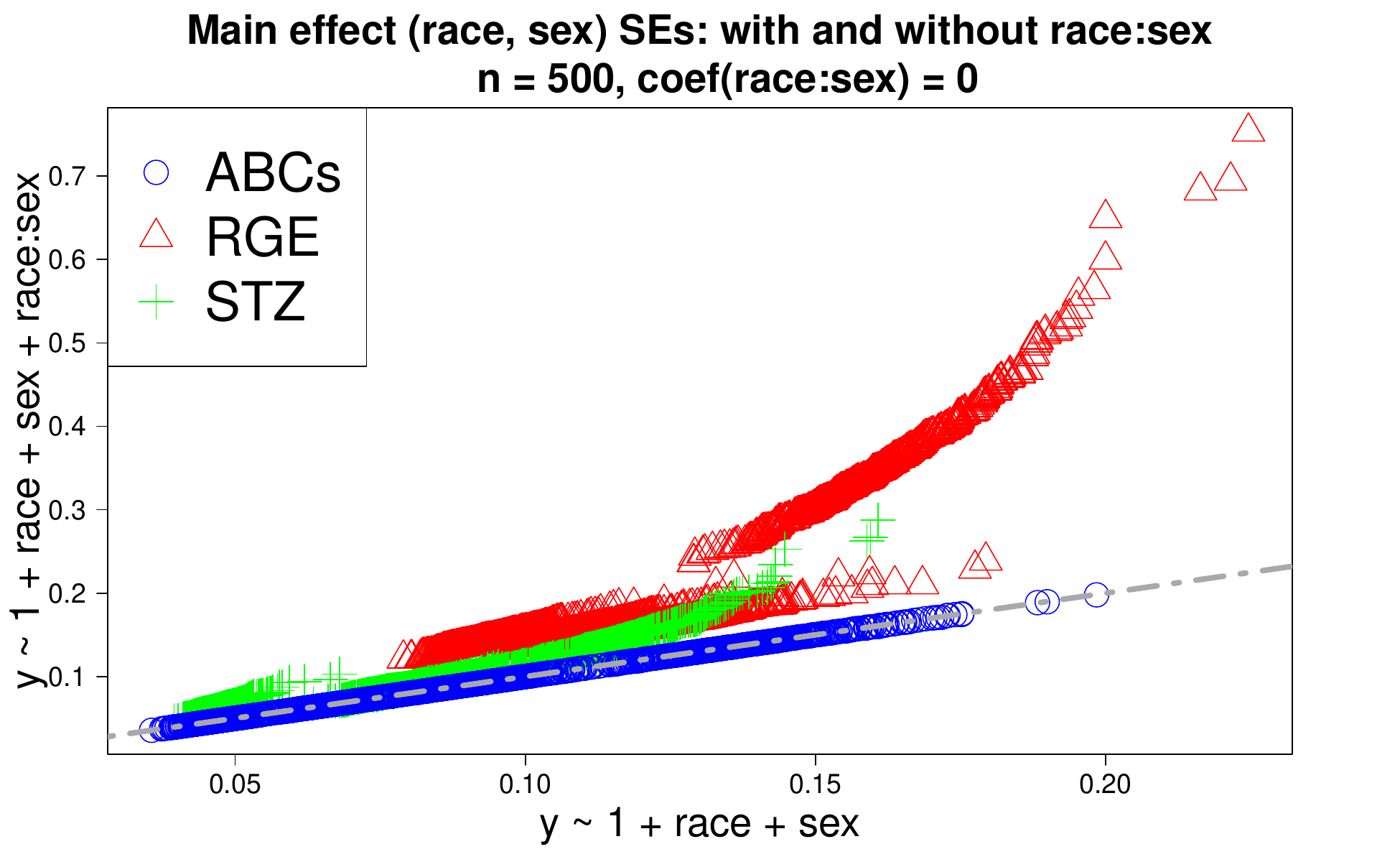}
\includegraphics[width=.49\linewidth]{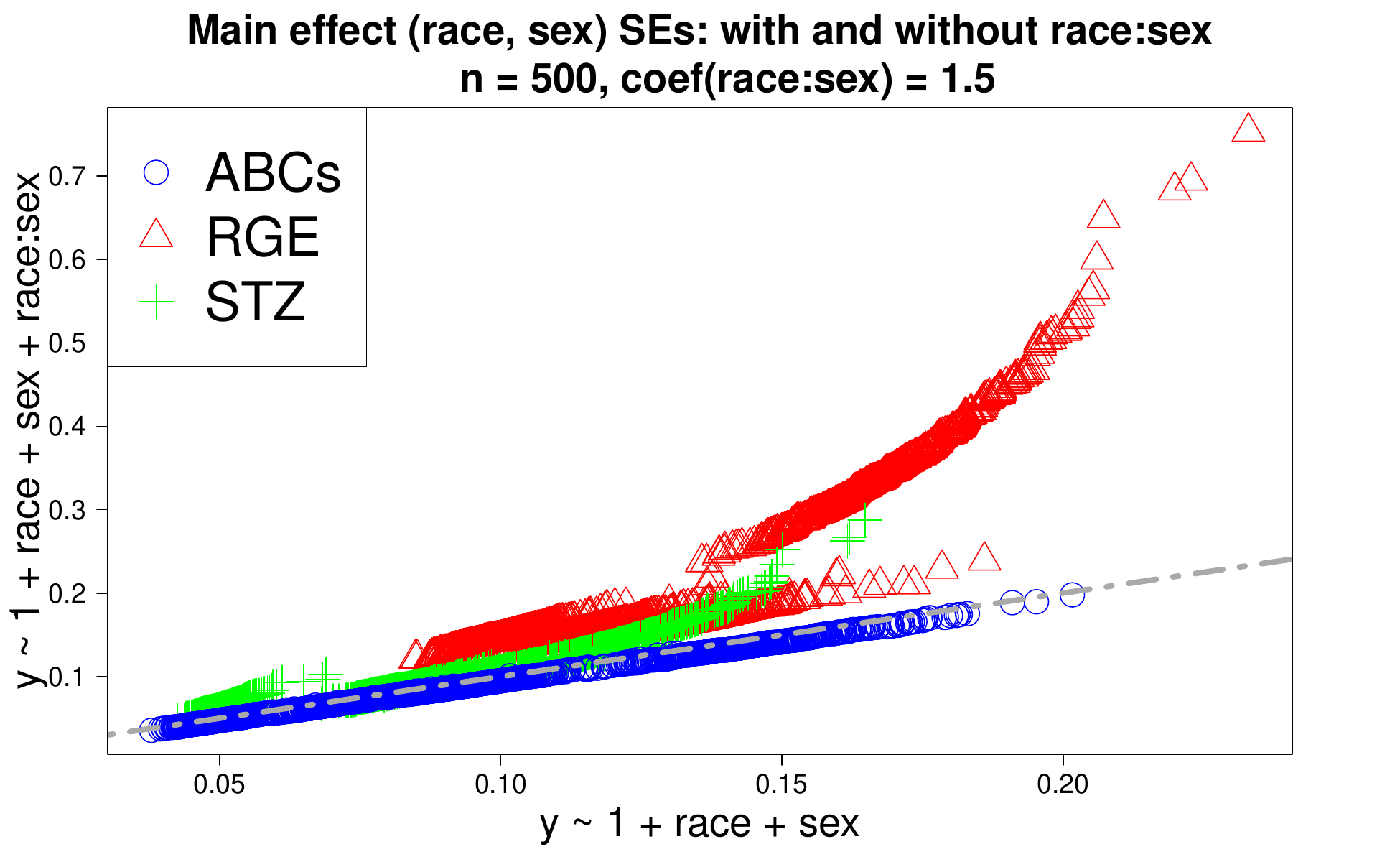}
\caption{\small Standard errors (SEs) for all \texttt{race} and \texttt{sex} main effects 
for models that do (y-axis) and do not (x-axis) include the \texttt{race:sex} interaction across 500 simulated datasets. Under ABCs, the SEs are nearly identical between the two models ($45^\circ$ line) when the interaction effect is 
zero ($\gamma = 0$, left) and slightly less (about a 5\% reduction) for the cat-modified model when the interaction effect is larger  ($\gamma = 1.5$, right). The RGE and STZ main effect SEs increase substantially when the interaction term is included in the model (above $45^\circ$ line) regardless of $\gamma$.  
}
\label{fig:sim-race-sex-se}
\end{figure}

\subsubsection{Categorical-continuous interactions}\label{sims-cts-cat}
We now revise this analysis for categorical-continuous interactions: given categorical \texttt{race} and continuous \texttt{x}, what is the effect of including the \texttt{x:race} interaction on the main $x$-effect? The theory of ABCs (Section~\ref{sec-inv}) predicts that invariance for estimation and inference is contingent on  the equal-variance condition \eqref{eq-v}. We investigate the sensitivity to this condition as well as to the magnitude of the interaction effect. 

To incorporate dependencies between \texttt{race} and \texttt{x}, we simulate \texttt{race} as in Section~\ref{sims-cat-cat} and then simulate \texttt{x} conditional on \texttt{race}: 
$[\texttt{x} \mid \texttt{race} = \texttt{A}] \sim 5  + \sigma_{ac} N(0,1)$, 
$[\texttt{x} \mid \texttt{race} = \texttt{B}] \sim \sqrt{12}\ \mbox{Uniform}(0,1)$, 
$[\texttt{x} \mid \texttt{race} = \texttt{C}] \sim -5 + \sigma_{ac}t_8(0,1)$, and 
$[\texttt{x} \mid \texttt{race} = \texttt{D}] \sim  \mbox{Gamma}(1,1)$. 
Each \texttt{race} group features a unique distribution with varying means, so \texttt{x} and \texttt{race} are strongly dependent and highly correlated. 
Here, $\sigma_{ac}$ controls the degree to which the equal-variance condition \eqref{eq-v} is violated: $\sigma_{ac} = 1$ is a mild violation (the race-specific \emph{population} variances are identical, but the sample quantities $\hat \sigma_{x[r]}^2$ are not) while $\sigma_{ac} = 1.5$ is a strong violation. The response variable $y$ is simulated with expectation \eqref{reg-cm-x} with $\alpha_0 = \alpha_1 = 1$, $\beta_c = -1$, and $\gamma_b = \gamma$, and all other coefficients zero, or equivalently, 
$ 
\mu(x, r) = 1 + x  - \mathbb{I}\{r = \texttt{C}\} + \gamma x \mathbb{I}\{r = \texttt{B}\}
$ 
plus $t_4(0,1)$-distributed errors. 
This data-generating process satisfies RGE ($\beta_a = 0$), but not ABCs, and includes non-Gaussian errors. Again, $\gamma \in\{0,0.5,1.5\}$ determines the magnitude of the interaction effect.  We repeat this process to create 500 synthetic datasets. 

For each simulated dataset, we fit the main-only model \eqref{reg-main-x}  
and the cat-modified model \eqref{reg-cm-x} 
and compare the estimates and SEs for main $x$-effect $\alpha_1$ between the two models under ABCs, RGE, and STZ. 
The estimates are in Figure~\ref{fig:sim-x-race-est}. Even with mild deviations from the equal-variance condition \eqref{eq-v}, the $x$-effect estimates under ABCs are nearly identical between models that do and do not include the \texttt{x:race} interaction. Crucially, this invariance persists regardless of the true interaction effect magnitude $\gamma$. Under strong violations of \eqref{eq-v} \emph{and} a strong interaction effect, Theorem~\ref{thm-cts} no longer applies. However, as argued in Section~\ref{sec-inv}, this behavior is appropriate: when \eqref{eq-v}  is strongly violated,  a one-unit change in $x$ is not comparable for different \texttt{race} groups, so only the model that includes race-specific $x$-effects (via the \texttt{x:race} interaction) is appropriate. Finally, we note the absence of invariance for estimation with RGE or STZ. These estimators change dramatically when $\gamma$ is moderate to large. Even when $\gamma = 0$---when classical consistency results for OLS should provide asymptotic invariance in this case---they do not match the invariance of ABCs.

\begin{figure}
\centering
\includegraphics[width=.49\linewidth]{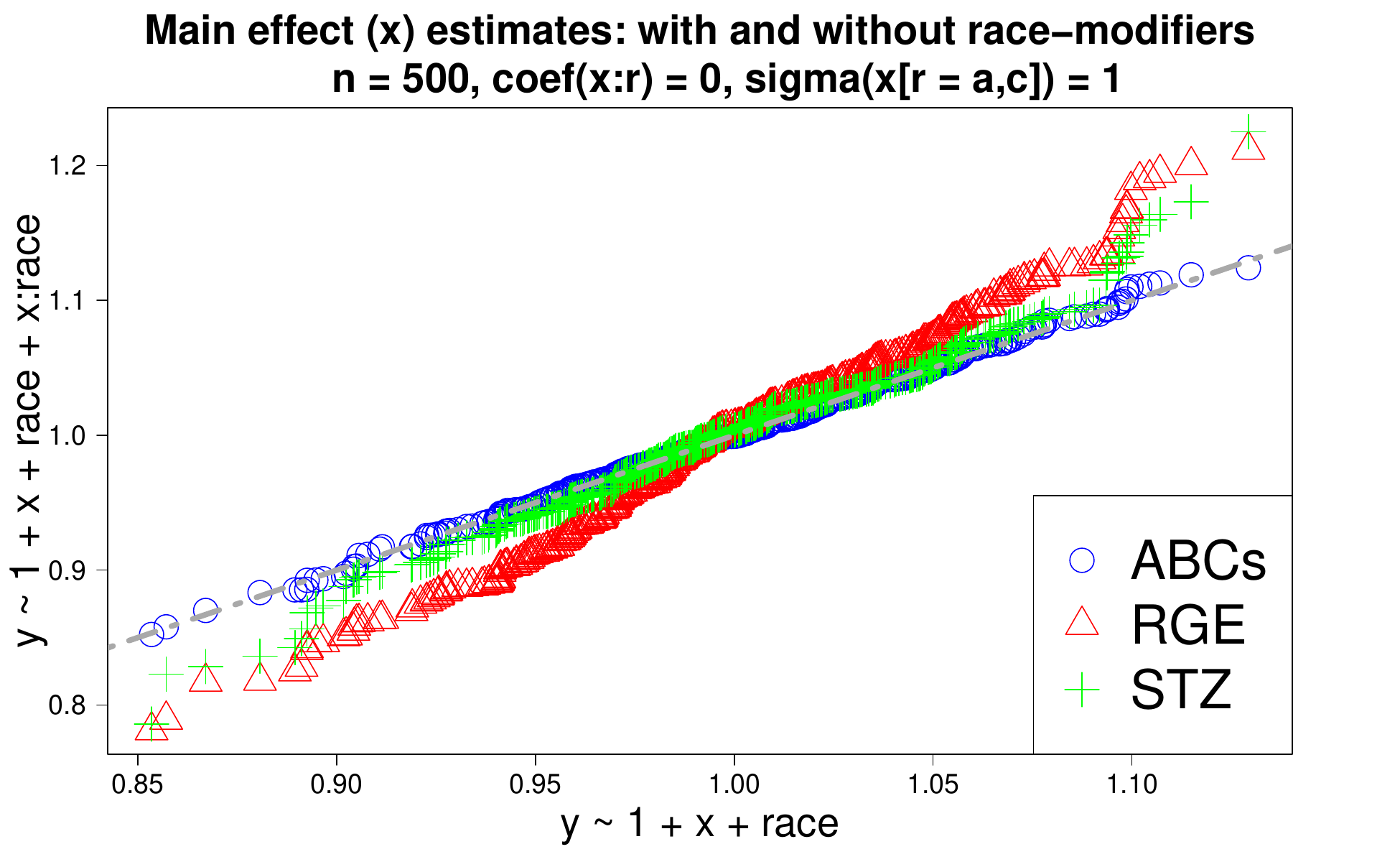}
\includegraphics[width=.49\linewidth]
{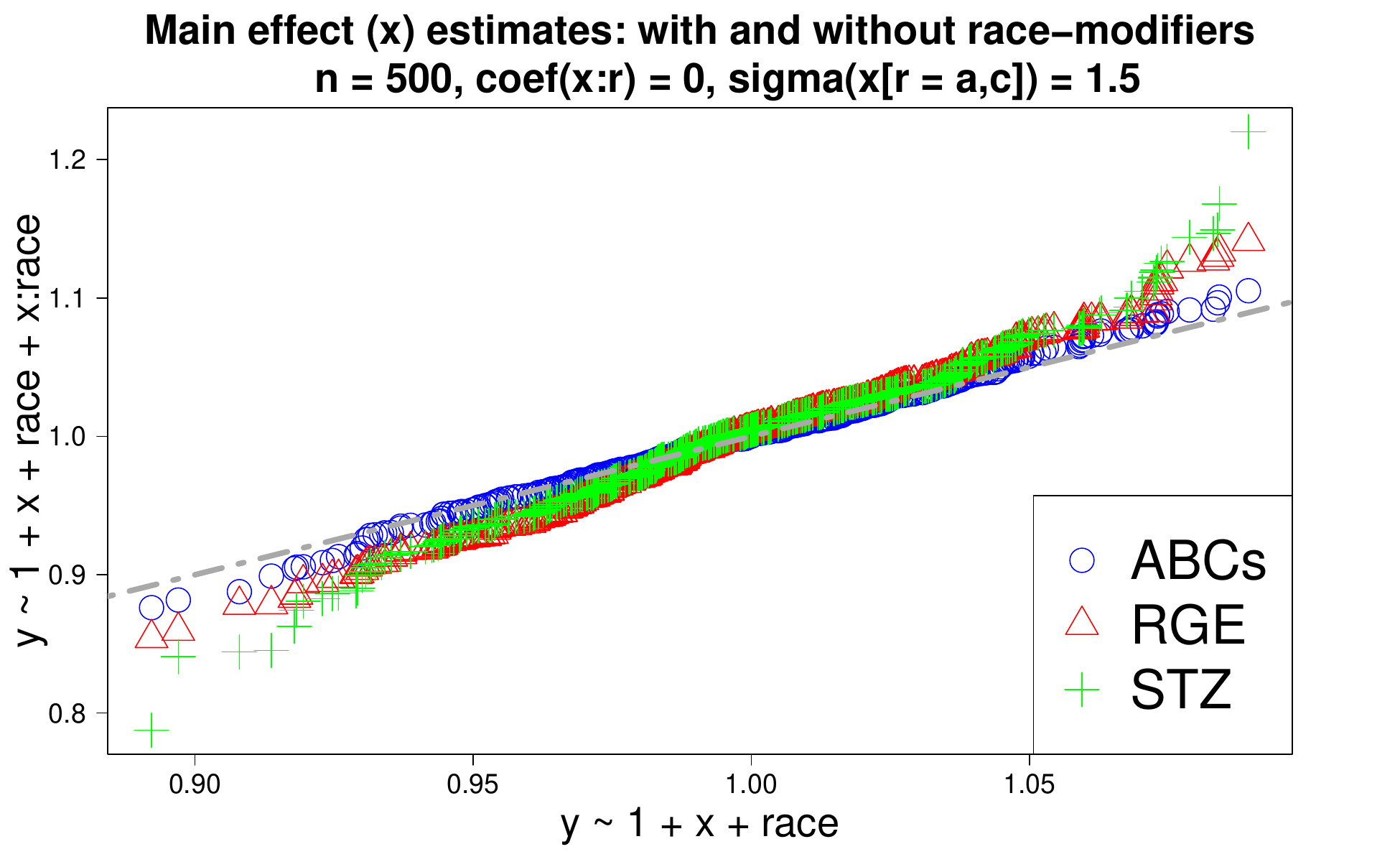}

\includegraphics[width=.49\linewidth]{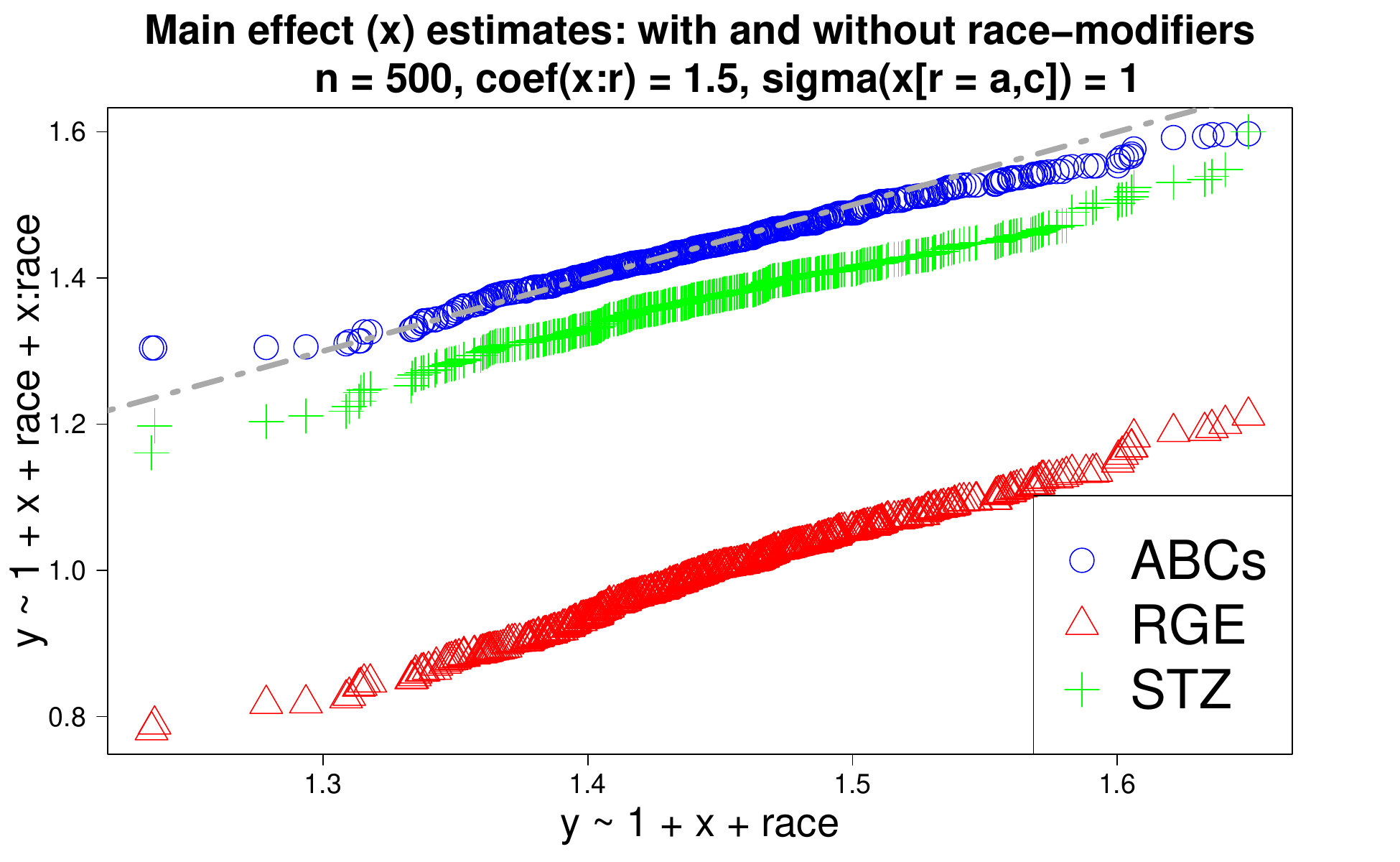}
\includegraphics[width=.49\linewidth]{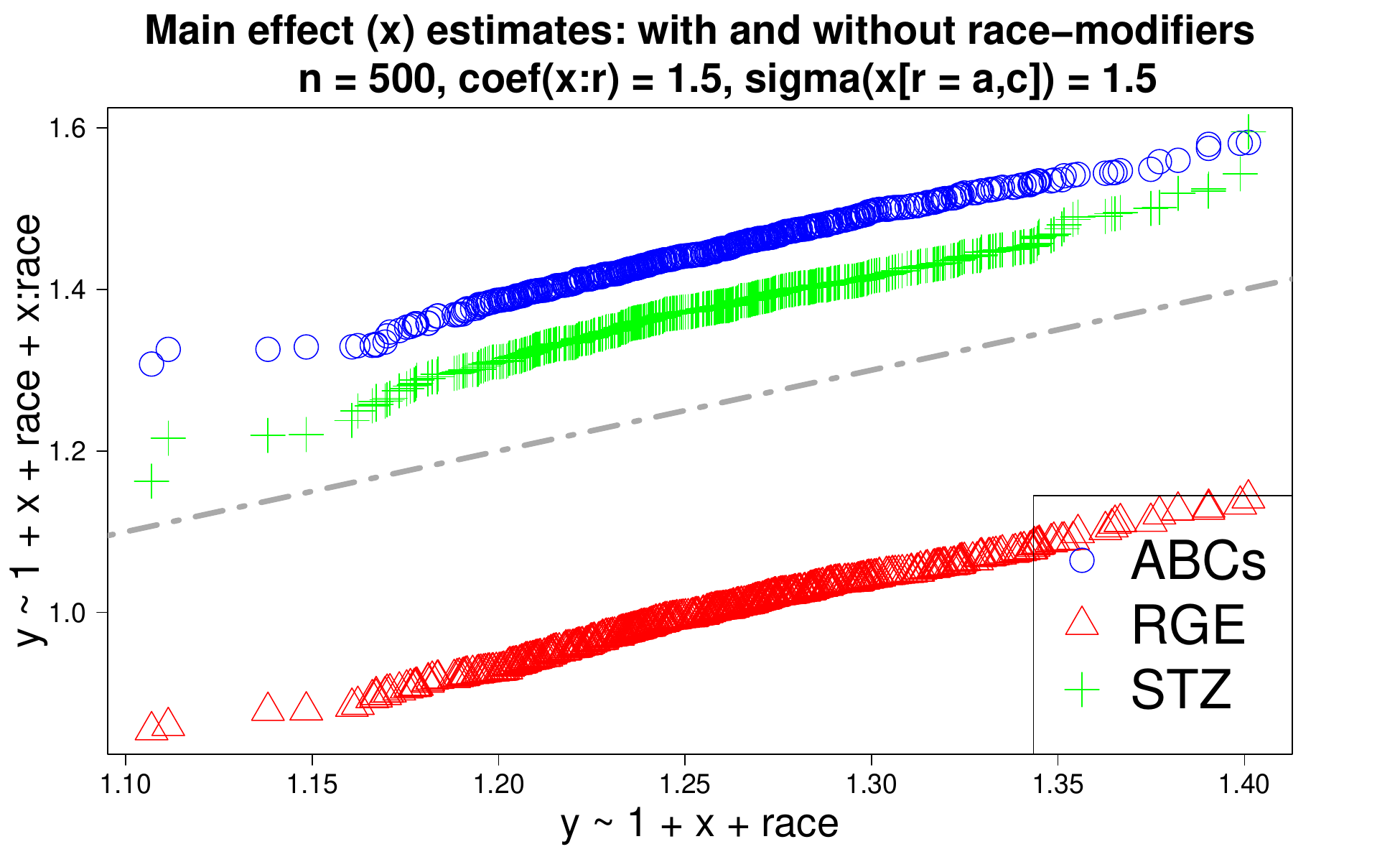}
\caption{\small Estimates for the main $x$-effect 
for models that do (y-axis) and do not (x-axis) include the \texttt{x:race} interaction across 500 simulated datasets. Under ABCs, the estimates are nearly invariant ($45^\circ$ line) as long as the deviations from equal-variance \eqref{eq-v} are  mild ($\sigma_{ac} = 1$, left), regardless of whether the true interaction effect is zero ($\gamma = 0$, top) or large ($\gamma = 1.5$, bottom). When $\gamma$ is large \emph{and} \eqref{eq-v} is strongly violated (bottom right), ABCs no longer offer invariance under Theorem~\ref{thm-cts}. RGE and STZ offer no such invariance and depend critically on $\gamma$.  
}
\label{fig:sim-x-race-est}
\end{figure}

The SEs are in Figure~\ref{fig:sim-x-race-se}. As long as the violations of \eqref{eq-v} are mild ($\sigma_{ac} = 1$), the SEs of the $x$-effect, under ABCs, are 1) nearly identical between the main-only and cat-modified models when the true interaction effect is small  and 2) smaller for the cat-modified model  when the true interaction effect is large. Critically, including the \texttt{x:race} interaction \emph{under ABCs} does not sacrifice any statistical power for the main $x$-effect, and in some cases enhances it. This is decisively not the case for RGE or STZ: when the true interaction effect is zero, adding the \texttt{x:race} interaction decreases statistical power for the main $x$-effect. 

The same caveats apply as in Section~\ref{sims-cat-cat}: RGE, STZ, and ABCs are targeting different functionals of   $\mu(x, r)$, but again we argue that the estimation and inference properties of the ``main effects" are most ideal under ABCs. 




\begin{figure}
\centering
\includegraphics[width=.49\linewidth]{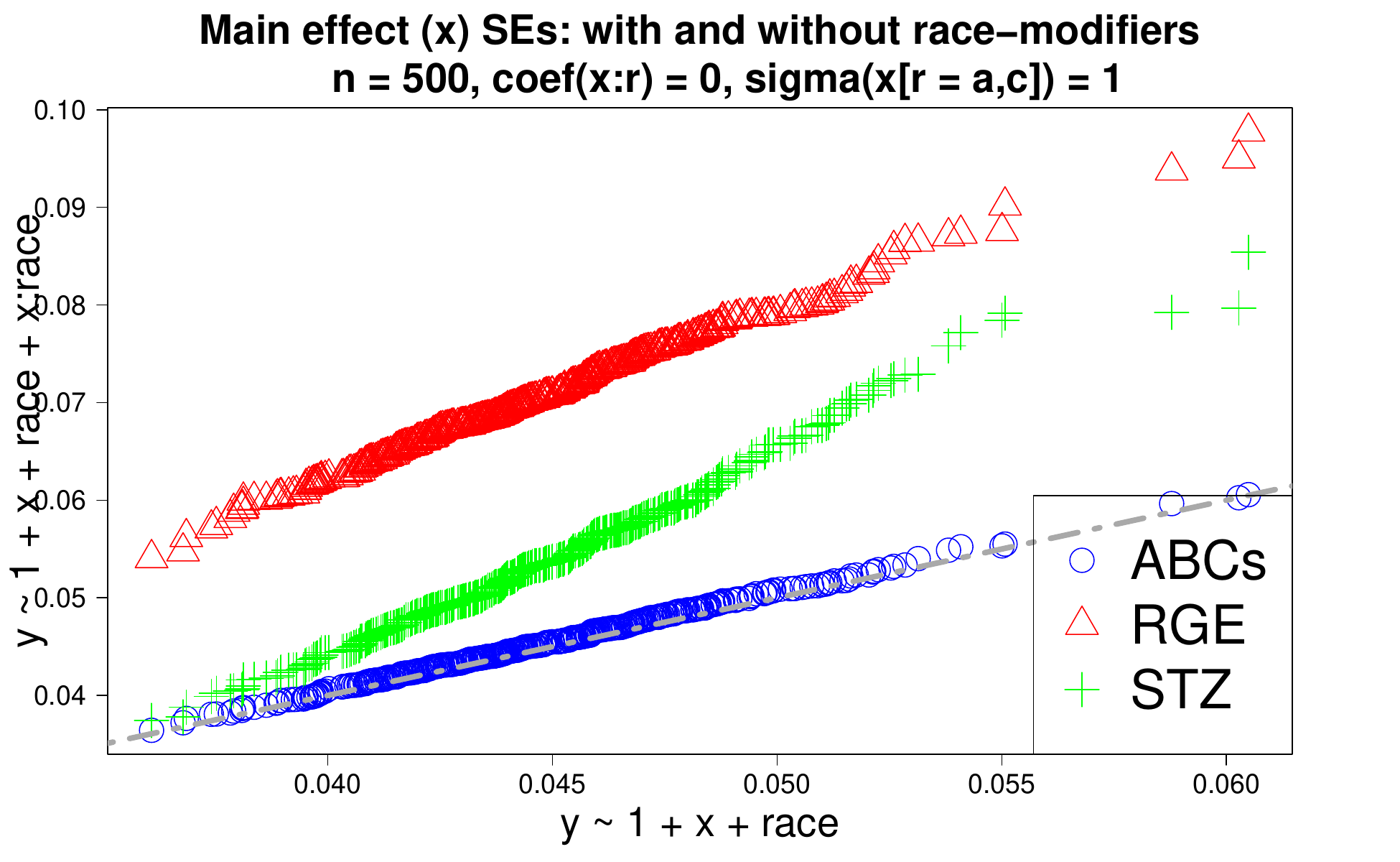}
\includegraphics[width=.49\linewidth]
{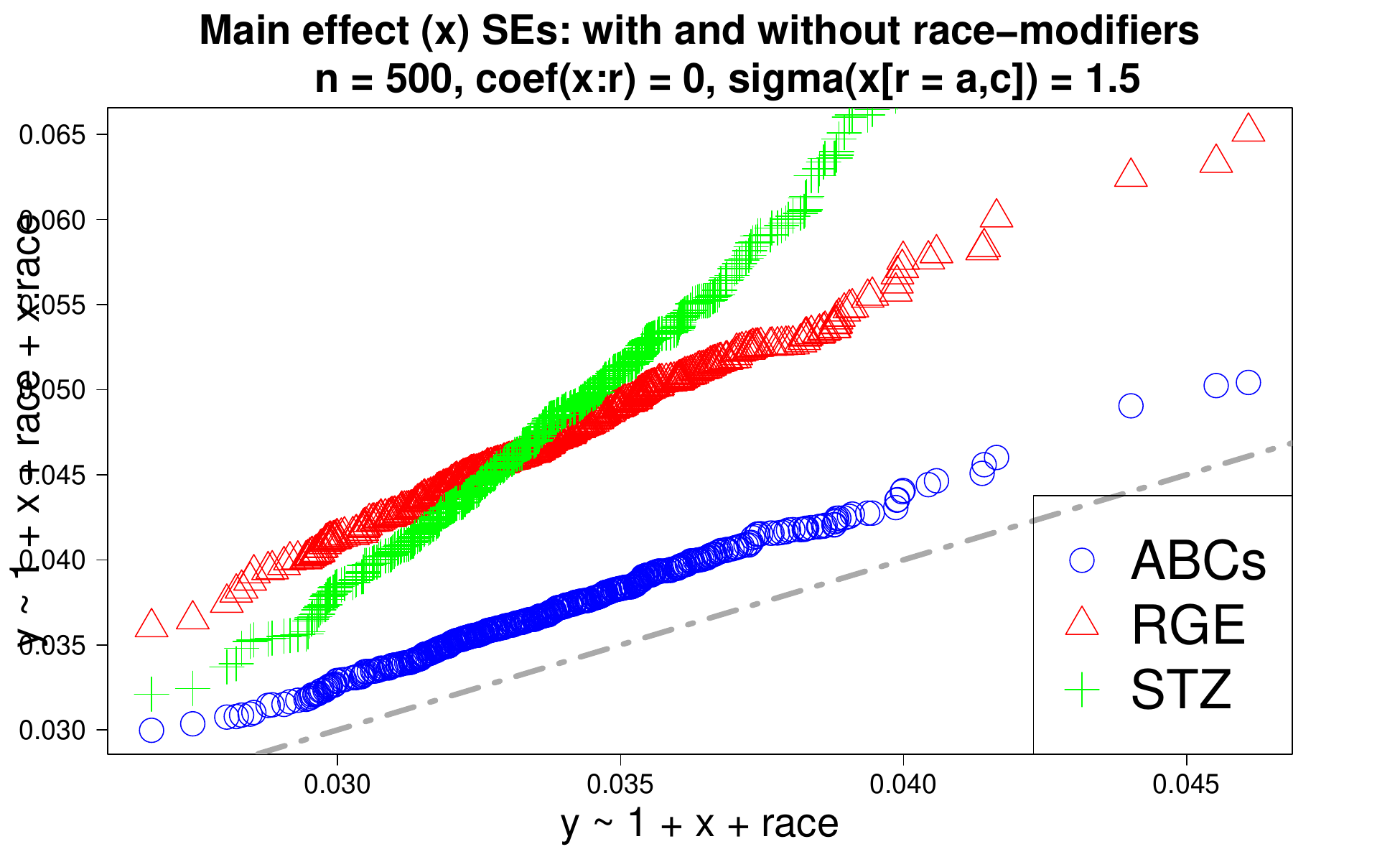}

\includegraphics[width=.49\linewidth]{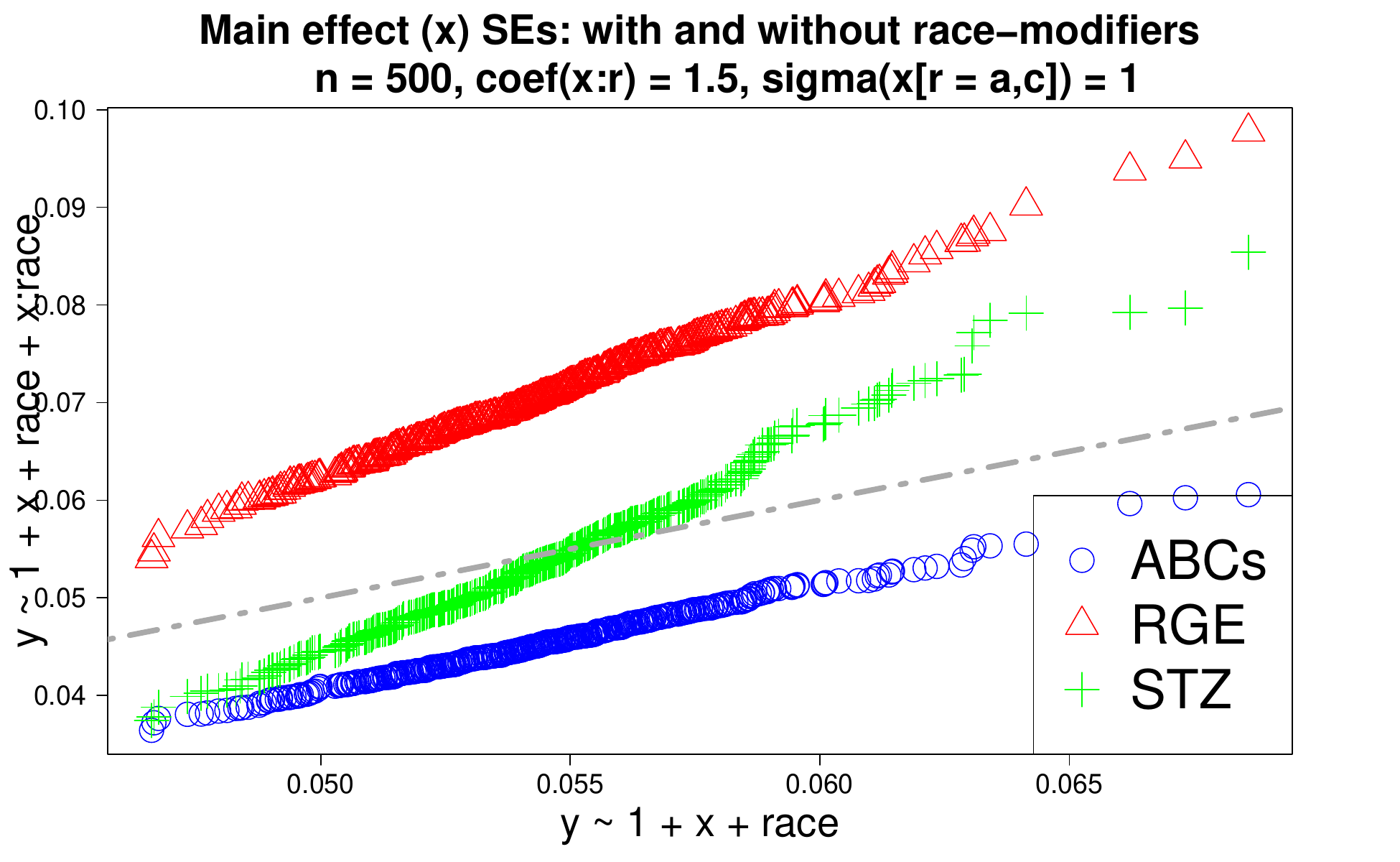}
\includegraphics[width=.49\linewidth]{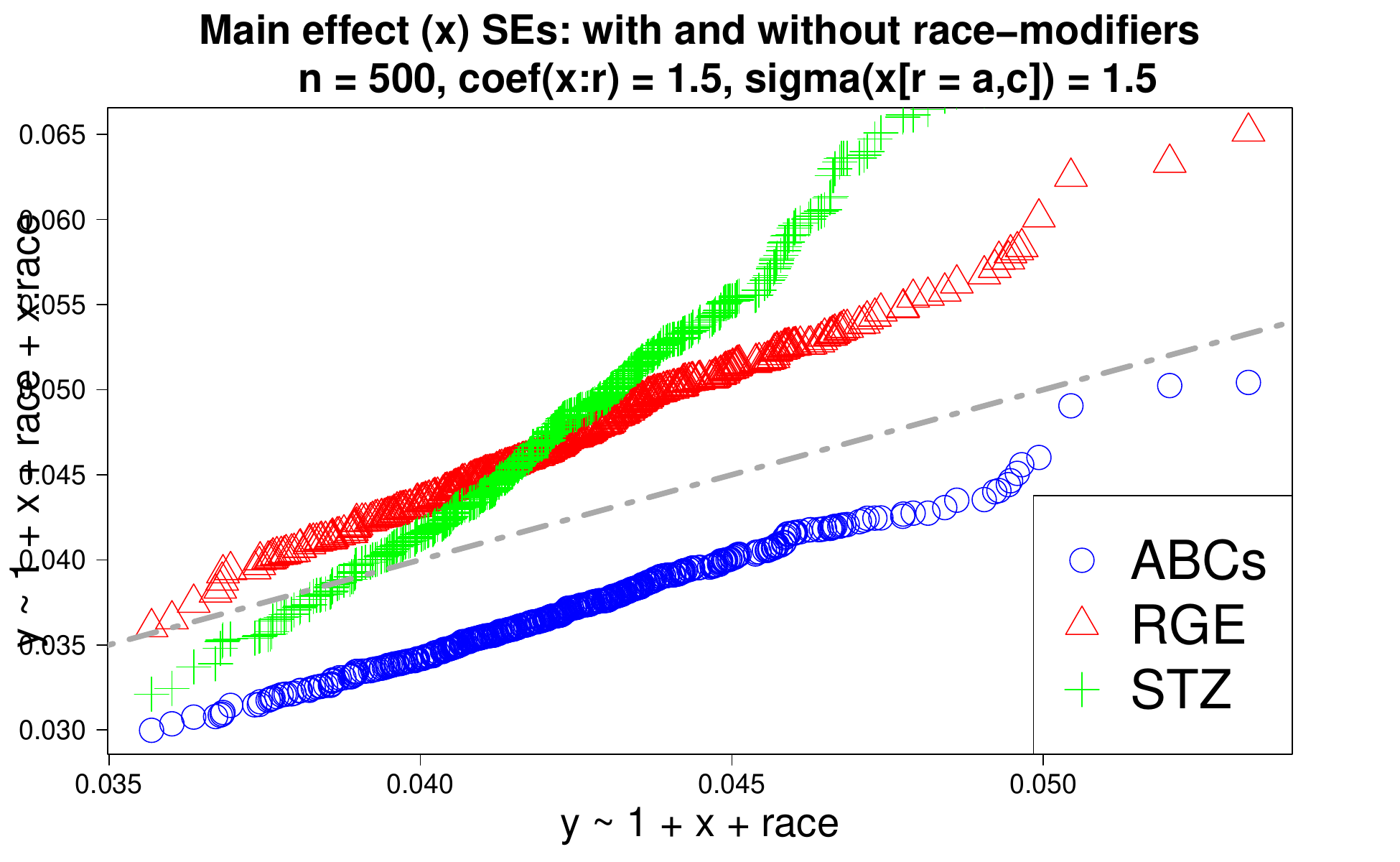}
\caption{\small Standard errors (SEs) for the main $x$-effect for models that do (y-axis) and do not (x-axis) include the \texttt{x:race} interaction across 500 simulated datasets. Under ABCs, the SEs are nearly identical between the two models ($45^\circ$ line) when the  true interaction effect is zero \emph{and} deviations from equal-variance \eqref{eq-v} are  mild   ($\gamma = 0$, $\sigma_{ac} = 1$, top left). If instead the interaction effect is large ($\gamma =1.5$, $\sigma_{ac} = 1$, bottom left), the SEs under ABCs reduce substantially (by about 15\%) for the model that includes the \texttt{x:race} interaction. These effects are not assured when \eqref{eq-v} is strongly violated ($\sigma_{ac} = 1.5$, right). All results are consistent with Theorem~\ref{thm-se}. Similar properties do \emph{not} occur for RGE or STZ, regardless of $\gamma$ and $\sigma_{ac}$.
}
\label{fig:sim-x-race-se}
\end{figure}


\subsection{Evaluating estimation and inference with cat-modifiers}\label{sec-sims-est}
We evaluate the practical impacts of the estimation invariance and enhanced power of ABCs. The goal is quantify the extent to which ABCs 1) maintain accurate estimates and precise uncertainty quantification when \emph{extraneous} cat-modifiers are included and 2) improve estimation and reduce uncertainty when \emph{necessary} cat-modifiers are included. The simulation design has three main features, described below.

First, we generate multiple, dependent categorical and continuous covariates. Dependent categorical variables \texttt{race} and \texttt{sex} are generated as in  Section~\ref{sims-cat-cat}, while $p=10$ dependent continuous variables 
are generated as follows: $x_j$ is drawn as in Section~\ref{sims-cts-cat} with $\sigma_{ac}=1$ for $j=1,3,5,7,9$ and  $N(0,1)$ for $j=2,4,6,8,10$. Some $x$-variables are correlated with \texttt{race}, which induces correlations among those $x$-variables with each other and with \texttt{sex}, while others are uncorrelated.  

Second, the regression coefficients are constructed  to satisfy both RGE and ABCs. 
These include an intercept $\alpha_0=1$,  active main $x$-effects $\alpha_j = 1$ for $j=1,\ldots,5$, \texttt{race} main effects $\beta_b = 1$ and $\beta_c = \beta_d = -1$, and cat-modifiers $\gamma_{b,j} = \gamma$ and $\gamma_{c,j} = \gamma_{d,j} = -\gamma$ for $j=1,\ldots,5$; the remaining coefficients are all zero. RGE is enforced because all reference group coefficients are zero  ($\beta_a  = 0$, $\beta_{uu} = 0$, $\gamma_{a,uu} = 0$, and $\gamma_{a, j} = 0 $ for all $j=1,\ldots,p$), while ABCs are satisfied for the \emph{population} proportions $( \pi_a,  \pi_b,  \pi_c,  \pi_d) = (0.4, 0.3, 0.2, 0.1)$. Thus, it is meaningful to compare coefficient estimates and inference between RGE and ABCs (STZ is not satisfied and thus excluded). ABCs actually use the \emph{sample} proportions and are at a slight disadvantage. All \texttt{sex} main and interaction effects are zero, since this is the only way to satisfy both RGE and ABCs for a variable with two groups.

Third, a parameter $\gamma$ controls the magnitude of the cat-modifier (\texttt{x:race}) effect. We consider $\gamma = 0$ for \emph{extraneous} cat-modifiers and $\gamma = 1.5$ for \emph{necessary} cat-modifiers.  

Using these covariates and coefficient values, the response variable $y$ is simulated with expectation \eqref{reg-cm}  plus Gaussian errors and a signal-to-noise ratio of one. We vary the sample size $n \in \{200, 500, 1000\}$  and repeat this process to create 500 synthetic datasets. 

For each synthetic dataset, we fit the main-only  model \texttt{y $\sim$ x$_1$ + \ldots + x$_p$ + sex + race} and the cat-modified model \texttt{y $\sim$ (x$_1$ + \ldots + x$_p$ + sex)*race} that includes \texttt{race} interactions with all continuous covariates and \texttt{sex}, both under ABCs and RGE. The main-only model is favored when $\gamma = 0$, while the opposite is true for $\gamma > 0$. Either way, the true data-generating process is sparse, so both models include many extraneous parameters. The main-only model includes 15 identifiable parameters (9 true signals)  while the cat-modified model includes 48 identifiable parameters. When $\gamma = 0$, the cat-modified model estimates 33 extraneous (identified) parameters; even when $\gamma > 0$, only 15 of those cat-modifier effects are nonzero.

Evaluation primarily focuses on the  main $x$-effects  ($\alpha_1,\ldots,\alpha_{10}$), which isolates the impacts of including  extraneous ($\gamma = 0$) or necessary ($\gamma > 0$) cat-modifiers on estimation and inference for the main effects. For benchmarking, we also include evaluations for all (main and cat-modifier) coefficients. Note that for the main-only model, RGE and ABCs produce main $x$-effect estimates and SEs that are identical and nearly identical, respectively. 

Estimation accuracy is evaluated by root mean squared error (RMSE) for the regression coefficients (Figure~\ref{fig:sim-est}). The cat-modified model \emph{with ABCs} preserves estimation accuracy of the main $x$-effects, even when (all 33) cat-modifier effects are included unnecessarily (Figure~\ref{fig:sim-est}, top left). When \emph{some} cat-modifiers are necessary, the cat-modified model with ABCs delivers slightly more accurate main $x$-effect estimates than the main-only models (Figure~\ref{fig:sim-est}, bottom left). Neither result holds for RGE. By comparison, estimation accuracy across all coefficients 
(Figure~\ref{fig:sim-est}, right) overwhelmingly favors the correctly-specified model (main-only for $\gamma=0$, cat-modified for $\gamma = 1.5$), regardless of RGE or ABCs. This result is not surprising, but rather serves as a contrast to emphasize the extraordinary robustness of \emph{main} effect estimation accuracy for cat-modified models---but only under ABCs.  

\begin{figure}[h]
\centering
\includegraphics[width=.49\linewidth]{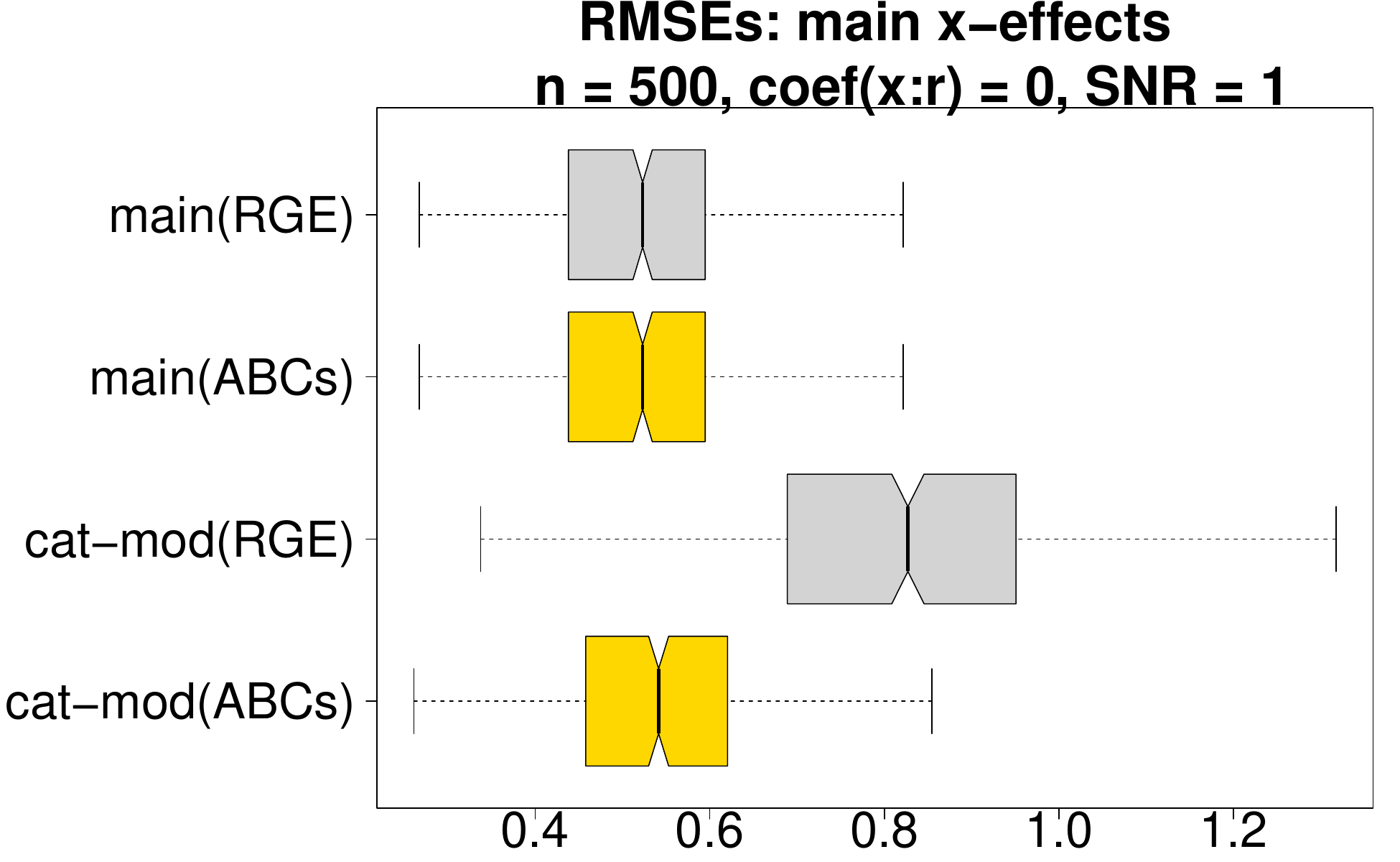}
\includegraphics[width=.49\linewidth]{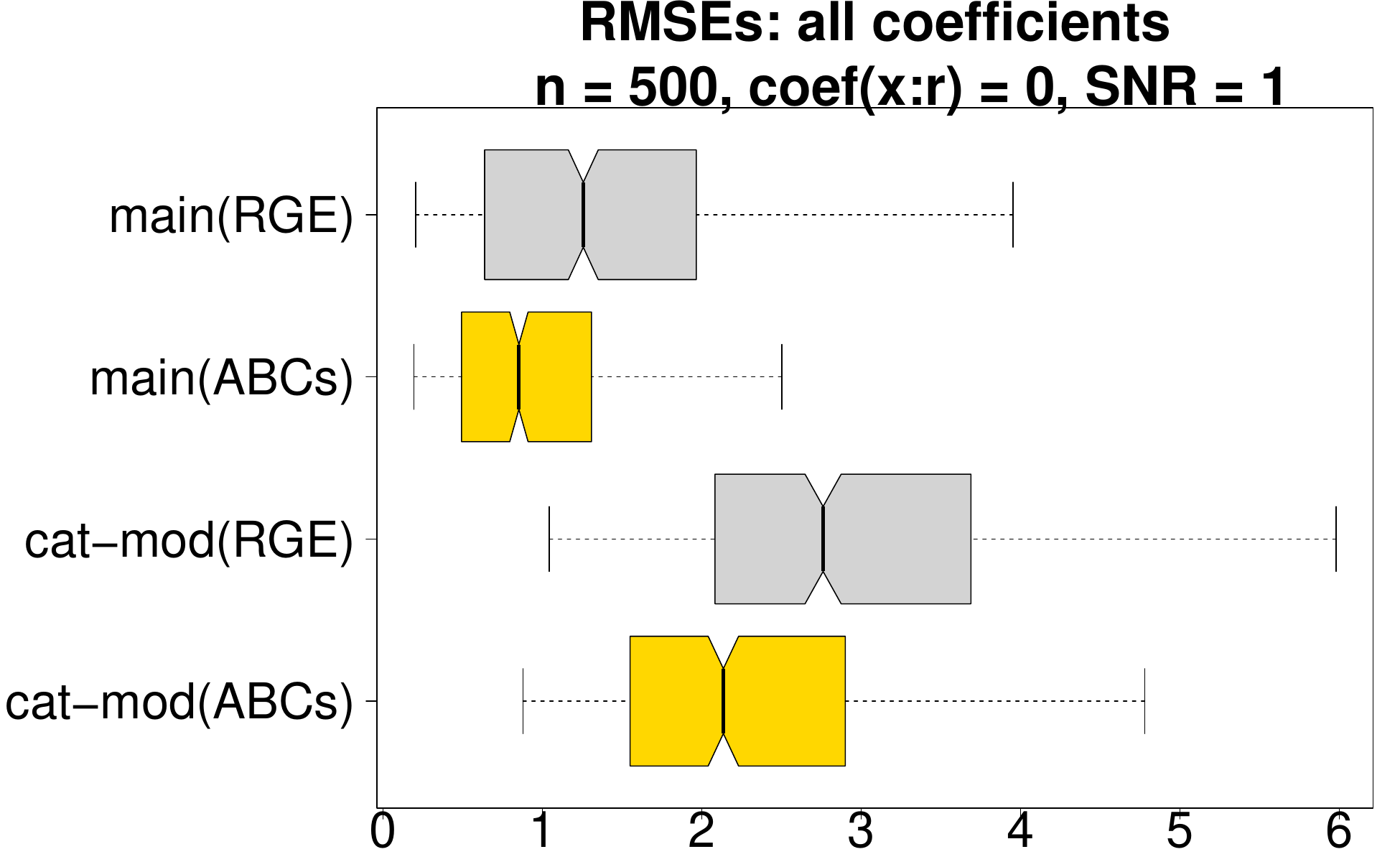}

\includegraphics[width=.49\linewidth]{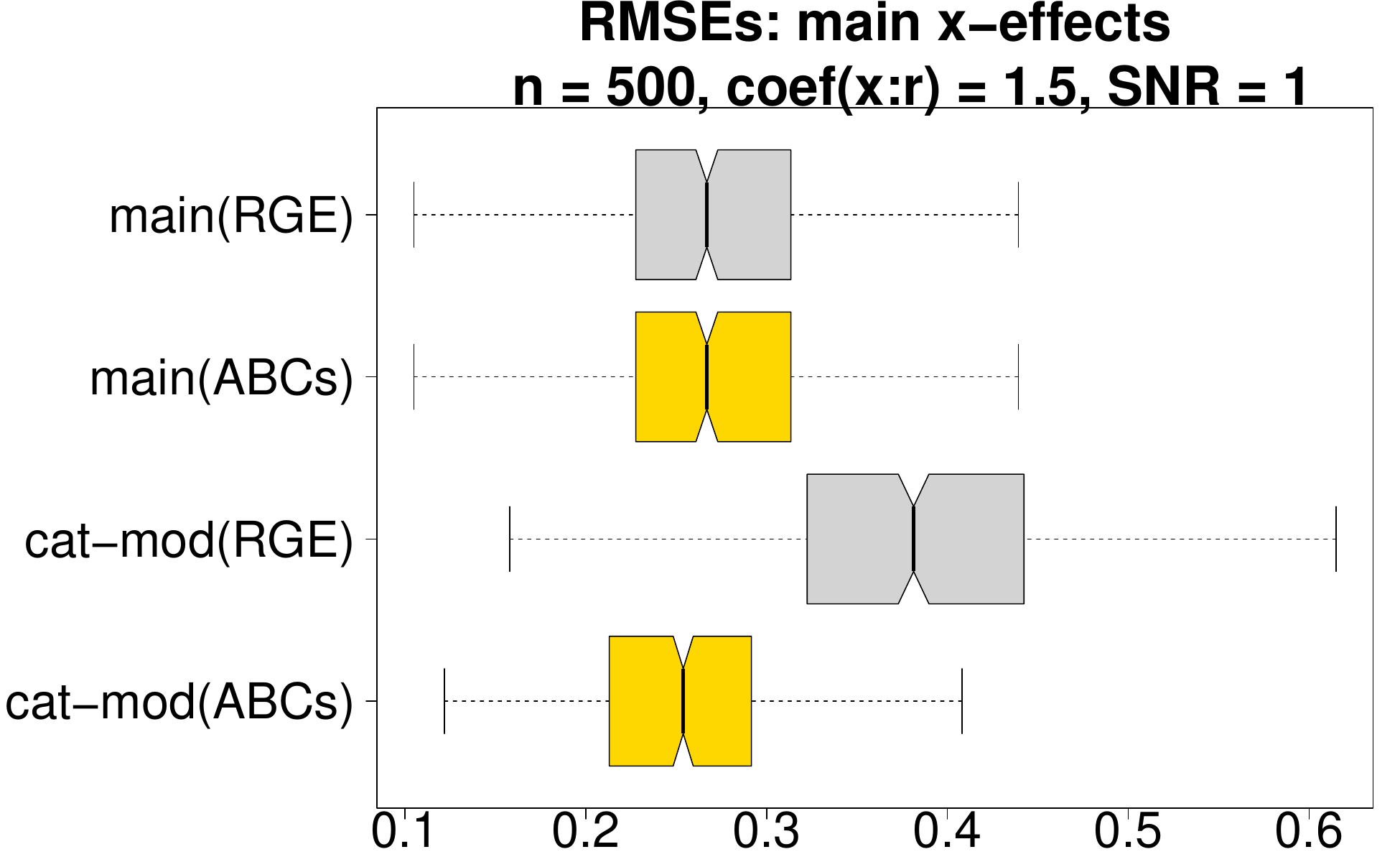}
\includegraphics[width=.49\linewidth]{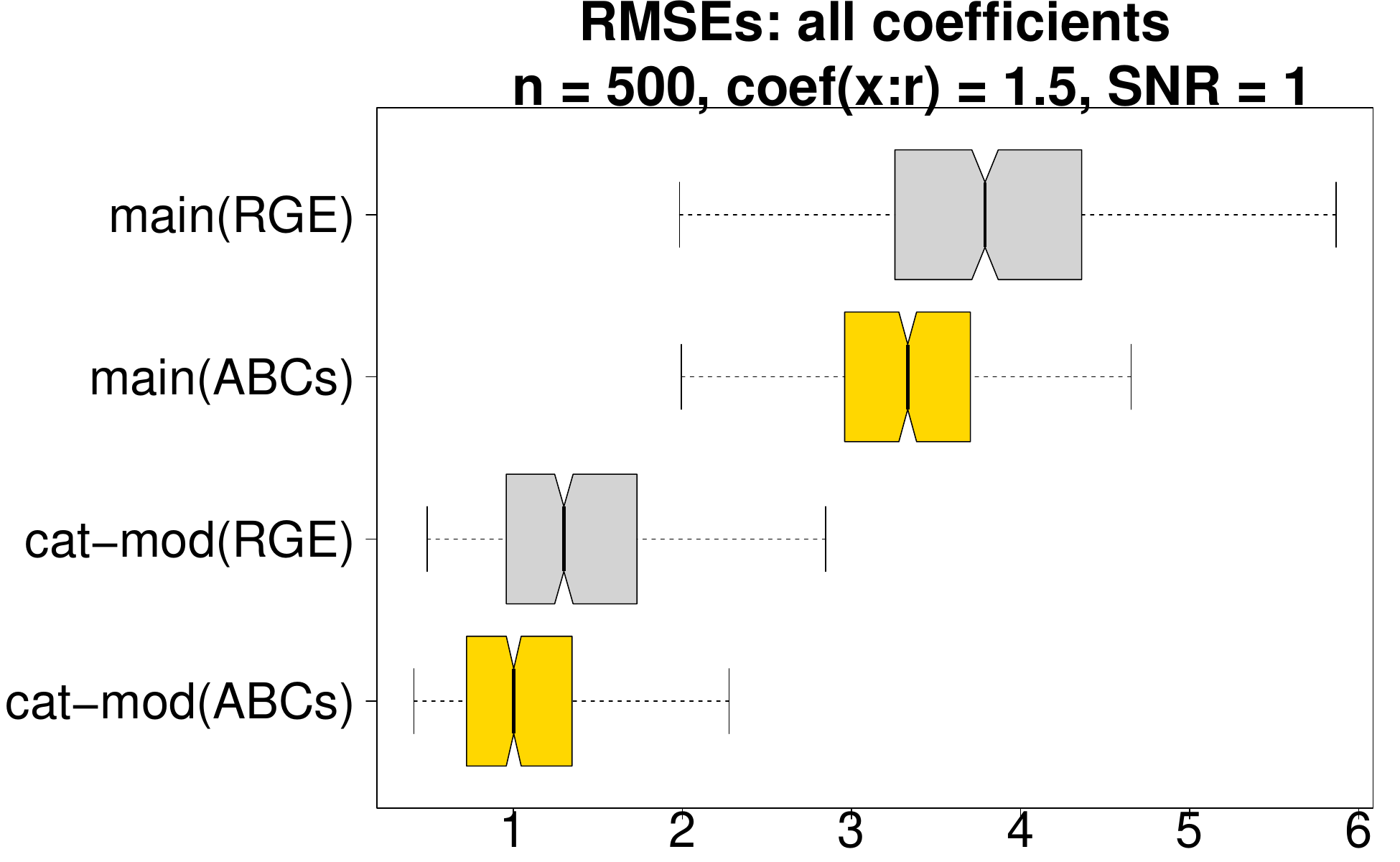}

\caption{\small RMSEs for the main $x$-effects (left) and all coefficients (right) under main-only and cat-modified models with ABCs (gold)  and RGE (gray). Boxplots are across 500 simulations; nonoverlapping notches indicate a difference in medians. Under ABCs, the cat-modified model main $x$-effect estimates are just as accurate as the main-only ones, even when the cat-modifiers are extraneous (top left), with slight gains when the cat-modifiers are necessary (bottom left). Neither result holds for RGE. For comparison, the accuracy across all coefficients  is primarily determined by whether the correct model (main-only, top right; cat-modified, bottom right) is used. 
}
\label{fig:sim-est}
\end{figure}

Inference is evaluated by mean interval widths and empirical coverage for 95\% confidence intervals for the regression coefficients (Figure~\ref{fig:sim-int}); narrow intervals are preferred, subject to nominal coverage.  For ABCs, the cat-modified model offers nearly the same statistical power for the main $x$-effects as does the main-only model, even when (all 33) cat-modifier effects are included extraneously (Figure~\ref{fig:sim-int}, top left). Compare that to inference for all coefficients (Figure~\ref{fig:sim-int}, top right): here, the inclusion of extraneous cat-modifiers inflates interval widths by more than 300\%. Clearly, this inferential robustness against extraneous cat-modifiers is a special property for 1) main effects and 2) ABCs.  When \emph{some} cat-modifiers are necessary, the cat-modified model with ABCs  improves statistical power for the main $x$-effects compared to the main-only models. Again, no such results hold for RGE, for which the cat-modified model consistently sacrifices statistical power. Finally, as expected, main-only models fail to provide coverage for active cat-modified parameters (Figure~\ref{fig:sim-int}, bottom right).

\begin{figure}[h]
\centering
\includegraphics[width=.49\linewidth]{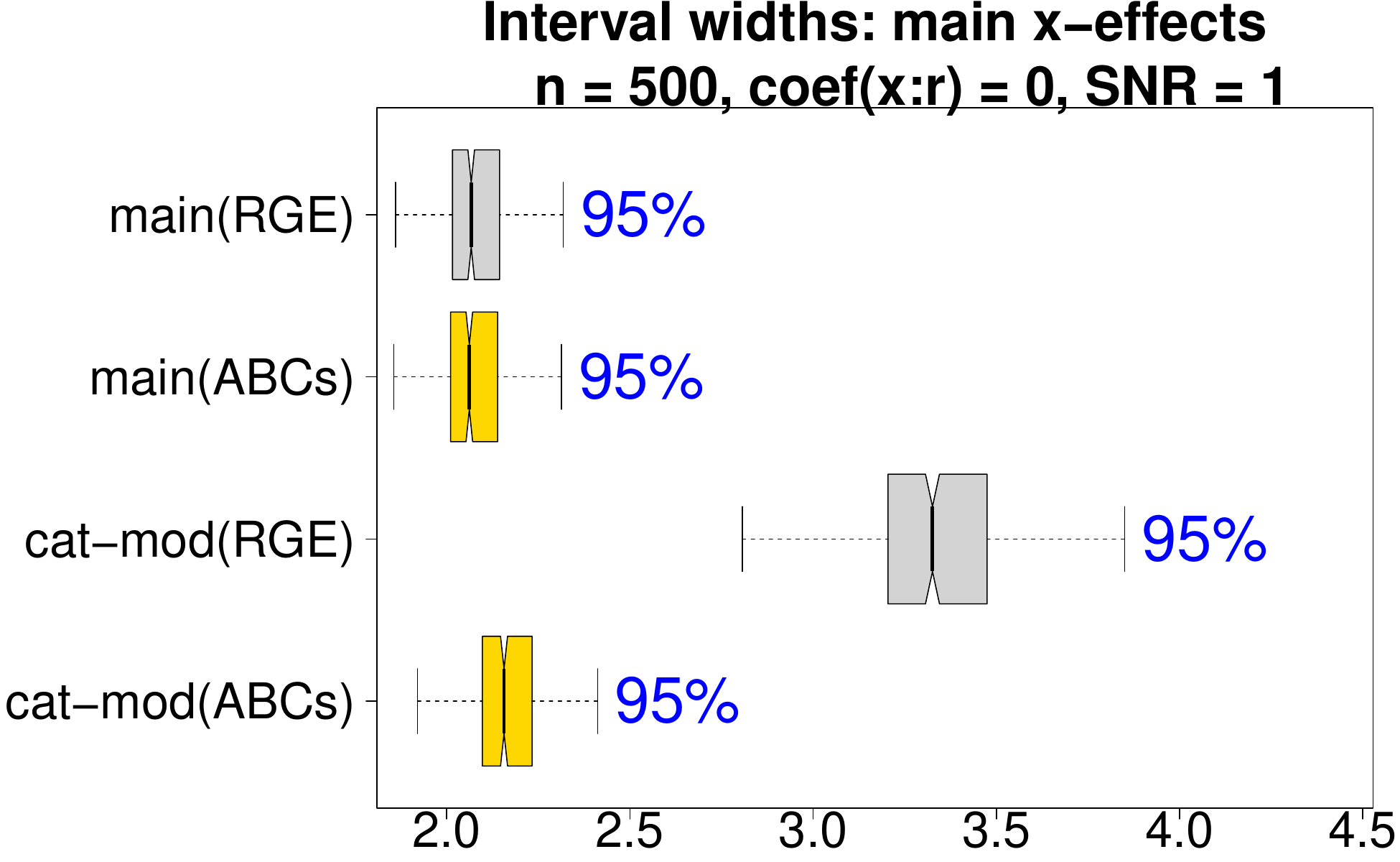}
\includegraphics[width=.49\linewidth]{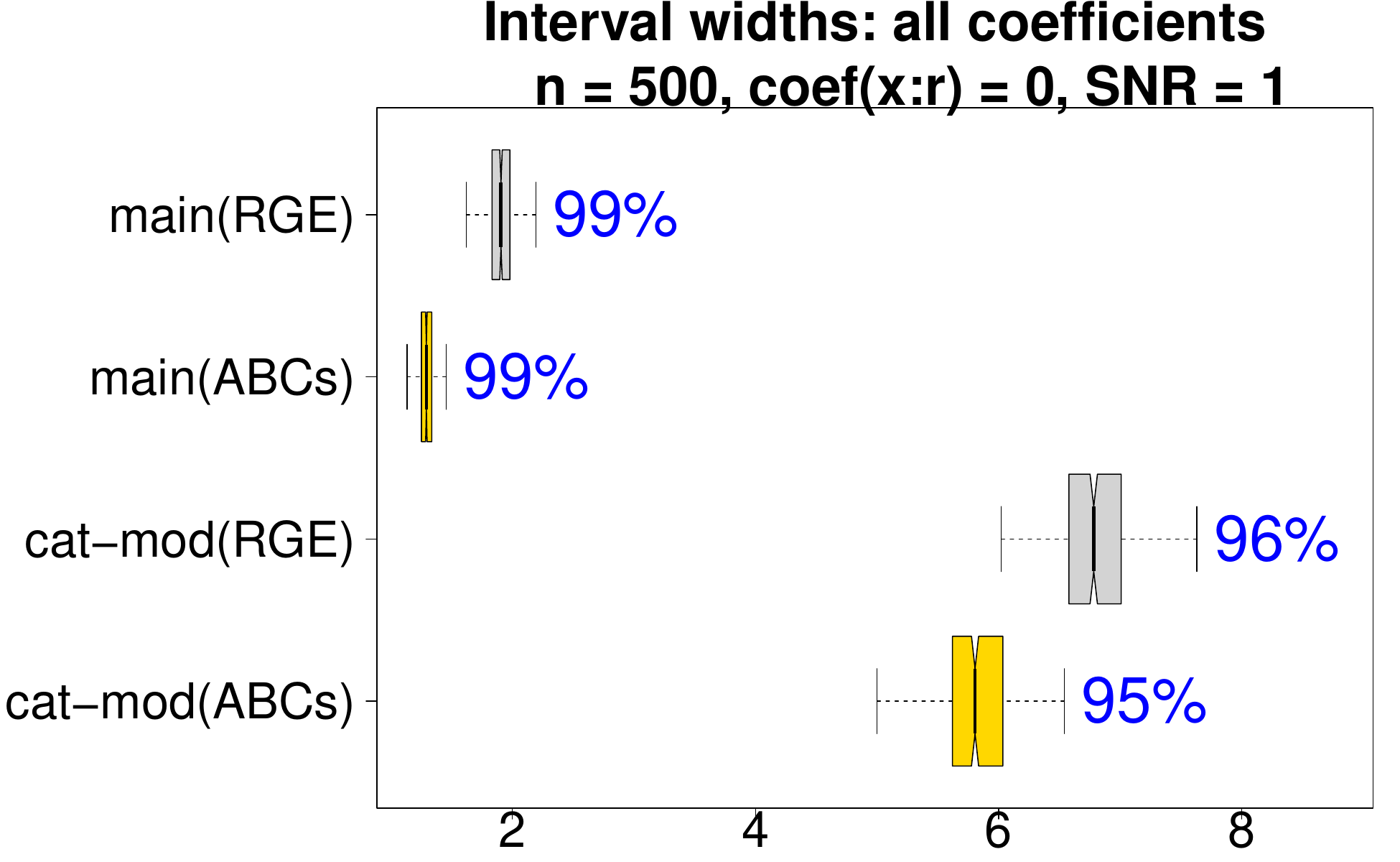}

\includegraphics[width=.49\linewidth]{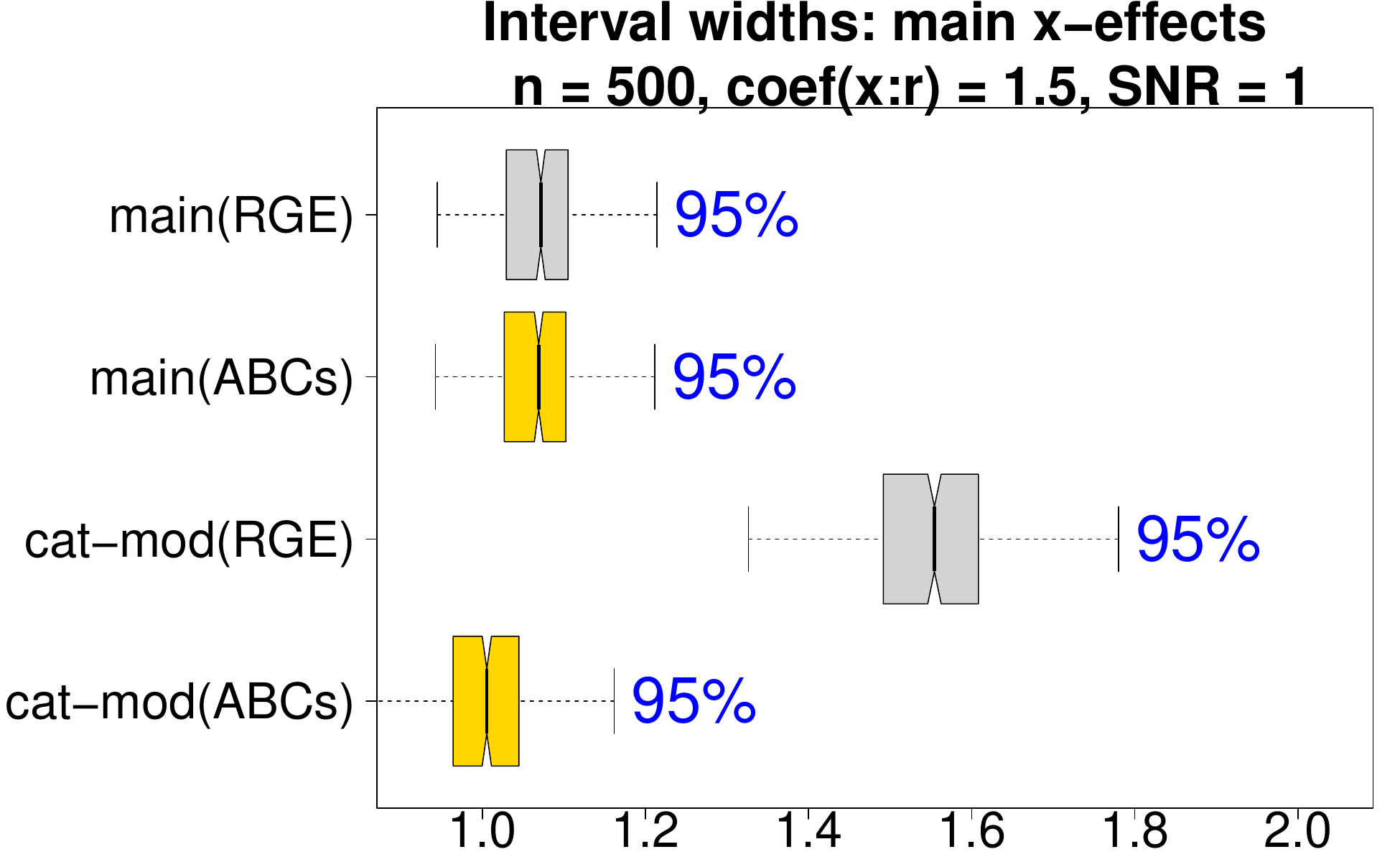}
\includegraphics[width=.49\linewidth]{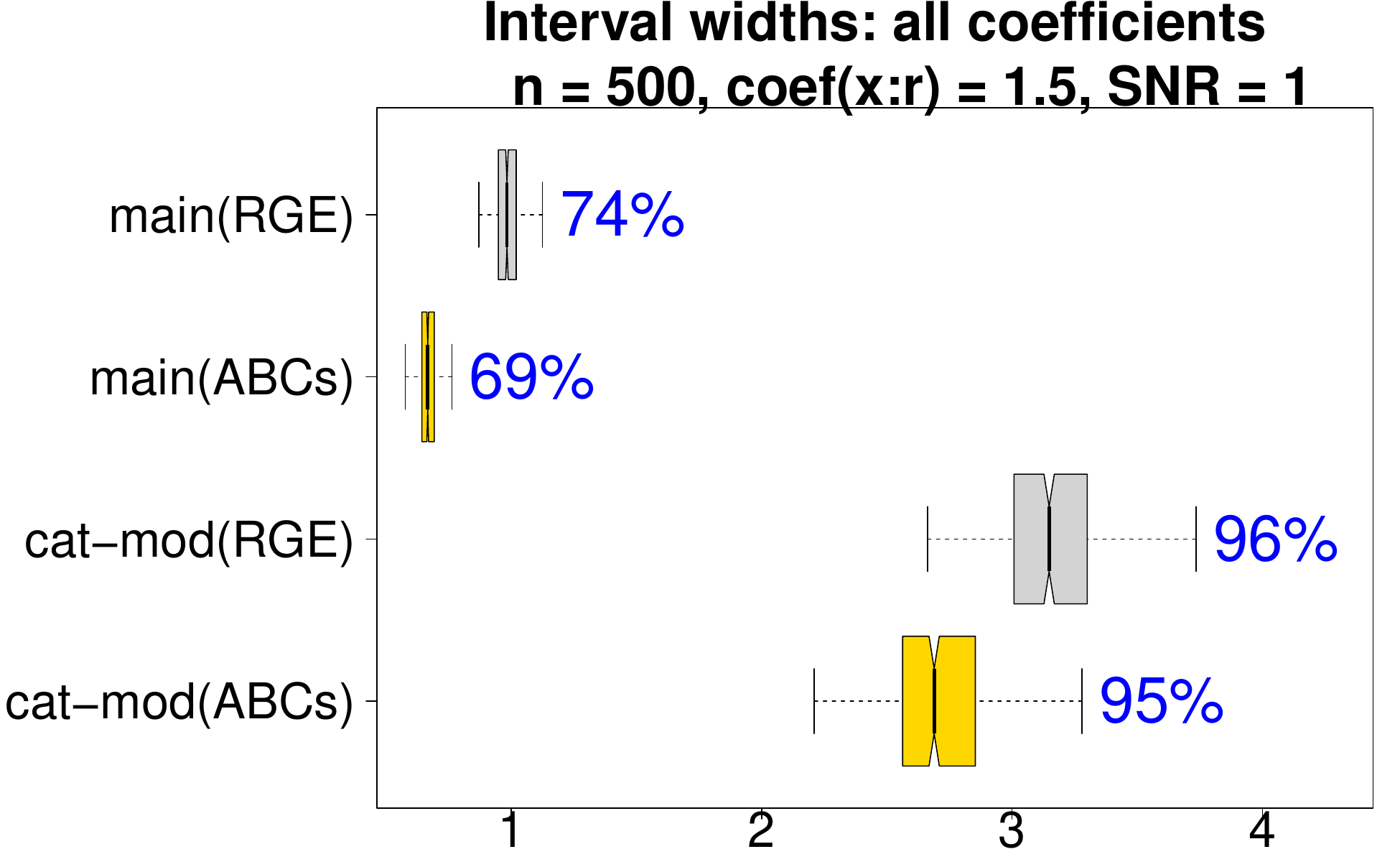}

\caption{\small Interval widths (boxplots) and empirical coverage (annotations) for 95\% confidence intervals for the main $x$-effects (left) and all coefficients (right) under main-only and cat-modified models with ABCs (gold) and RGE (gray). Under ABCs, inference for the main $x$-effects is nearly as powerful for the cat-modified model, even when the cat-modifiers are extraneous (top left), with greater power  when the cat-modifiers are necessary (bottom left). Neither result holds for RGE. For comparison, extraneous cat-modifiers increase interval widths overall (top right), while the omission of necessary cat-modifiers sacrifices coverage for the main-only models (bottom right). 
}
\label{fig:sim-int}
\end{figure}

The supplementary material includes additional results for smaller ($n=200$) and larger ($n=1000$) sample sizes; predictive evaluations based on RMSEs for $\mu(\bm x, \bm c)$; comparisons between ABCs and RGE for lasso and ridge regression, also including an ``overparametrized" version that does not impose any constraints; and modifications for $\sigma_{ac}=1.5$ that strongly violate the equal-variance condition \eqref{eq-v}, with similar results as for $\sigma_{ac}=1$.

\section{Application}\label{sec-app}
We apply cat-modified regression to assess heterogeneity among factors linked to STEM educational outcomes. Our dataset\footnote{Data management, access, and analysis are governed by data use
agreements and an Institutional Review Board–approved research protocol at the
University of Illinois Chicago.} links three administrative datasets to provide individual-level data for $n=27,638$ children in North Carolina (NC):   NC Detailed Birth Records, NC Blood Lead  Surveillance, and NC Standardized Testing Data; details are provided elsewhere \citep{ChildrensEnvironmentalHealthInitiative2020,KowalPRIME2020,Bravo2022}. The STEM educational outcome variable $y_i$ is the end-of-4th-grade standardized math score for student $i$, centered and scaled by year of test (2010, 2011, or 2012). These math scores are linked with a rich collection of  demographic, social, and environmental exposure variables. The continuous covariates are  racial (residential) isolation (\texttt{RI}), which is a measure of structural racism based on neighborhood information; blood lead level (\texttt{BLL}), which measures lead exposure; 
birthweight percentile (\texttt{BWTpct}); mother's age at time of child's birth (\texttt{mAge}); and exposure to the air pollutant $\mbox{PM}_{2.5}$  during the year prior to the exam (\texttt{PM2.5}). The continuous covariates are centered and scaled. The categorical covariates are mother's race (\texttt{race}), child's sex (\texttt{sex}), mother's education level (\texttt{mEdu}), and an indicator of economically disadvantaged (\texttt{EconDisadv}) determined by participation in the National Lunch Program; see Table~\ref{tab:results} for categorical levels and proportions. 

Our linear regression analysis spans from main-only models to a variety of cat-modified models, expanding significantly upon the simple models from  Tables~\ref{tab:ex-x}~and~\ref{tab:ex-cat}. First, we establish a \emph{main-only} model that includes each of these covariates (\texttt{RI}, \texttt{BLL}, \texttt{BWTpct}, \texttt{mAge}, \texttt{PM2.5}, \texttt{race}, \texttt{sex}, \texttt{mEdu}, and \texttt{EconDisadv}) but no interactions. 
The main-only model features a variety of interesting demographic, socio-economic, maternal, and environmental exposure variables, with  16 regression parameters (12 identified). Next, the \emph{race-modified} model adds an interaction between \texttt{race} and every other covariate. 
This expansion allows for heterogeneous effects of each variable by race, thus providing insights into the myriad impacts of race on each child's life course and educational outcomes,  with  52 regression parameters (30 identified). Finally, the \emph{cat-modified} model adds all pairwise categorical-continuous and categorical-categorical interactions. This instance of \eqref{reg-cm} allows the fullest (pairwise) extent of heterogeneous effects across the rich collection of demographic and socio-economic variables (\texttt{race}, \texttt{sex}, \texttt{mEdu}, and \texttt{EconDisadv}), with 103 regression parameters (55 identified). We fit each of these models under ABCs and RGE (references \texttt{White}, \texttt{Male}, lowest \texttt{mEdu} (\texttt{mEdu<HS}), and not \texttt{EconDisadv}). 

While each model offers potential for insight, a critical limitation of popular identification approaches, especially RGE, is that the estimates, inference, and interpretations of the main effects are highly sensitive to the choice of cat-modifiers. To see this, we present the main effect OLS estimates and 95\% confidences intervals across these models in Figure~\ref{fig:nc}. With RGE (right), the main effects shift and the intervals widen considerably---with increases from 160\% to 230\% in interval widths---upon adding race- (blue) and other (red) cat-modifiers. These main effects and accompanying interaction effects (not shown) are anchored at the reference groups and refer to different functionals of $\mu(\bm x, \bm c)$ under each model---even though the statistical output for the ``main effects" is typically presented identically, regardless of any cat-modifiers. Thus, while cat-modified models are essential for heterogeneous effects, there is a cost incurred under RGE: each additional  cat-modifier requires careful re-consideration of the main  and interaction effects, which impedes statistical analysis and undermines interpretability. 


\begin{figure}[h]
\centering
\includegraphics[width=.49\linewidth]{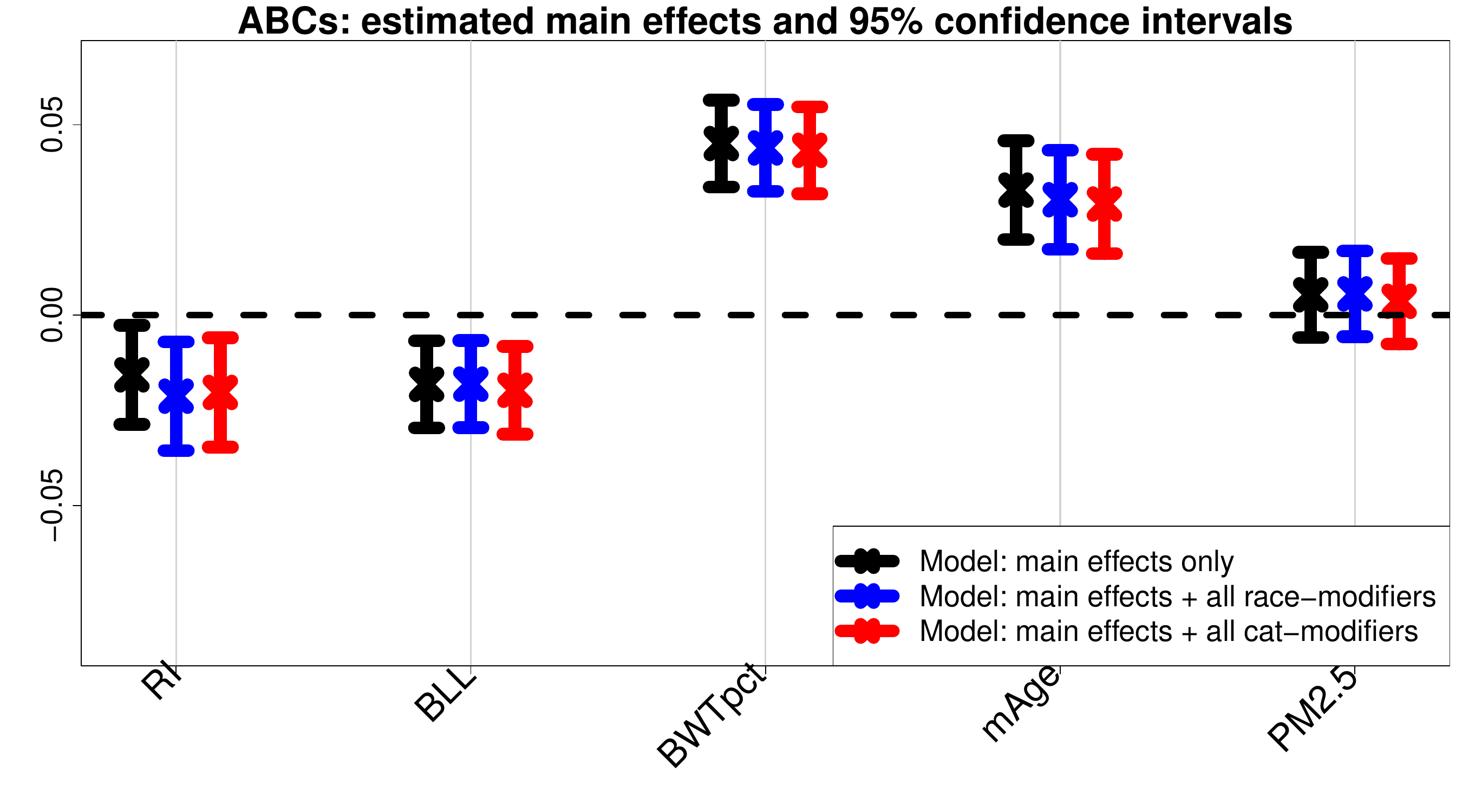}
\includegraphics[width=.49\linewidth]{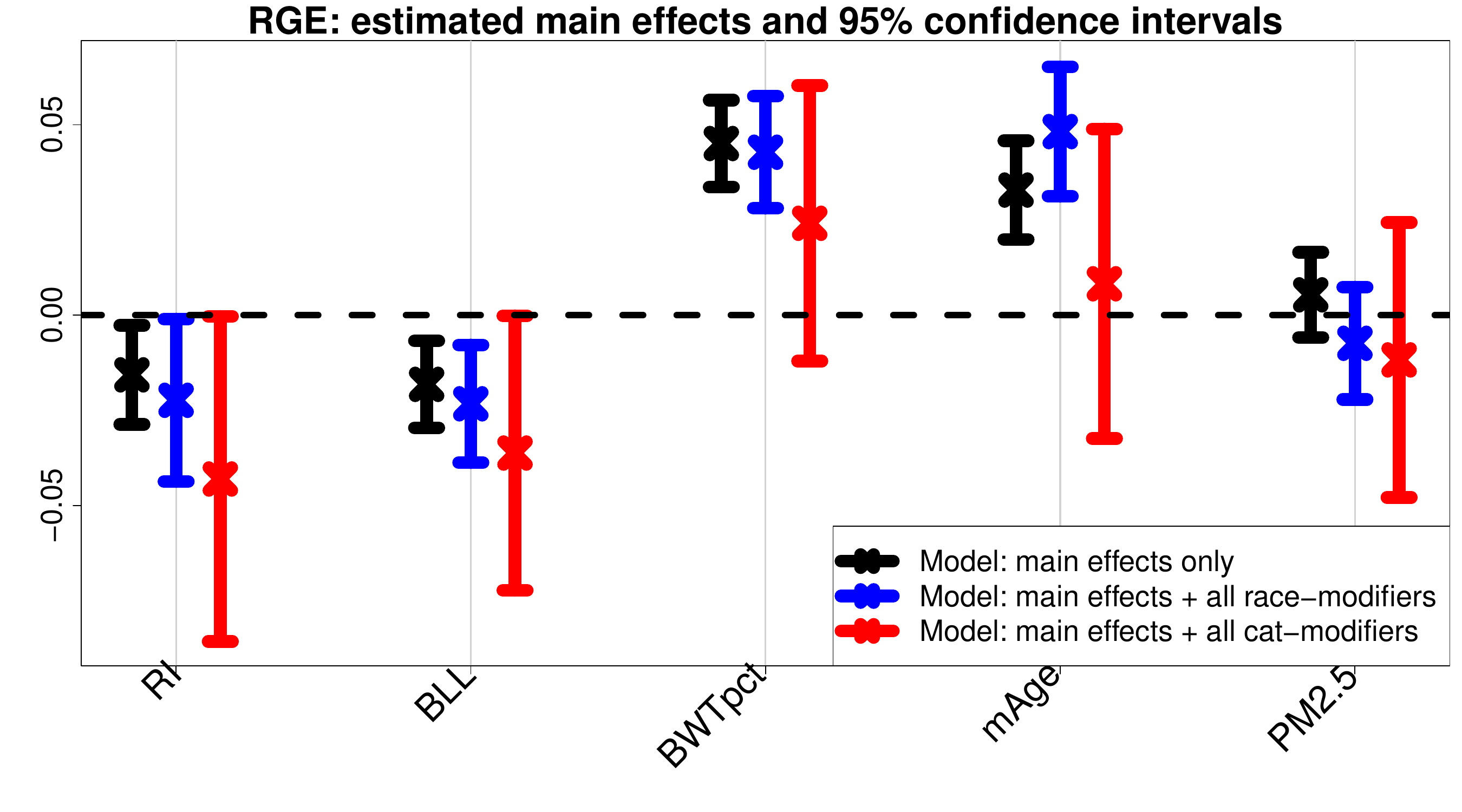}
\includegraphics[width=.49\linewidth]{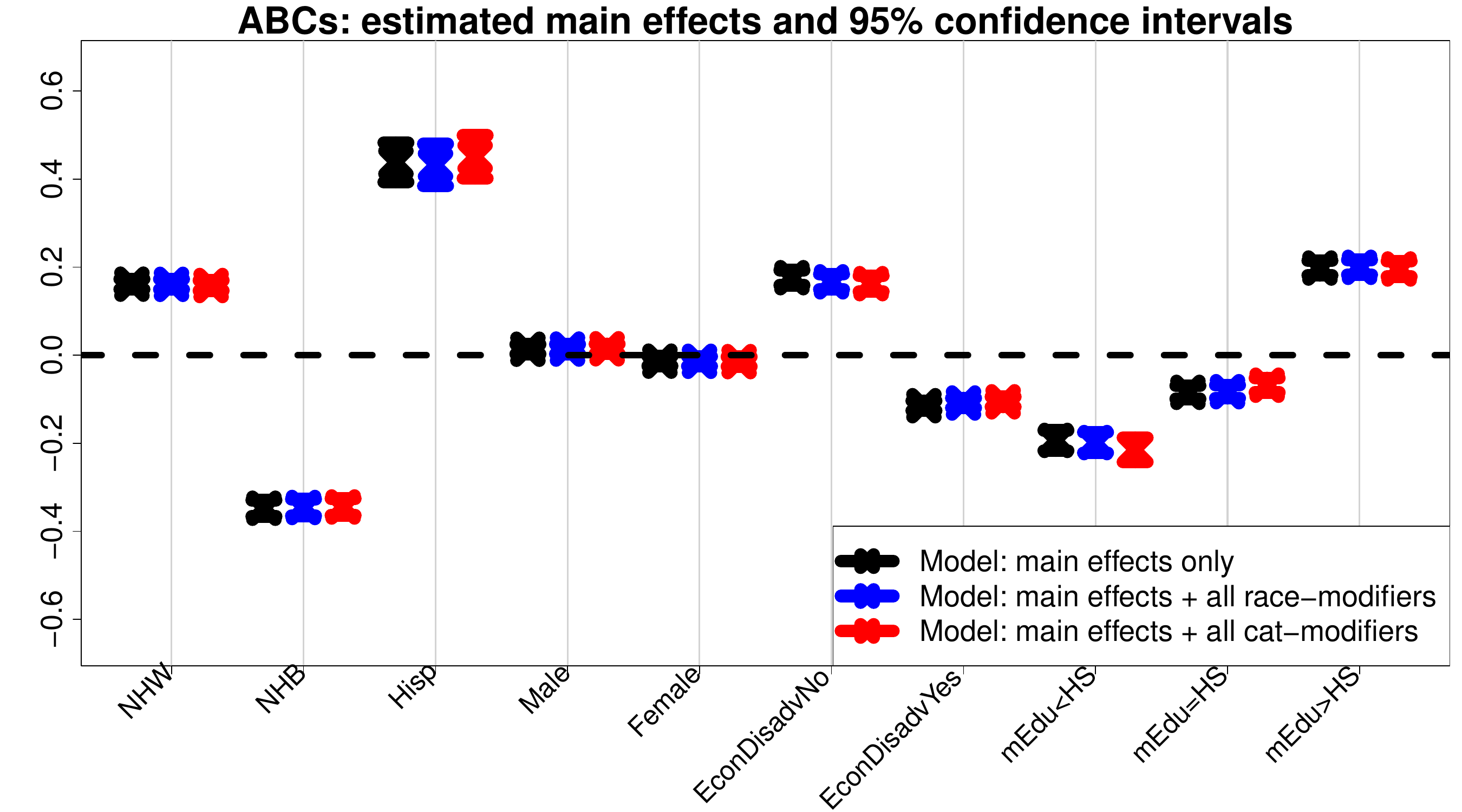}
\includegraphics[width=.49\linewidth]{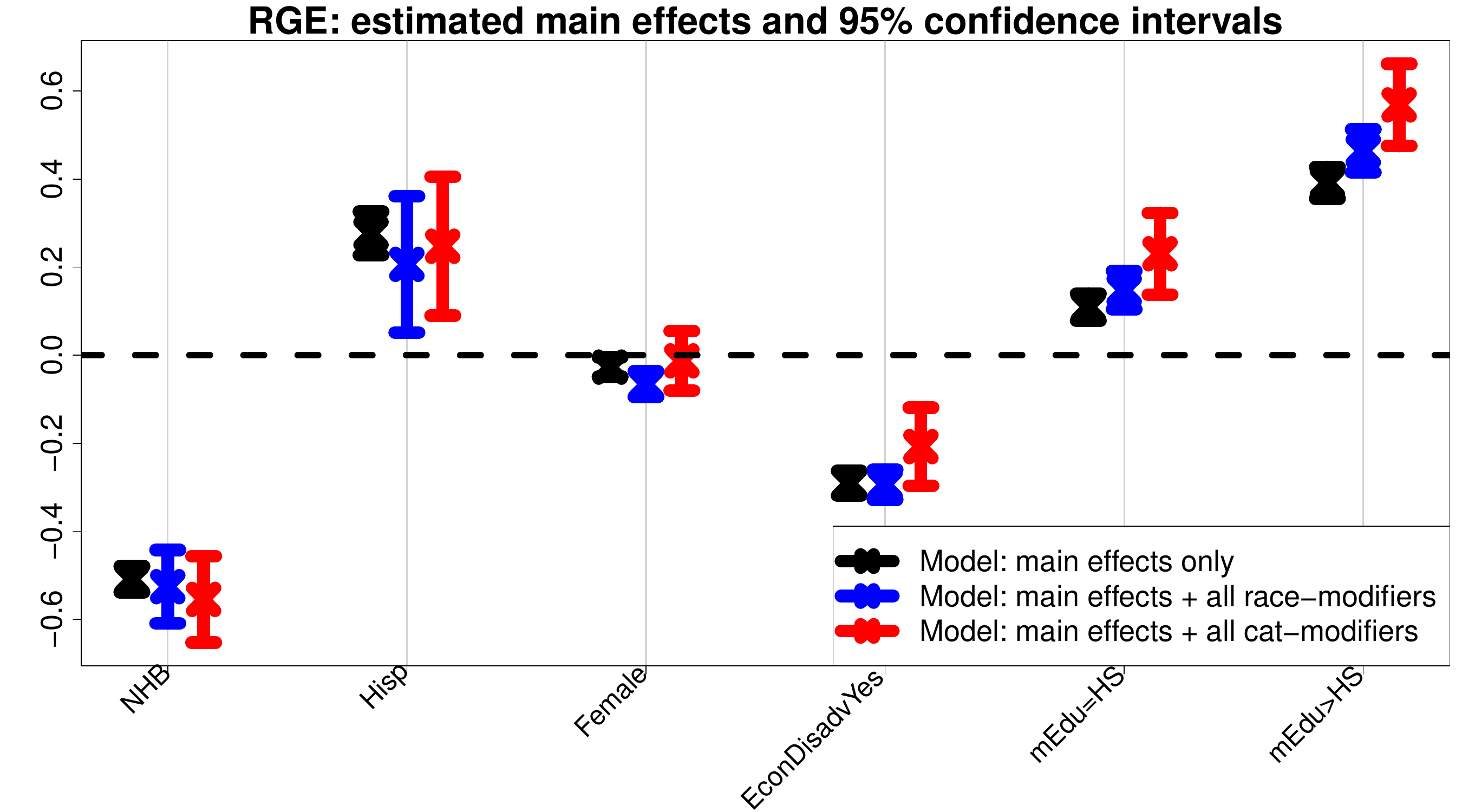}
\caption{\small OLS estimates and 95\% confidence intervals for continuous (top) and categorical (bottom) main effects under ABCs (left) and RGE (right) for three linear models: the \emph{main-only model} (black) includes \texttt{RI}, \texttt{BLL}, \texttt{BWTpct}, \texttt{mAge}, \texttt{PM2.5}, \texttt{race}, \texttt{sex}, \texttt{mEdu}, and \texttt{EconDisadv}; 
the \emph{race-modified} model (blue) adds interactions between \texttt{race} and every other covariate; and  the \emph{cat-modified} model (red) adds all pairwise categorical-continuous and categorical-categorical interactions. With ABCs, main effect inference is invariant to the cat-modifiers: all point and interval estimates are nearly identical across these substantially different models.  
With RGE, the main effect estimates shift and the intervals expand considerably as more cat-modifiers are added. 
}
\label{fig:nc}
\end{figure}

The invariance of ABCs resolves these limitations: estimation and inference for the main effects (Figure~\ref{fig:nc}, left) are nearly identical across these substantially different models. This occurs despite strong dependencies among the covariates (and interactions) with both continuous and categorical variables. ABCs effectively decouple the main effects from the cat-modifiers: even adding 87 parameters (43 identified) from the main-only model to obtain the cat-modified model does not lessen, and in some cases \emph{increases} the statistical power for the main effects. With ABCs, the statistical analyst may consider these or other cat-modified models without compromising or complicating inferences for the main effects.

The full regression output from the cat-modified model with ABCs is in Tables~\ref{tab:results}~and~\ref{tab:results-app}. 
Lower math scores are strongly ($p < 0.01$) associated with racial (residential) isolation, lead exposure, lower (mother's) education levels, and occur for non-Hispanic Black and economically disadvantaged students; higher math scores are strongly associated with birthweight percentile, mother's age, and the opposing categories from above. ABCs provide output for all levels of all categorical variables, thus eliminating the presentation bias of RGE that presents all output relative to the reference groups (White, Male, etc.). For categorical variables with ABCs, the estimates and SEs are directly related to the abundances: for instance, categorical variables with equal proportions such as \texttt{sex}, the main and \texttt{sex}-continuous interaction estimates are equal and opposite with identical SEs for  Male and Female. 
The regression output strongly supports heterogeneous effects, most notably via mother's education level and with intersectionality of race and sex (e.g., \citealp{Bauer2014}).

    


\begin{table}[h] \scriptsize
\centering
\begin{tabular}{lrr}
Variable & Estimate (SE)  & $p$-value  \\
\hline
Intercept & -0.026  (0.008) & 0.001 \\
Racial isolation (\texttt{RI}) & -0.020  (0.007) & 0.006 \\
Blood lead level (\texttt{BLL}) & -0.020  (0.006) & 0.001 \\
Birthweight percentile (\texttt{BWTpct}) & 0.043  (0.006) & $<$0.001 \\
Mother's age (\texttt{mAge})  & 0.029 (0.007) & $<$0.001 \\
$\mbox{PM}_{2.5}$ exposure (\texttt{PM2.5}) & 0.004  (0.006) & 0.527 \\
Mother's race (\texttt{race}) & \\
\quad \texttt{White} (58.7\%) & 0.158  (0.006) & $<$0.001 \\
\quad \texttt{Black} (35.1\%) & -0.345  (0.010) & $<$0.001 \\
\quad \texttt{Hispanic} (6.2\%) & 0.451  (0.025) & $<$0.001 \\
Child's sex (\texttt{sex})  & \\
\quad \texttt{Male} (49.9\%) & 0.015  (0.006) & 0.010 \\
\quad \texttt{Female} (50.1\%) & -0.015  (0.006) & 0.010 \\
Mother's education level (\texttt{mEdu}) & \\
\quad Did not complete high school \\ \quad (\texttt{<HS}; 24.0\%) & -0.215  (0.014) & $<$0.001 \\
\quad Completed high school   (\texttt{=HS}; 36.8\%)& -0.068  (0.008) & $<$0.001 \\
\quad At least some postsecondary \\ \quad (\texttt{>HS}; 39.2\%)  & 0.196  (0.009) & $<$0.001 \\
\hline
\texttt{White:Male} & 0.023  (0.006) & $<$0.001 \\
        \texttt{Black:Male} & -0.049  (0.010) & $<$0.001 \\
        \texttt{Hisp:Male} & 0.056  (0.024) & 0.019 \\
        \texttt{White:Female} & -0.023  (0.006) & $<$0.001 \\
        \texttt{Black:Female} & 0.048  (0.009) & $<$0.001 \\
        \texttt{Hisp:Female} & -0.051  (0.022) & 0.019 \\
        \texttt{White:mEdu<HS} & -0.042  (0.014) & 0.003 \\
        \texttt{Black:mEdu<HS} & 0.023  (0.017) & 0.166 \\
        \texttt{Hisp:mEdu<HS} & 0.062  (0.018) & 0.001 \\
        \texttt{White:mEdu=HS} & 0.000  (0.008) & 0.971 \\
        \texttt{Black:mEdu=HS} & 0.008  (0.011) & 0.462 \\
        \texttt{Hisp:mEdu=HS} & -0.071  (0.038) & 0.059 \\
        \texttt{White:mEdu>HS} & 0.017  (0.007) & 0.012 \\
        \texttt{Black:mEdu>HS} & -0.031  (0.016) & 0.064\\
        \texttt{Hisp:mEdu>HS} & -0.172  (0.066) & 0.009 \\
        \texttt{Male:mEdu<HS} & -0.018 (0.012) & 0.131 \\
        \texttt{Female:mEdu<HS} & 0.017 (0.011) & 0.131 \\
        \texttt{Male:mEdu=HS} & 0.007 (0.008) & 0.390 \\
        \texttt{Female:mEdu=HS} & -0.006 (0.007) & 0.390 \\
        \texttt{Male:mEdu>HS} & 0.005 (0.008) & 0.570 \\
        \texttt{Female:mEdu>HS} & -0.005 (0.009) & 0.570 \\
        \hline 
        \vspace{-5mm}
\end{tabular}
\begin{tabular}{lrr} \scriptsize
        Variable (continued) & Estimate (SE)  & $p$-value  \\
        \hline
        \texttt{RI:White} & -0.002 (0.006) & 0.795 \\
        \texttt{RI:Black} & -0.005  (0.009) & 0.565 \\
        \texttt{RI:Hisp} & 0.046  (0.025) & 0.063 \\
        \texttt{BLL:White} & -0.003  (0.005) & 0.582 \\
        \texttt{BLL:Black} & -0.004  (0.008) & 0.620 \\
        \texttt{BLL:Hisp} & 0.050  (0.023) & 0.033 \\
        \texttt{BWTpct:White} & -0.002  (0.005) & 0.731 \\
        \texttt{BWTpct:Black} & 0.006  (0.008) & 0.512 \\
        \texttt{BWTpct:Hisp}  & -0.014 (0.023) & 0.548 \\
\texttt{mAge:White} & 0.009  (0.006) & 0.120 \\
\texttt{mAge:Black} & -0.017  (0.009) & 0.071 \\
\texttt{mAge:Hisp} & 0.009  (0.027) & 0.733 \\
\texttt{PM2.5:White} & -0.019  (0.005) & $<$0.001 \\
\texttt{PM2.5:Black} & 0.024  (0.008) & 0.004 \\
\texttt{PM2.5:Hisp} & 0.037  (0.024) & 0.123 \\
\texttt{RI:Male} & 0.001  (0.007) & 0.835 \\
\texttt{RI:Female} & -0.001  (0.007) & 0.835 \\
\texttt{BLL:Male} & 0.002  (0.006) & 0.793 \\
\texttt{BLL:Female} & -0.002  (0.006) & 0.793 \\
\texttt{BWTpct:Male} & 0.001  (0.006) & 0.908 \\
\texttt{BWTpct:Female} & -0.001  (0.006) & 0.908 \\
\texttt{mAge:Male} & -0.007  (0.007) & 0.290 \\
\texttt{mAge:Female} & 0.007  (0.007) & 0.290 \\
\texttt{PM2.5:Male} & -0.005  (0.006) & 0.370 \\
\texttt{PM2.5:Female} & 0.005  (0.006) & 0.370 \\
\texttt{RI:mEdu<HS} & -0.015  (0.012) & 0.201 \\
\texttt{RI:mEdu=HS} & -0.007  (0.009) & 0.442 \\
\texttt{RI:mEdu>HS} & 0.016  (0.010) & 0.104 \\
\texttt{BLL:mEdu<HS} & -0.004  (0.011) & 0.682 \\
\texttt{BLL:mEdu=HS} & 0.011  (0.008) & 0.145 \\
\texttt{BLL:mEdu>HS} & -0.008  (0.008) & 0.350 \\
\texttt{BWTpct:mEdu<HS} & -0.018  (0.011) & 0.110 \\
\texttt{BWTpct:mEdu=HS} & 0.011  (0.008) & 0.156 \\
\texttt{BWTpct:mEdu>HS} & 0.001  (0.008) & 0.912 \\
\texttt{mAge:mEdu<HS} & -0.039  (0.013) & 0.003 \\
\texttt{mAge:mEdu=HS} & -0.022  (0.009) & 0.011 \\
\texttt{mAge:mEdu>HS} & 0.045  (0.009) & $<$0.001 \\
\texttt{PM2.5:mEdu<HS} & -0.002  (0.011) & 0.849 \\
\texttt{PM2.5:mEdu=HS} & -0.013  (0.008) & 0.091 \\
\texttt{PM2.5:mEdu>HS} & 0.013  (0.008) & 0.096 \\
    \hline
\end{tabular}
\vspace{-2mm} \caption{\footnotesize Cat-modified model output under ABCs for NC STEM education outcomes with all pairwise categorical-continuous and categorical-categorical interactions (see Table~\ref{tab:results-app} for \texttt{EconDisadv} effects). Categorical variable proportions are also  indicated. Data are restricted to individuals with 37-42 weeks gestation, $\texttt{mAge} \in [15, 44]$  years, $\texttt{BLL} \le 80 \mu g/dL$ (and capped at $10\mu g/dL$), birth order $\le 4$, 
no current English language learners,
and residence in NC at the time of birth and time of 4th end-of-grade test. 
\label{tab:results}}
\end{table}

Finally, we simplify the heavily-parametrized cat-modified model by fitting a lasso regression under ABCs; $\lambda$ is selected using 10-fold cross-validation and the one-standard-error rule \citep{hastie2009elements}. The selected main effects  (\texttt{RI}, \texttt{BLL}, \texttt{BWTpct}, \texttt{mAge}, \texttt{race}, \texttt{mEdu}, and \texttt{EconDisadv}) match the conclusions from Figure~\ref{fig:nc}. Among interactions,  coefficients from \texttt{race:mEdu}, \texttt{race:mAge}, \texttt{race:EconDisadv}, \texttt{mEdu:mAge}, and \texttt{EconDisadv:mAge} are selected. The accompanying coefficients of these cat-modifiers suggest that some positive effects are not as beneficial for minoritized groups: the positive effect of  mother's education (\texttt{mEdu>HS}) are attenuated for  Black and Hispanic students, while the benefits of mother's age are less so for lower mother's education, Black, or economically disadvantaged students.


\section{Conclusion}
\label{sec-conc}
To encourage and enable statistical analysis of heterogeneous effects, we analyzed and advocated ABCs---an alternative parametrization and estimation strategy for cat-modified models that include categorical-continuous or categorical-categorical interactions. Unlike default methods, ABCs allow the inclusion of cat-modifiers ``for free": there is virtually no impact on the main effect estimates, while main effect inference is stable or more powerful. We rigorously proved these estimation and inference invariance properties and validated them empirically with extensive simulation studies. We also provided strategies for estimation and inference, including both generalized and regularized regression. Finally, we applied these tools to analyze STEM educational outcomes and showed how ABCs facilitate  identification and estimation of (demographic) heterogeneous effects without incurring any costs---in estimation, inference, or interpretation---for the main effects.  


Despite these many advantages, we note several caveats. First, ABCs may increase susceptibility to $p$-hacking. Because ABCs facilitate the inclusion of interactions, and with a large enumeration of potential interactions,  there is a heightened potential for both discovery \emph{and} false discovery. Proper statistical analyses require careful consideration of hypothesis tests with  multiple testing corrections
as appropriate.  Second, ABCs cannot guarantee that cat-modifiers will be (practically or statistically) significant. Detection of heterogeneous effects often requires well-designed studies or large sample sizes. 
Third, our invariance results apply for least squares estimation, but not more general loss functions. Finally, many categorical variables, especially race, sex, and other protected groups, are susceptible to misinterpetation, inaccurate labelings, and exclusions of small groups.







 
\section*{Acknowledgements}
We thank Virginia Baskin,  Caleb Fikes, Prayag Gordy, and Jai Uparkar for helpful discussions and their contributions to software development.

\bibliographystyle{chicago}
\bibliography{refs}

\begin{thebibliography}{}

\bibitem[\protect\citeauthoryear{Bauer}{Bauer}{2014}]{Bauer2014}
Bauer, G.~R. (2014).
\newblock Incorporating intersectionality theory into population health
  research methodology: challenges and the potential to advance health equity.
\newblock {\em Social Science \& Medicine\/}~{\em 110}, 10--17.

\bibitem[\protect\citeauthoryear{Bien, Taylor, and Tibshirani}{Bien
  et~al.}{2013}]{Bien2013}
Bien, J., J.~Taylor, and R.~Tibshirani (2013).
\newblock A lasso for hierarchical interactions.
\newblock {\em The Annals of Statistics\/}~{\em 41}, 1111.

\bibitem[\protect\citeauthoryear{Bravo, Zephyr, Kowal, Ensor, and
  Miranda}{Bravo et~al.}{2022}]{Bravo2022}
Bravo, M., D.~Zephyr, D.~R. Kowal, K.~B. Ensor, and M.~L. Miranda (2022).
\newblock Racial residential segregation shapes relationships between early
  childhood lead exposure and 4th grade standardized test scores.
\newblock {\em Proceedings of the National Academy of Sciences\/}~{\em 119},
  e2117868119.

\bibitem[\protect\citeauthoryear{Brehm and Alday}{Brehm and
  Alday}{2022}]{Brehm2022}
Brehm, L. and P.~M. Alday (2022).
\newblock Contrast coding choices in a decade of mixed models.
\newblock {\em Journal of Memory and Language\/}~{\em 125}, 104334.

\bibitem[\protect\citeauthoryear{Chen, Pierson, Rose, Joshi, Ferryman, and
  Ghassemi}{Chen et~al.}{2021}]{Chen2021}
Chen, I.~Y., E.~Pierson, S.~Rose, S.~Joshi, K.~Ferryman, and M.~Ghassemi
  (2021).
\newblock Ethical machine learning in healthcare.
\newblock {\em Annual Review of Biomedical Data Science\/}~{\em 4}, 123--144.

\bibitem[\protect\citeauthoryear{Chestnut and Markman}{Chestnut and
  Markman}{2018}]{Chestnut2018}
Chestnut, E.~K. and E.~M. Markman (2018).
\newblock “girls are as good as boys at math” implies that boys are
  probably better: A study of expressions of gender equality.
\newblock {\em Cognitive science\/}~{\em 42}, 2229--2249.

\bibitem[\protect\citeauthoryear{Fujikoshi}{Fujikoshi}{1993}]{Fujikoshi1993}
Fujikoshi, Y. (1993).
\newblock Two-way anova models with unbalanced data.
\newblock {\em Discrete Mathematics\/}~{\em 116}, 315--334.

\bibitem[\protect\citeauthoryear{Grotenhuis, Pelzer, Eisinga, Nieuwenhuis,
  Schmidt-Catran, and Konig}{Grotenhuis et~al.}{2017a}]{TeGrotenhuis2017a}
Grotenhuis, M.~T., B.~Pelzer, R.~Eisinga, R.~Nieuwenhuis, A.~Schmidt-Catran,
  and R.~Konig (2017a).
\newblock A novel method for modelling interaction between categorical
  variables.
\newblock {\em International Journal of Public Health\/}~{\em 62}, 427--431.

\bibitem[\protect\citeauthoryear{Grotenhuis, Pelzer, Eisinga, Nieuwenhuis,
  Schmidt-Catran, and Konig}{Grotenhuis et~al.}{2017b}]{TeGrotenhuis2017}
Grotenhuis, M.~T., B.~Pelzer, R.~Eisinga, R.~Nieuwenhuis, A.~Schmidt-Catran,
  and R.~Konig (2017b).
\newblock When size matters: advantages of weighted effect coding in
  observational studies.
\newblock {\em International Journal of Public Health\/}~{\em 62}, 163--167.

\bibitem[\protect\citeauthoryear{Hastie, Tibshirani, and Friedman}{Hastie
  et~al.}{2009}]{hastie2009elements}
Hastie, T., R.~Tibshirani, and J.~Friedman (2009).
\newblock {\em The Elements of Statistical Learning}, Volume~2.
\newblock Springer.

\bibitem[\protect\citeauthoryear{Initiative}{Initiative}{2020}]{ChildrensEnvironmentalHealthInitiative2020}
Initiative, C. E.~H. (2020).
\newblock Linked births, lead surveillance, grade 4 end-of-grade (eog) scores
  [data set].

\bibitem[\protect\citeauthoryear{Johfre and Freese}{Johfre and
  Freese}{2021}]{Johfre2021}
Johfre, S.~S. and J.~Freese (2021).
\newblock Reconsidering the reference category.
\newblock {\em Sociological Methodology\/}~{\em 51}, 253--269.

\bibitem[\protect\citeauthoryear{Knol, Egger, Scott, Geerlings, and
  Vandenbroucke}{Knol et~al.}{2009}]{Knol2009}
Knol, M.~J., M.~Egger, P.~Scott, M.~I. Geerlings, and J.~P. Vandenbroucke
  (2009).
\newblock When one depends on the other: Reporting of interaction in
  case-control and cohort studies.
\newblock {\em Epidemiology\/}~{\em 20}.

\bibitem[\protect\citeauthoryear{Kowal}{Kowal}{2024}]{Kowal2024}
Kowal, D.~R. (2024).
\newblock Regression with race-modifiers: towards equity and interpretability.
\newblock {\em medRxiv\/}, 2021--2024.

\bibitem[\protect\citeauthoryear{Kowal, Bravo, Leong, Griffin, Ensor, and
  Miranda}{Kowal et~al.}{2021}]{KowalPRIME2020}
Kowal, D.~R., M.~Bravo, H.~Leong, R.~J. Griffin, K.~B. Ensor, and M.~L. Miranda
  (2021).
\newblock Bayesian variable selection for understanding mixtures in
  environmental exposures.
\newblock {\em Statistics in Medicine\/}~{\em 40}, 4850--4871.

\bibitem[\protect\citeauthoryear{Krefeld-Schwalb, Sugerman, and
  Johnson}{Krefeld-Schwalb et~al.}{2024}]{Krefeld2024}
Krefeld-Schwalb, A., E.~R. Sugerman, and E.~J. Johnson (2024).
\newblock Exposing omitted moderators: Explaining why effect sizes differ in
  the social sciences.
\newblock {\em Proceedings of the National Academy of Sciences\/}~{\em 121},
  e2306281121.

\bibitem[\protect\citeauthoryear{Lim and Hastie}{Lim and
  Hastie}{2015}]{Lim2015}
Lim, M. and T.~Hastie (2015).
\newblock Learning interactions via hierarchical group-lasso regularization.
\newblock {\em Journal of Computational and Graphical Statistics\/}~{\em 24},
  627--654.

\bibitem[\protect\citeauthoryear{Miao, Wu, and Lu}{Miao
  et~al.}{2024}]{Miao2024}
Miao, J., Y.~Wu, and Q.~Lu (2024).
\newblock Statistical methods for gene–environment interaction analysis.
\newblock {\em Wiley Interdisciplinary Reviews: Computational
  Statistics\/}~{\em 16}, e1635.

\bibitem[\protect\citeauthoryear{Park, Petkova, Tarpey, and Ogden}{Park
  et~al.}{2021}]{Park2021}
Park, H., E.~Petkova, T.~Tarpey, and R.~T. Ogden (2021).
\newblock A constrained single‐index regression for estimating interactions
  between a treatment and covariates.
\newblock {\em Biometrics\/}~{\em 77}, 506--518.

\bibitem[\protect\citeauthoryear{Park, Petkova, Tarpey, and Ogden}{Park
  et~al.}{2023}]{Park2023}
Park, H., E.~Petkova, T.~Tarpey, and R.~T. Ogden (2023).
\newblock Functional additive models for optimizing individualized treatment
  rules.
\newblock {\em Biometrics\/}~{\em 79}, 113--126.

\bibitem[\protect\citeauthoryear{Pocock, Collier, Dandreo, de~Stavola, Goldman,
  Kalish, Kasten, and McCormack}{Pocock et~al.}{2004}]{Pocock2004}
Pocock, S.~J., T.~J. Collier, K.~J. Dandreo, B.~L. de~Stavola, M.~B. Goldman,
  L.~A. Kalish, L.~E. Kasten, and V.~A. McCormack (2004).
\newblock Issues in the reporting of epidemiological studies: a survey of
  recent practice.
\newblock {\em BmJ\/}~{\em 329}, 883.

\bibitem[\protect\citeauthoryear{Scheffe}{Scheffe}{1999}]{Scheffe1999}
Scheffe, H. (1999).
\newblock {\em The analysis of variance}, Volume~72.
\newblock John Wiley \& Sons.

\bibitem[\protect\citeauthoryear{Schoendorf, Hogue, Kleinman, and
  Rowley}{Schoendorf et~al.}{1992}]{Schoendorf1992}
Schoendorf, K.~C., C.~J.~R. Hogue, J.~C. Kleinman, and D.~Rowley (1992).
\newblock Mortality among infants of black as compared with white
  college-educated parents.
\newblock {\em New England journal of medicine\/}~{\em 326}, 1522--1526.

\bibitem[\protect\citeauthoryear{Searle, Speed, and Milliken}{Searle
  et~al.}{1980}]{Searle1980}
Searle, S.~R., F.~M. Speed, and G.~A. Milliken (1980).
\newblock Population marginal means in the linear model: an alternative to
  least squares means.
\newblock {\em The American Statistician\/}~{\em 34}, 216--221.

\bibitem[\protect\citeauthoryear{Simpson}{Simpson}{1951}]{Simpson1951}
Simpson, E.~H. (1951).
\newblock The interpretation of interaction in contingency tables.
\newblock {\em Journal of the Royal Statistical Society: Series B
  (Methodological)\/}~{\em 13}, 238--241.

\bibitem[\protect\citeauthoryear{Sweeney and Ulveling}{Sweeney and
  Ulveling}{1972}]{Sweeney1972}
Sweeney, R.~E. and E.~F. Ulveling (1972).
\newblock A transformation for simplifying the interpretation of coefficients
  of binary variables in regression analysis.
\newblock {\em The American Statistician\/}~{\em 26}, 30--32.

\bibitem[\protect\citeauthoryear{Wang and Lin}{Wang and Lin}{2024}]{Wang2024}
Wang, T. and C.-W. Lin (2024).
\newblock Using a centered general linear model for detection of interactions
  among biomarkers.
\newblock {\em Statistical Methods in Medical Research\/}, 09622802231224639.

\bibitem[\protect\citeauthoryear{Williams, Lawrence, and Davis}{Williams
  et~al.}{2019}]{Williams2019}
Williams, D.~R., J.~A. Lawrence, and B.~A. Davis (2019).
\newblock Racism and health: evidence and needed research.
\newblock {\em Annual Review of Public Health\/}~{\em 40}, 105--125.

\end{thebibliography}
\doublespacing

\clearpage

\setcounter{page}{1}
\counterwithin{figure}{section}
\counterwithin{table}{section}
\counterwithin{equation}{section}

\appendix

\begin{center} \onehalfspacing
    \LARGE {\bf 
    Supplement to \\
    ``Facilitating heterogeneous effect estimation via statistically efficient categorical modifiers"}
      \if0\blind
    \large{\\ \vspace{5mm} Daniel R. Kowal}
    \fi
\end{center}

    This supplementary file includes proofs of all results (Section~\ref{sec-a-proofs}), details for generalized linear models (Section~\ref{sec-a-glm}), additional simulation results (Section~\ref{sec-a-sims}), and additional details and analyses of the North Carolina education data (Section~\ref{sec-a-app}).

\section{Proofs} \label{sec-a-proofs}

We first provide a sketch of the general proof technique. Our results require only basic linear algebra, but the notation can be cumbersome. Here, the goal is to provide clear intuition for our results and to put forth a blueprint to analyze similar invariance properties in other settings. 

Consider two generic but nested models:
\begin{quote}
   \texttt{y $\sim$ X$_*$ + X$_0$}  \\
   \texttt{y $\sim$ X$_*$ + X$_0$ + X$_1$} 
\end{quote}
The task is to establish conditions under which the OLS estimates of the coefficients on $\bm X_*$ are unchanged by the addition of $\bm X_1$, with $\bm X_0$ also present in both models.  In our typical setting,   $\bm X_*$ is a matrix of (continuous) covariates, $\bm X_0$ is a matrix of categorical indicator variables, and  $\bm X_1$  contains cat-modifiers. Crucially, for \emph{identifiable} estimation and inference, these matrices involving categorical covariates or cat-modifiers must already be parametrized to enforce the identifiable constraints, such as omitting certain columns for RGE or applying the QR reparametrization from Section~\ref{sec-est} for ABCs. 

The most relevant classical result is due to Frisch and Waugh (1933) and Lovell (1963): 

\vspace{2mm} \noindent{\bf Frisch-Waugh-Lovell (FWL) Theorem:}
    \label{thm-resid}
    For a partition of the $n \times p$ covariate matrix $\bm X = (\bm X_0:\bm X_1)$ into $p_0$ and $p_1$ columns, the partition of the ordinary least squares estimator $\bm{\hat \beta} = (\bm{\hat \beta}_0^\top, \bm{\hat \beta}_1^\top)^\top$ satisfies $\bm{\hat \beta}_0 = (\bm X_0^\top\bm E_{01})^{-1} \bm E_{01}^\top\bm y = (\bm E_{01}^\top\bm E_{01})^{-1} \bm E_{01}^\top\bm y$, where $\bm E_{01} = (\bm I_n - \bm H_{X_1}) \bm X_0$ is the $n \times p_0$ matrix of residuals from regressing each column of $\bm X_0$ on $\bm X_1$, $\bm H_{X_1} = \bm X_1 (\bm X_1^\top\bm X_1)^{-1} \bm X_1^\top$ is the corresponding hat matrix for $\bm X_1$, and $\bm y = (y_1,\ldots,y_n)^\top$ is the vector of outcomes. 
\vspace{2mm} 

Applying the FWL Theorem, our target result occurs when \texttt{residuals(X$_*$ $\sim$  X$_0$)} = \texttt{residuals(X$_*$ $\sim$  X$_0$ + X$_1$)}, for which a sufficient condition is $\bm X_*^\top \bm E_{10} = 0$
where $\bm E_{10} = $ \texttt{residuals(X$_1$ $\sim$  X$_0$)}. More formally, let $\bm{H}_0 \coloneqq \bm X_0(\bm X_0^\top \bm X_0)^{-1}\bm X_0^\top$ be the hat matrix for the covariates $\bm X_0$ that are always included. Then the sufficient condition is 
\begin{equation}\label{cond-a}
    \bm X_*^\top (\bm X_1 - \bm H_0 \bm X_1) = \bm 0
\end{equation}
or equivalently, $\bm X_*^\top \bm X_1 - (\bm H_0\bm X_*)^\top \bm X_1 = \bm 0$, if we prefer to consider regressing $\bm X_0$ on $\bm X_*$ instead of $\bm X_1$. 
 In the simpler case without a common $\bm X_0$ term, the requirement simplifies to $\bm X_*^\top \bm X_1 = \bm 0$, where the role orthogonality is now abundantly clear. 

In the presence of ABCs (or other linear constraints), we apply the reparametrization from Section~\ref{sec-est} that replaces $\bm X_1$ with $\bm X_1 \bm Q_{\bm{\hat \pi}}$ to enforce the constraints. The main condition \eqref{cond-a} is now 
\begin{equation} \label{cond-q-a}
\bm X_*^\top (\bm X_1 - \bm H_0 \bm X_1)\bm Q_{\bm{\hat \pi}} = \bm 0.
\end{equation}
The key observation is that   $\bm A_{\bm{\hat \pi}} \bm Q_{\bm{\hat \pi}} = \bm 0$ by construction; this is true for the QR-based approach with \emph{any} constraints of the form $\bm A_{\bm{\hat \pi}} \bm \theta = \bm 0$, including but not limited to ABCs. Thus, the general  requirement is to show that $
\bm X_*^\top (\bm X_1 - \bm H_0 \bm X_1)$ is  row-wise proportional to $\bm A_{\bm{\hat \pi}}$, which produces the necessary zeros. 

We apply this strategy for Theorems~\ref{thm-cts}--\ref{thm-se}, but prove the main results in sequence.

\begin{proof}[Proof (Lemma~\ref{lemma-int})]
    For simplicity, we prove this result for the case of \eqref{reg-cm-cat}, but the same ideas apply more generally. 
    It is sufficient to show that  $\mathbb{E}_{\bm{\hat \pi}}(\beta_{1,R} + \beta_{2,S} + \gamma_{RS}) = 0$. Direct application of \eqref{abcs-gen} implies $\mathbb{E}_{\bm{\hat \pi}}(\beta_{1,R} + \beta_{2,S} + \gamma_{RS}) = \mathbb{E}_{\bm{\hat \pi}_R}(\beta_{1,R}) + \mathbb{E}_{\bm{\hat \pi}_S}(\beta_{2,S}) + \mathbb{E}_{\bm{\hat \pi}}(\gamma_{RS}) = \mathbb{E}_{\bm{\hat \pi}}(\gamma_{RS})$, and further simplifying, $\mathbb{E}_{\bm{\hat \pi}}(\gamma_{RS}) = \sum_{r=1}^{L_R}\sum_{s=1}^{L_S} \hat \pi_{rs} \gamma_{rs} = 0$ since the internal summation is zero for all $r$ by \eqref{abcs-cat-cat-gen}.
\end{proof}

\begin{proof}[Proof (Lemma~\ref{lemma-cat-cat})]
    We prove this result for the case of \eqref{reg-cm-cat} for simplicity. Applying \eqref{abcs-cat-cat-gen} to all but $r=1$, we have $\sum_{r=1}^{L_R} \hat \pi_{rs} \gamma_{rs} = 0$ for $s=1,\ldots,L_S$ and thus $\gamma_{1s} = -\hat\pi_{1s}^{-1} \sum_{r=2}^{L_R}\hat\pi_{rs} \gamma_{rs}$. The conditional expectation is then 
    $
        \mathbb{E}_{\bm{\hat \pi}_{S \mid R=1}}(\gamma_{RS}) = \sum_{s=1}^{L_S} \hat \pi_{1s} \gamma_{1s}
        = -\sum_{s=1}^{L_S} 
        \sum_{r=2}^{L_R}\hat\pi_{rs} \gamma_{rs}
        =\sum_{r=2}^{L_R} \sum_{s=1}^{L_S} 
        \hat\pi_{rs} \gamma_{rs}
        =0 
    $ 
    since the internal summation equals zero for all $r>1$. 
\end{proof}

\begin{proof}[Proof (Theorem~\ref{thm-intercept})]
    Under OLS, $\bar y$ equals the sample mean of the fitted values $\{\hat y_i\}_{i=1}^n$; this is true for ABCs, RGE, STZ, etc. Then we simplify: 
    \begin{align*}
    \bar y &= n^{-1} \sum_{i=1}^n \hat y_i = n^{-1} \sum_{i=1}^n(\hat \alpha_0 + \bm x_i^\top \bm{\hat \alpha} + \sum_{k=1}^K \hat\beta_{k, c_k} + \sum_{k=1}^{K-1} \sum_{k'=k+1}^K \hat\gamma_{k,k', c_k, c_{k'}} )\\
    &= \hat \alpha_0 + \bm{\bar x}^\top \bm{\hat \alpha} +  \sum_{k=1}^K \sum_{c_k =1}^{L_k} \hat \pi_{k, c_k} \hat \beta_{k, c_k} + \sum_{k=1}^{K-1} \sum_{k'=k+1}^K \sum_{c_k =1}^{L_k} \sum_{c_{k'} =1}^{L_{k'}}\hat \pi_{k,k', c_k, c_{k'}}  \hat \gamma_{k,k', c_k, c_{k'}}  \\
    &= \hat \alpha_0
    \end{align*}
    since the continuous covariates are centered ($\bm{\bar x} = \bm 0$) and the main categorical effects and categorical-categorical interactions satisfy ABCs, so the interior summations equal zero for all $k, k'$. 
\end{proof}

\begin{proof}[Proof (Theorem~\ref{thm-cat-cat})]
    Following the \texttt{race} and \texttt{sex} terminology from Example~\ref{ex:cat},  define the design matrix by letting $\bm 1$ be an $n$-dimensional vector of ones, $\bm Z_1$ the $n \times L_R$ matrix of \texttt{race}  indicators with  entries $[\bm Z_1]_{ir} = 1$ if $r_i = r$ and zero otherwise, and $\bm Z_2$ $n \times L_S$ matrix of  \texttt{sex} indicators  with entries $[\bm Z_2]_{is} = 1$ if $s_i = s$ and zero otherwise. Similarly, let $\bm Z_{12}$ be the $n \times L_RL_S$ matrix with indicators for the interaction terms. Consider the cross-produces of each main effect with the interaction matrix. 
    First, $\bm 1^{\top} \bm Z_{12}$ is the $1 \times L_RL_S$ matrix where each entry is the joint total by \texttt{race} and \texttt{sex}, i.e., $\sum_{i=1}^n \mathbb{I}(r_i = r, s_i = s)$ for each $r,s$ combination. 
    Next, $\bm Z_1^\top \bm Z_{12}$ is $L_R \times L_RL_S$, where each row $r$ includes the totals  $\sum_{i=1}^n \mathbb{I}(r_i = r, s_i = s)$ for all $s=1,\ldots,L_S$ but zeros for columns with other \texttt{race} groups, $r' \ne r$. Similarly, $\bm Z_2^\top \bm Z_{12}$ is $L_S \times L_RL_S$, where each row $s$ includes the totals  $\sum_{i=1}^n \mathbb{I}(r_i = r, s_i = s)$ for all $r=1,\ldots,L_R$ but  zeros for columns with other \texttt{sex} groups, $s' \ne s$. 

    Estimation invariance occurs when these cross-products are zero. However, we must also account for identifiability constraints. Following Section~\ref{sec-est},   $\bm Z_{12}$ is replaced by $\bm Z_{12}\bm Q_{\bm{\hat\pi}}$, where $\bm A_{\bm{\hat\pi}} \bm Q_{\bm{\hat\pi}} = \bm 0$ and  $\bm A_{\bm{\hat\pi}}$ encodes the constraints on the interaction coefficients. Thus, it suffices to show that $\bm 1^{\top} \bm Z_{12} \bm Q_{\bm{\hat\pi}} = \bm 0$, $\bm Z_1^{\top} \bm Z_{12} \bm Q_{\bm{\hat\pi}} = \bm 0$, and $\bm Z_2^{\top} \bm Z_{12} \bm Q_{\bm{\hat\pi}} = \bm 0$, with each zero of the appropriate dimension. 
    For ABCs, the latter two cross-products, when scaled by $n^{-1}$, exactly match the joint ABCs 
    \eqref{abcs-cat-cat-gen} in the form of $\bm A_{\bm{\hat\pi}}$, and thus are zero upon post-multiplication by $\bm Q_{\bm{\hat\pi}} $. Similarly, the first cross-product is also zero by applying the arguments from Lemma~\ref{lemma-int}.
\end{proof}

\begin{proof}[Proof (Theorem~\ref{thm-cts})]
    Let  $\bm y = (y_1,\ldots,y_n)^\top$, 
    $\bm x = (x_1,\ldots, x_n)^\top$, and $\bm Z$ be the matrix of categorical (\texttt{race}) indicators with entries $[\bm Z]_{ir} = 1$ if $r_i = r$ and zero otherwise. The cat-modifier term is $\bm Z_X = \bm D_X \bm Z$ and $\bm D_X = \mbox{diag}(\bm x)$. The goal is to show that, under the stated conditions, \eqref{cond-q-a} holds with $\bm x = \bm X_*$, $\bm X_1 = \bm Z_X$, and $\bm H_0  = \bm H_Z = \bm Z(\bm Z^\top\bm Z)^{-1} \bm Z^\top$ is the hat matrix for the categorical covariate. 
    
    For clarity, we provide more detailed results en route. Applying the FWL Theorem, the estimated coefficients under \eqref{reg-main-x} satisfy 
    $ 
        \hat \alpha_1^M = (\bm x^\top \bm {\hat e}_{x\sim r})^{-1}\bm {\hat e}_{x\sim r}^\top\bm y,
    $ 
    where  $\bm {\hat e}_{x\sim r}$ is the vector of residuals from regressing the continuous variable $\{x\}_{i=1}^n$ on the categorical variable $\{r_i\}_{i=1}^n$ (i.e., $\bm Z$). Similarly, the estimated coefficients under \eqref{reg-cm-x} satisfy 
    $ 
        \hat \alpha_1 = (\bm x^\top \bm {\hat e}_{x\sim r + Z_{XQ}})^{-1}\bm {\hat e}_{x\sim r + Z_{XQ}}^\top\bm y,
    $ 
     where $\bm {\hat e}_{x\sim r + Z_{XQ}}$ are the residuals from regressing the continuous variable $\{x\}_{i=1}^n$ on the categorical variable $\{r_i\}_{i=1}^n$ (i.e., $\bm Z$) \emph{and} the reparametrized interaction term that enforces  ABCs,  $\bm Z_{XQ} = \bm Z_X \bm Q_{-(1:m)}$ (see Section~\ref{sec-est}). Thus, it suffices to show that $\bm {\hat e}_{x\sim r} = \bm {\hat e}_{x\sim r + Z_{XQ}}$, which occurs when the additional (interaction) coefficients from the latter model, say $\bm{\hat b}_{Z_{XQ}}$ (corresponding to $\bm Z_{XQ}$), are identically zero. Again using  the FWL Theorem, these estimated coefficients are $\bm{\hat b}_{Z_{XQ}} =  \bm Q_{-(1:m)} (\bm Z_{XQ}^\top\bm E_{Z_{XQ}})^{-1} \bm E_{Z_{XQ}}^\top \bm x$, where $\bm E_{Z_{XQ}}$ is the matrix of residuals from regressing $\bm Z_{XQ}$ on $\bm Z$, i.e., $\bm E_{Z_{XQ}} = \bm Z_{XQ} - \bm H_Z \bm Z_{XQ}$. Thus, showing $\bm x^\top \bm E_{Z_{XQ}} = \bm 0$ is sufficient, and factoring $\bm x^\top \bm E_{Z_{XQ}} = (\bm x^\top \bm Z_X - \bm x^\top\bm H_Z \bm Z_X) \bm Q_{-(1:m)}$ shows the connection with \eqref{cond-q-a}.
    
     First, observe that  $\bm x^\top \bm Z_X = \bm x^\top \bm D_X \bm Z = (s_{x[1]}^2,\ldots, 
     s_{x[L_R]}^2)$ is the vector of $s_{x[r]}^2 = \sum_{r_i=r} x_i^2$ across groups. Next, observe that $(\bm Z^\top \bm Z)^{-1} \bm Z^\top \bm Z_X = \mbox{diag}(\{\bar x_r\}_r)$ contains the sample means of $\{x\}_{i=1}^n$ by each group $r$, and therefore $\bm x^\top\bm H_Z \bm Z_X = \bm x^\top \bm Z \mbox{diag}(\{\bar x_r\}_r) = (n_1 \bar x_1^2, \ldots, n_{L_R} \bar x_{L_R}^2)$ with $n_r = n \hat \pi_r$. Combining these results, we have $\bm x^\top \bm E_{Z_{XQ}} = \bm v^\top \bm Q_{-(1:m)}$, where $\bm v^\top = (s_1^2 - n_1 \bar x_1^2,\ldots, s_{L_R}^2 - n_{L_R} \bar x_{L_R}^2) = n ( \hat \pi_1 \hat \sigma_{x[1]}^2, \ldots,\hat \pi_{L_R} \hat \sigma_{x[L_R]}^2) = n \hat \sigma_{x[1]}^2 \bm{\hat \pi}^\top$ under the assumption that $\hat \sigma_{x[r]}^2 = \hat \sigma_{x[1]}^2$ is common for all $r$, which is precisely the equal-variance condition \eqref{eq-v}. Finally, the definition of $\bm Q_{-(1:m)}$ via  ABCs implies that $\bm{\hat \pi}^\top\bm Q_{-(1:m)} = \bm 0$, which proves the result.
\end{proof}

\begin{proof}[Proof (Theorem~\ref{thm-int-full})]
Let $\bm X$ denote the $n \times p$ matrix of continuous covariates, $\bm Z$  the matrix of categorical dummy variables with entries $[\bm Z]_{ir} = 1$ if $r_i = r$ and zero otherwise, and $\bm Z_{XQ} = (\bm Z_{X_1Q},\ldots, \bm Z_{X_pQ})$ with $\bm Z_{X_jQ} = \bm Z_{X_j} \bm Q_{-(1:m)}$, 
 $\bm Z_{X_j} = \bm D_{X_j} \bm Z$, and $\bm D_{X_j} = \mbox{diag}(\bm x_j)$. By the FWL Theorem, it suffices to show that $\bm E_M = \bm E$, where $\bm E_M = (\bm I_n - \bm H_Z)\bm X$ are the residuals from regressing each column of $\bm X$ on $\bm Z$ and $\bm E$ are similarly the residuals from regressing each column of $\bm X$ on $\bm Z$ \emph{and} $\bm Z_{XQ}$. Thus, it is sufficient to show that the coefficients associated with $\bm Z_{XQ}$ in the latter regression are identically zero. Again using the FWL Theorem, we see that this occurs whenever $\bm X^\top \bm E_{Z_{XQ}} = \bm 0$, where $\bm E_{Z_{XQ}} = \bm Z_{XQ} - \bm H_Z \bm Z_{XQ} = (\bm Z_{X_1Q} - \bm H_Z \bm Z_{X_1Q}, \ldots, \bm Z_{X_pQ} - \bm H_Z \bm Z_{X_pQ})$. Noticing that $\bm X^\top \bm E_{Z_{XQ}} = (\bm X^\top (\bm Z_{X_1Q} - \bm H_Z \bm Z_{X_1Q}), \ldots, \bm X^\top(\bm Z_{X_pQ} - \bm H_Z \bm Z_{X_pQ}))$, we consider the individual components $\bm x_h^\top(\bm Z_{X_jQ} - \bm H_Z \bm Z_{X_jQ}) = (\bm x_h^\top \bm D_{x_j} \bm Z - \bm x_h^\top \bm H_Z \bm D_{x_j} \bm Z) \bm Q_{-(1:m)}$, each of which must equal the zero vector with dimension equal to the number of categories. Noting that $\bm x_h^\top \bm D_{x_j} \bm Z = (\ldots, s_r(j,h),\ldots)$ with $s_r(j,h) = \sum_{r_i=r} x_{ij}x_{ih}$ and $\bm x_h^\top \bm H_Z \bm D_{x_j} \bm Z = (\ldots, n_r \bar x_r(j) \bar x_r(h), \ldots)$ with $\bar x_r(j) = n_r^{-1} \sum_{r_i=r} x_{ij}$, we apply the same arguments as in Theorem~\ref{thm-cts}. 
\end{proof}

\begin{proof}[Proof (Theorem~\ref{thm-int-local})]
    Let $\bm X_0$ be the matrix of covariates that includes $\bm X_{-1}$ (i.e., all covariates but $\bm x_1$) and the categorical (race) indicators $\bm Z$ with entries $[\bm Z]_{ir} = 1$ if $r_i = r$ and zero otherwise, and let $\bm{H}_0 \coloneqq \bm X_0(\bm X_0^\top \bm X_0)^{-1}\bm X_0^\top$ be its hat matrix. For the interaction terms, let $\bm Z_{x_1Q} = \bm Z_{x_1} \bm Q_{\bm{\hat\pi}}$ where $\bm Z_{x_1} = \bm D_{x_1} \bm Z$, $\bm D_{x_1} = \mbox{diag}(\bm x_1)$, and $\bm{\hat\pi}^\top\bm Q_{\bm{\hat\pi}} = \bm 0$ enforces the ABCs for the interaction terms (see Section~\ref{sec-est}). Now, it is sufficient to show that $(\bm x_1^\top \bm Z_{x_1} - \bm x_1^\top \bm H_0 \bm Z_{x_1})\bm Q_{\bm{\hat\pi}} = \bm 0$ as in \eqref{cond-q-a}. First, observe that $\bm x_1^\top \bm Z_{x_1} = (\cdots s_{x_1[r]}^2 \cdots )$. Second, $\bm x_1^\top \bm H_0 \bm Z_{x_1} = \bm{\hat x}_1^\top \bm D_{x_1} \bm Z = (\cdots \sum_{r_i = r} \hat x_{i1} x_{i1} \cdots )$ where $\bm{\hat x}_1 =  \bm H_0 \bm x$. Combining these results, we see that $\bm v^\top \coloneqq \bm x_1^\top \bm Z_{x_1} - \bm x_1^\top \bm H_0 \bm Z_{x_1} = (\cdots \sum_{r_i = r} (x_{i1}^2 - x_{i1} \hat x_{i1}) \cdots)$. Consider the interior terms for each $r$:  $\sum_{r_i = r} (x_{i1}^2 - x_{i1} \hat x_{i1}) =  \sum_{r_i = r} x_{i1} \hat e_{i1} = n_r\widehat{\mbox{Cov}}_r(\bm{\hat e}_1, \bm{x}_1)$, where the latter equality holds because $\sum_{r_i=r} \hat e_{i1} = 0$ for each $r$ due to the inclusion of $\bm Z$. Thus, $\bm v^\top = (\cdots n_r\widehat{\mbox{Cov}}_r(\bm{\hat e}_1, \bm{x}_1) \cdots ) = k(\cdots n_r \cdots)$ for some constant $k$ that does not depend on $r$, which implies that $\bm v^\top \bm Q_{\bm{\hat\pi}} = nk \bm{\hat \pi}^\top \bm Q_{\bm{\hat\pi}} = \bm 0$ under ABCs.
\end{proof}

\begin{proof}[Proof (Theorem~\ref{thm-cat-cat-se})]
    The proof of Theorem~\ref{thm-cat-cat} establishes orthogonality of ABCs-constrained interaction to the main effects under the same conditions. Given this orthogonality, the remainder of the proof follows the proof of Theorem~\ref{thm-se} and is omitted for brevity. 
\end{proof}

\begin{proof}[Proof (Theorem~\ref{thm-se})]
    Applying the same arguments from the proof of Theorem~\ref{thm-cts}, the variances satisfy 
        $\mbox{Var}(\hat\alpha_1^M) = \sigma_M^2 (\bm x^\top \bm {\hat e}_{x\sim r})^{-1}$ and 
        $\mbox{Var}(\hat\alpha_1) = \sigma^2 (\bm x^\top \bm {\hat e}_{x\sim r + Z_{XQ}})^{-1}$, 
    where $\sigma_M^2$ is the error variance from the main-only model and $\sigma^2$ is the error variance from the cat-modified model, assuming uncorrelated and homoskedastic errors in both models. These error assumptions are \emph{not} required to prove the result, but do motivate the definition of the SE. Under the equal-variance condition \eqref{eq-v}, we previously showed  that $\bm {\hat e}_{x\sim r} = \bm {\hat e}_{x\sim r + Z_{XQ}}$. Thus, the only difference in the variances of the estimators occurs because of the error variances, i.e., 
    $\mbox{Var}(\hat\alpha_1)/\mbox{Var}(\hat\alpha_1^M)  =  \sigma^2/\sigma_M^2$. The SEs  substitute point estimates for $\sigma_M$ and $\sigma$: $\mbox{SE}(\hat \alpha_1^M) = \hat S_M \sqrt{(\bm x^\top \bm {\hat e}_{x\sim r})^{-1}}$ and similarly,
        \begin{align*}
        \mbox{SE}(\hat \alpha_1) &= \hat S \sqrt{(\bm x^\top \bm {\hat e}_{x\sim r + Z_{XQ}})^{-1}}\\ 
        &= \hat S \sqrt{(\bm x^\top \bm {\hat e}_{x\sim r})^{-1}}\\ 
        &= \frac{\hat S}{\hat S_M} \mbox{SE}(\hat \alpha_1^M) \\
        & \le \mbox{SE}(\hat \alpha_1^M)
    \end{align*}
    since $\hat S \le \hat S_M$ under \eqref{rv}. 
\end{proof}

\section{Generalized linear models with ABCs} 
\label{sec-a-glm}
Generalized linear models (GLMs) are immensely useful for regression analysis with a variety of data types, including continuous, count, binary, and categorical data. Broadly, GLMs require a choice of data distribution (e.g., Gaussian, Poisson, Bernoulli, etc.) and a link function $g$ that replaces the expectation of $Y$, say $\mu(\bm x, \bm c)$  with a transformed version, say $g\{\mu(\bm x, \bm c)\}$, in the cat-modified model \eqref{reg-cm} (similarly for the main-only model \eqref{reg-main}). With categorical covariates and cat-modifiers, identification constraints are needed for the regression coefficients exactly as in the ordinary (untransformed) linear model. ABCs again provide a suitable identification strategy, with straightforward estimation and inference: the loss function $\mathcal{L}$ in Section~\ref{sec-est}  is specified to incorporate the appropriate negative log-likelihood and link function. 

More subtly, the presence of the link function $g$ implies that interpretations of the coefficients will be different from those in the  ordinary linear model (Section~\ref{sec-abcs}). For the main $x_j$-effects,  recall that $\alpha_j = \mathbb{E}_{\bm{\hat \pi}} (\alpha_j + \sum_{k=1}^K \gamma_{j, k, C_k})$ under ABCs \eqref{abcs-gen}, regardless of the data distribution or the link function. When the link $g$ is \emph{not} the identity,  it is no longer the case that the internal quantity in the expectation equals $\mu_{x_j}'(\bm C)$, and thus the previous interpretation from \eqref{joint-slopes-avg} requires modifications. Most generally, cat-modified GLMs satisfy
\begin{equation}\label{joint-slopes-glm-0}
g\{\mu(x_j+1, \bm x_{-j}, \bm c)\} - g\{\mu(x_j, \bm x_{-j}, \bm c)\} = \alpha_j + \sum_{k=1}^K \gamma_{j, k, c_k},
\end{equation}
so under ABCs \eqref{abcs-gen} the main $x_j$-effect is 
\begin{equation}\label{joint-slopes-glm}
    \alpha_j = \mathbb{E}_{\bm{\hat \pi}} [
    g\{\mu(x_j+1, \bm x_{-j}, \bm C)\} - g\{\mu(x_j, \bm x_{-j}, \bm C)\}].
\end{equation}
As with ordinary linear regression, ABCs identify each main effect as a group-averaged comparison between expectations at $(x_j+1, \bm x_{-j})$ and $(x_j, \bm x_{-j})$. The main differences for GLMs is the presence of the link function $g$ within this comparison. 

For clarity, we provide interpretations for logistic and Poisson regression. For binary data $Y \in\{0,1\}$, $\mu(\bm x, \bm c)$ equals the probability that $Y=1$ and logistic regression specifies $g$ as the logit link, $g(t) = \log\{t/(1-t)\}$. Now, \eqref{joint-slopes-glm-0} simplifies to the log-odds-ratio: 
\[
g\{\mu(x_j+1, \bm x_{-j}, \bm c)\} - g\{\mu(x_j, \bm x_{-j}, \bm c)\} = \log\left[\frac{\mbox{odds}\{\mu(x_j+1, \bm x_{-j}, \bm c)\}}{\mbox{odds}\{\mu(x_j, \bm x_{-j}, \bm c)\}}\right]
\]
where $\mbox{odds}\{\mu(\bm x, \bm c)\} = \mu(\bm x, \bm c)/\{1 - \mu(\bm x, \bm c)\}$. Thus, $\alpha_j$ is the group-averaged log-odds-ratio for $x_j$. This interpretation is natural: for the main-only logistic regression model, $\alpha_j$ is simply the log-odds-ratio for $x_j$. 
Similarly, for Poisson regression with $Y\in\{0,1,\ldots,\}$, $\mu(\bm x, \bm c)$ is the expectation of $Y$ and $g(t) = \log(t)$, so \eqref{joint-slopes-glm-0} is a log-ratio:
\[
g\{\mu(x_j+1, \bm x_{-j}, \bm c)\} - g\{\mu(x_j, \bm x_{-j}, \bm c)\} = \log\left\{\frac{\mu(x_j+1, \bm x_{-j}, \bm c)}{\mu(x_j, \bm x_{-j}, \bm c)}\right\}.
\]
Here, $\alpha_j$ is the group-averaged log-ratio. Finally, we note that both of these terms involve group-averaged quantities on the log-scale. Thus, it may be more natural to consider exponentiated versions on the $\mu$-scale, so the group-averages become weighted geometric means.


\section{Additional simulation results}\label{sec-a-sims}

\begin{figure}[h]
\centering
\includegraphics[width=.49\linewidth]{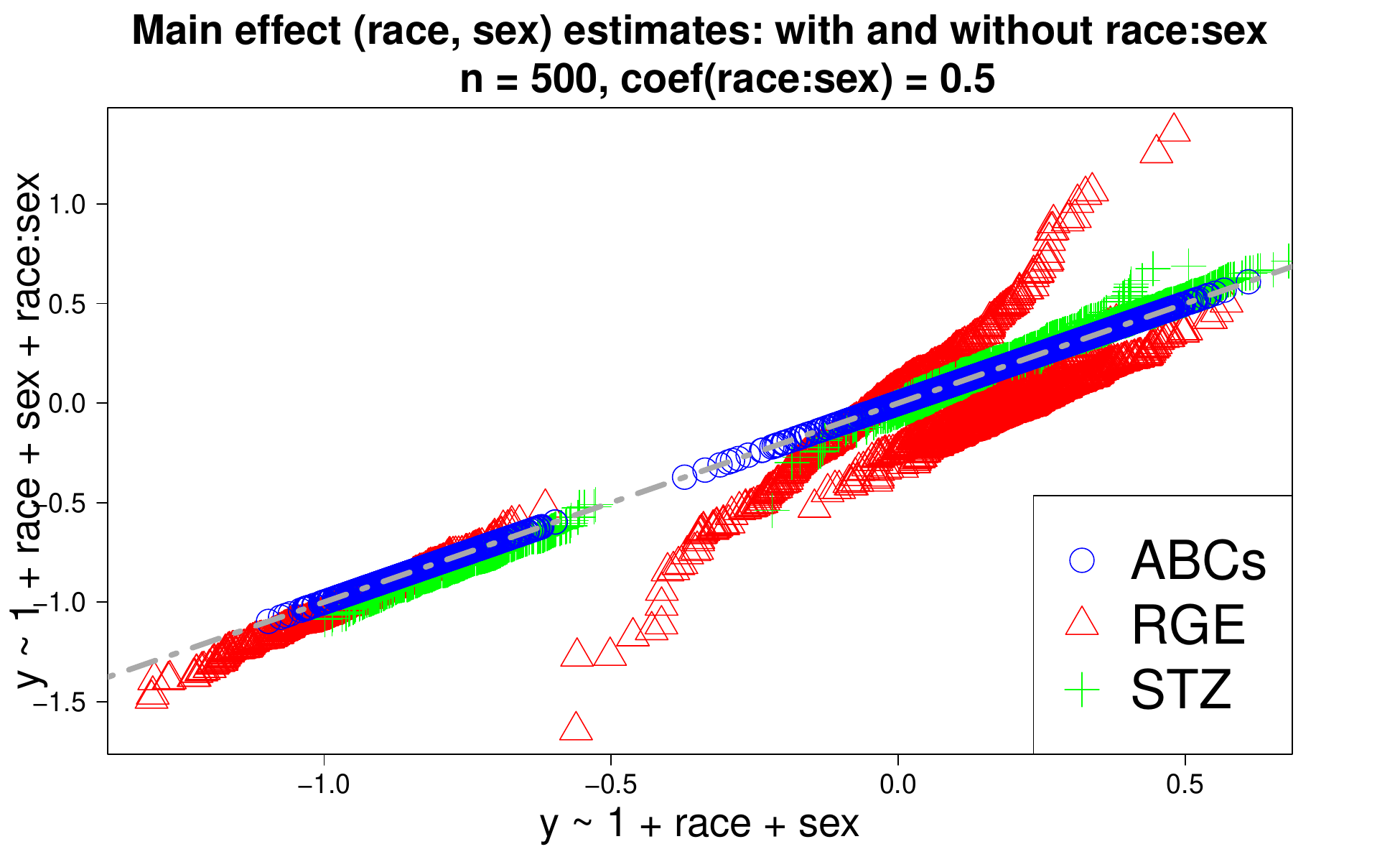}
\includegraphics[width=.49\linewidth]{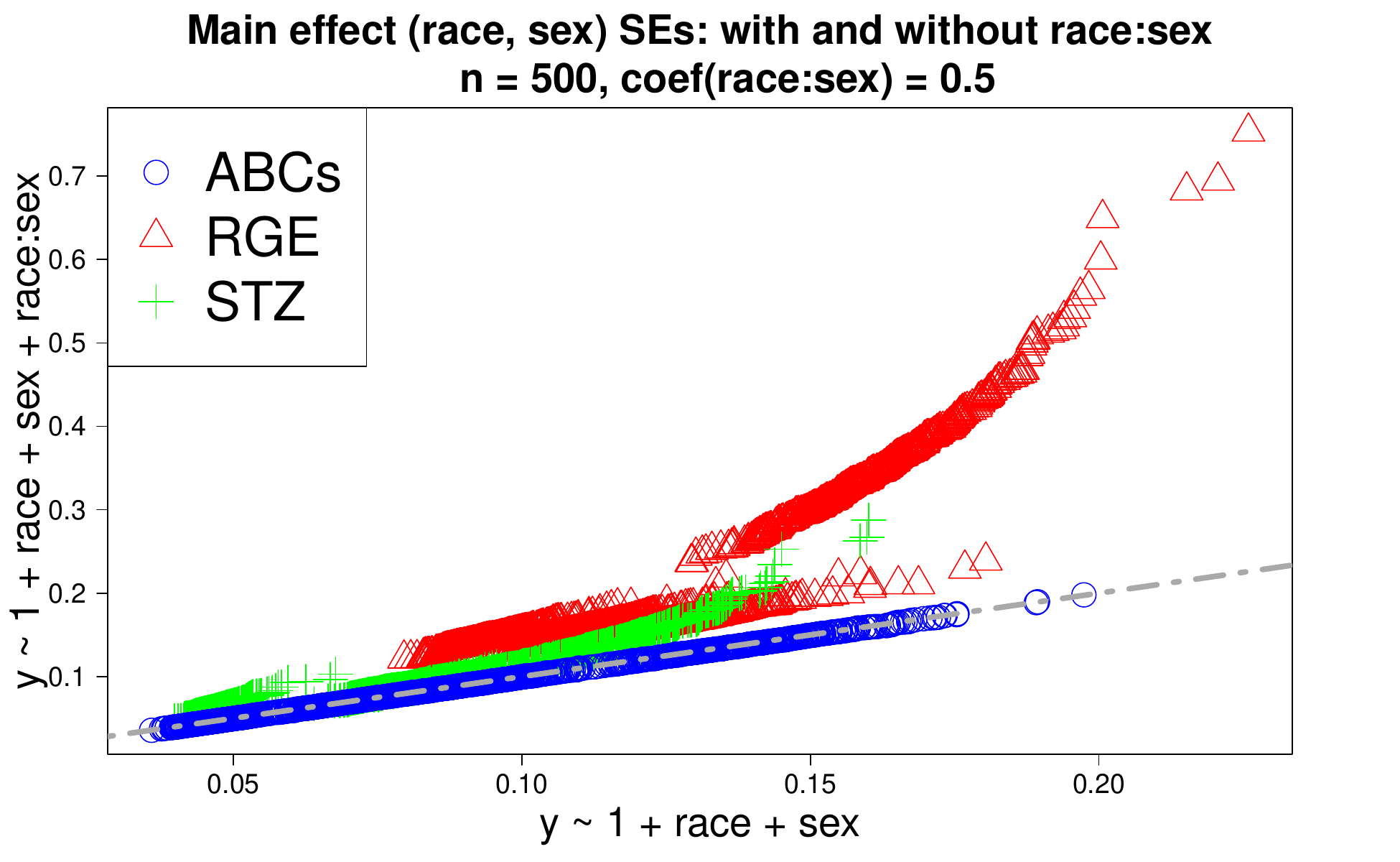}
\caption{\small Estimates (left) and standard errors (SEs, right) for all \texttt{race} and \texttt{sex} main effects 
for models that do (y-axis) and do not (x-axis) include the \texttt{race:sex} interaction across 500 simulated datasets.  Here, the interaction effect is moderate ($\gamma = 0.5$). Under ABCs, the estimates are exactly invariant and the SEs are nearly invariant ($45^\circ$ line).}
\label{fig:sim-race-sex-add}
\end{figure}

\begin{figure}[h]
\centering
\includegraphics[width=.49\linewidth]{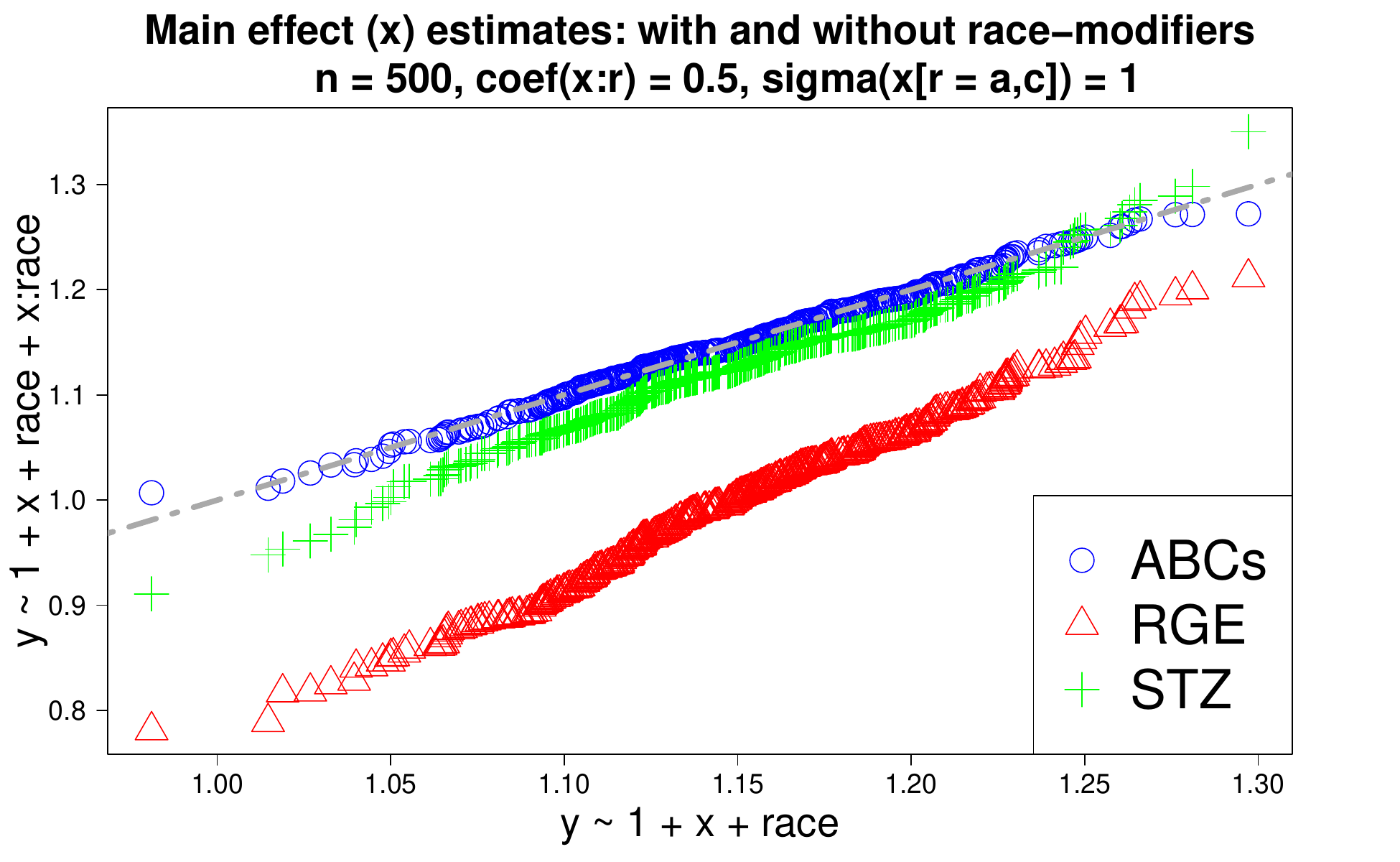}
\includegraphics[width=.49\linewidth]{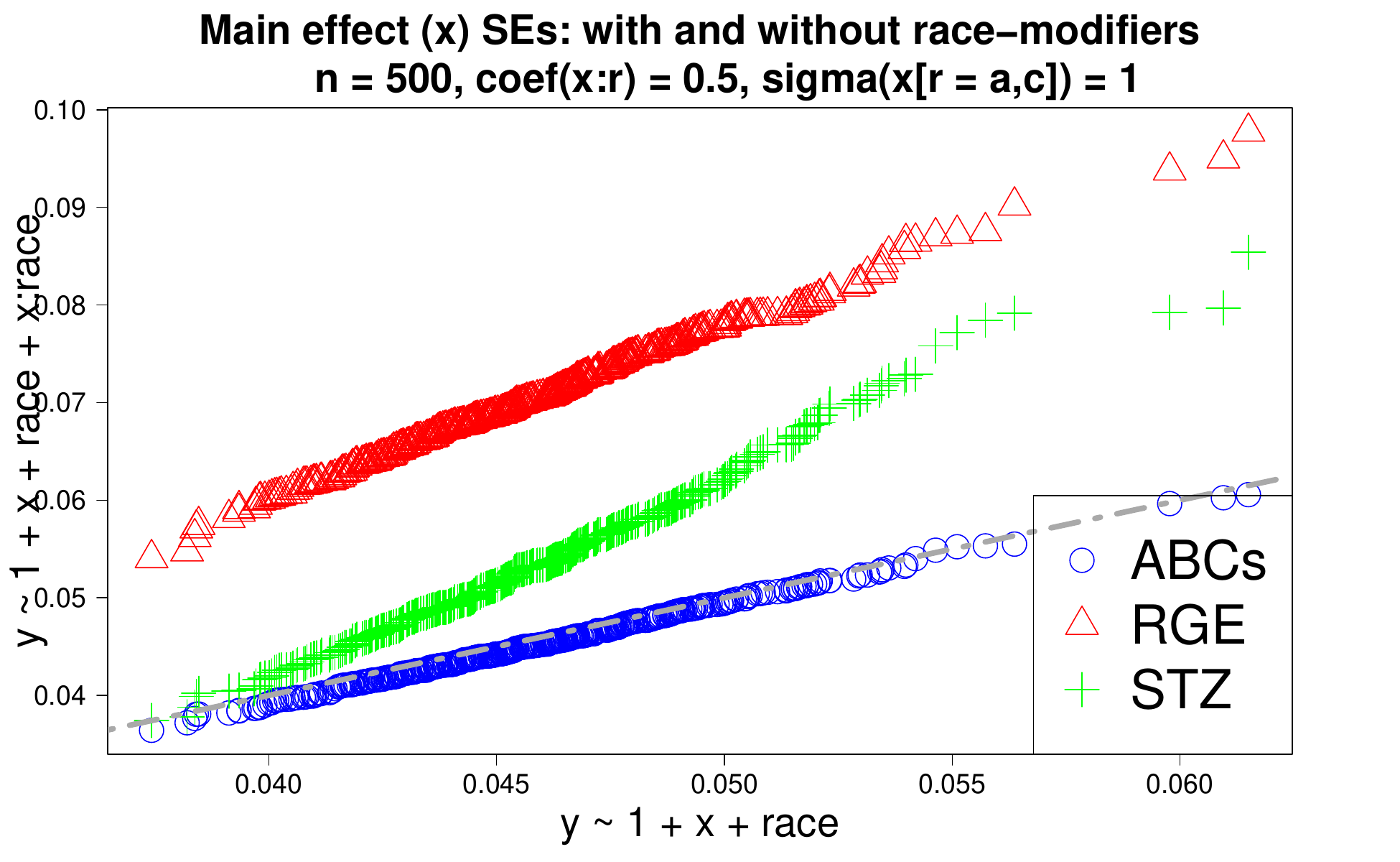}
\includegraphics[width=.49\linewidth]
{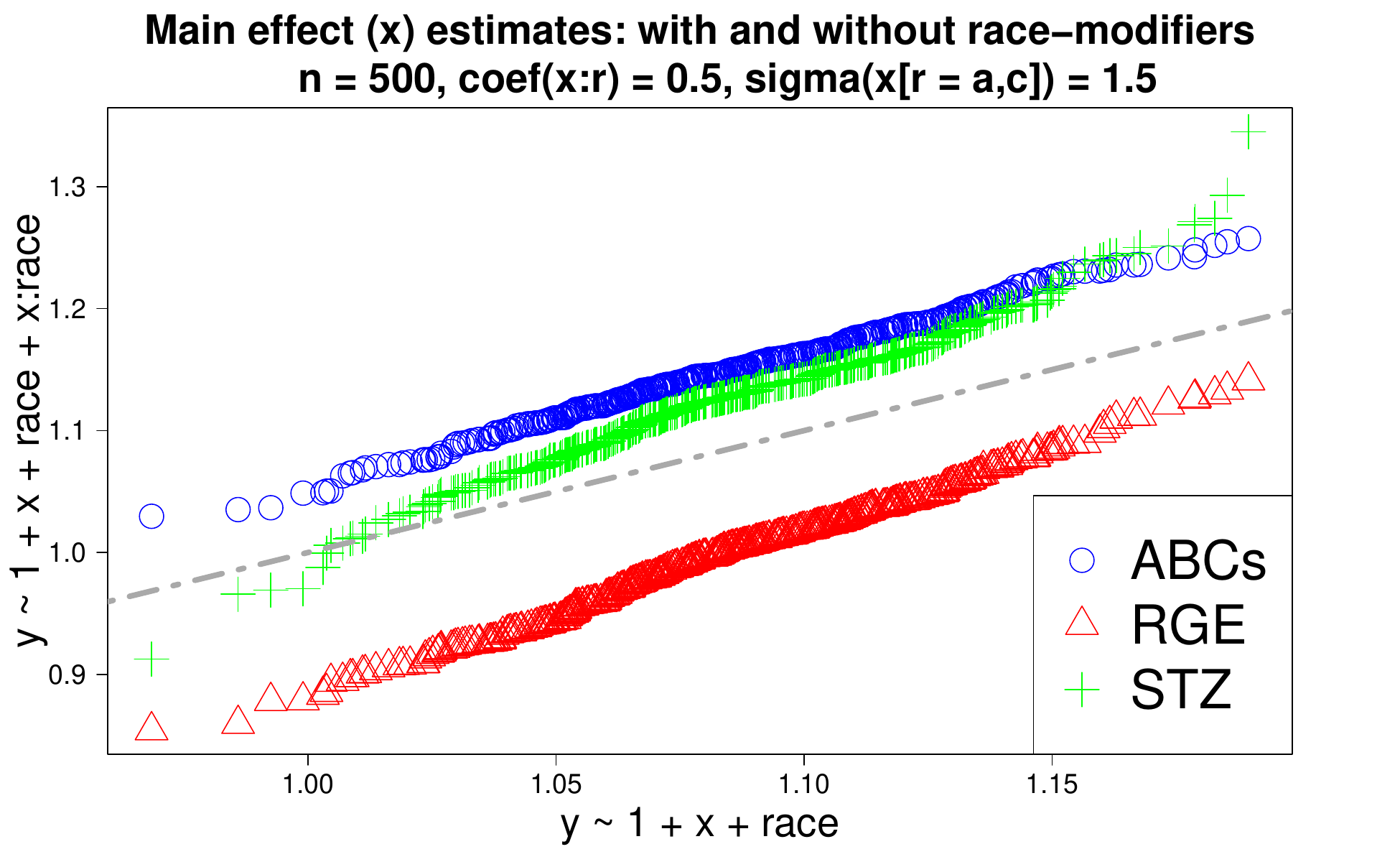}
\includegraphics[width=.49\linewidth]
{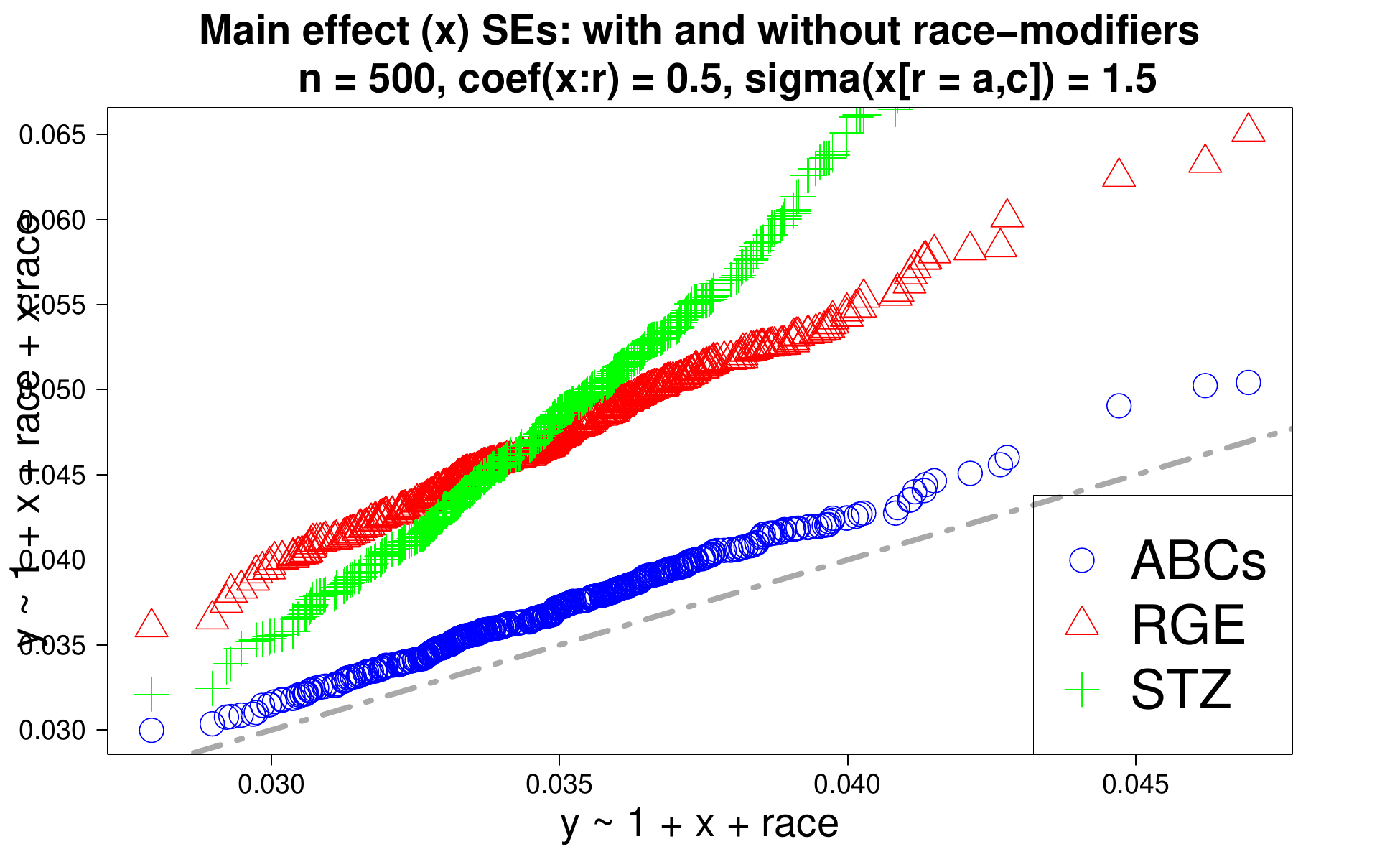}
\caption{\small Estimates (left) and standard errors (SEs, right) for the main $x$-effect for models that do (y-axis) and do not (x-axis) include the \texttt{x:race} interaction across 500 simulated datasets. Here, the interaction effect is moderate ($\gamma = 0.5$) in all cases. Under ABCs, the estimates and SEs are nearly invariant ($45^\circ$ line) as long as the deviations from equal-variance \eqref{eq-v} are  mild ($\sigma_{ac} = 1$, top). These effects are not assured when \eqref{eq-v} is strongly violated (bottom).  
}
\label{fig:sim-x-race-add}
\end{figure}

\begin{figure}[h]
\centering
\includegraphics[width=.49\linewidth]{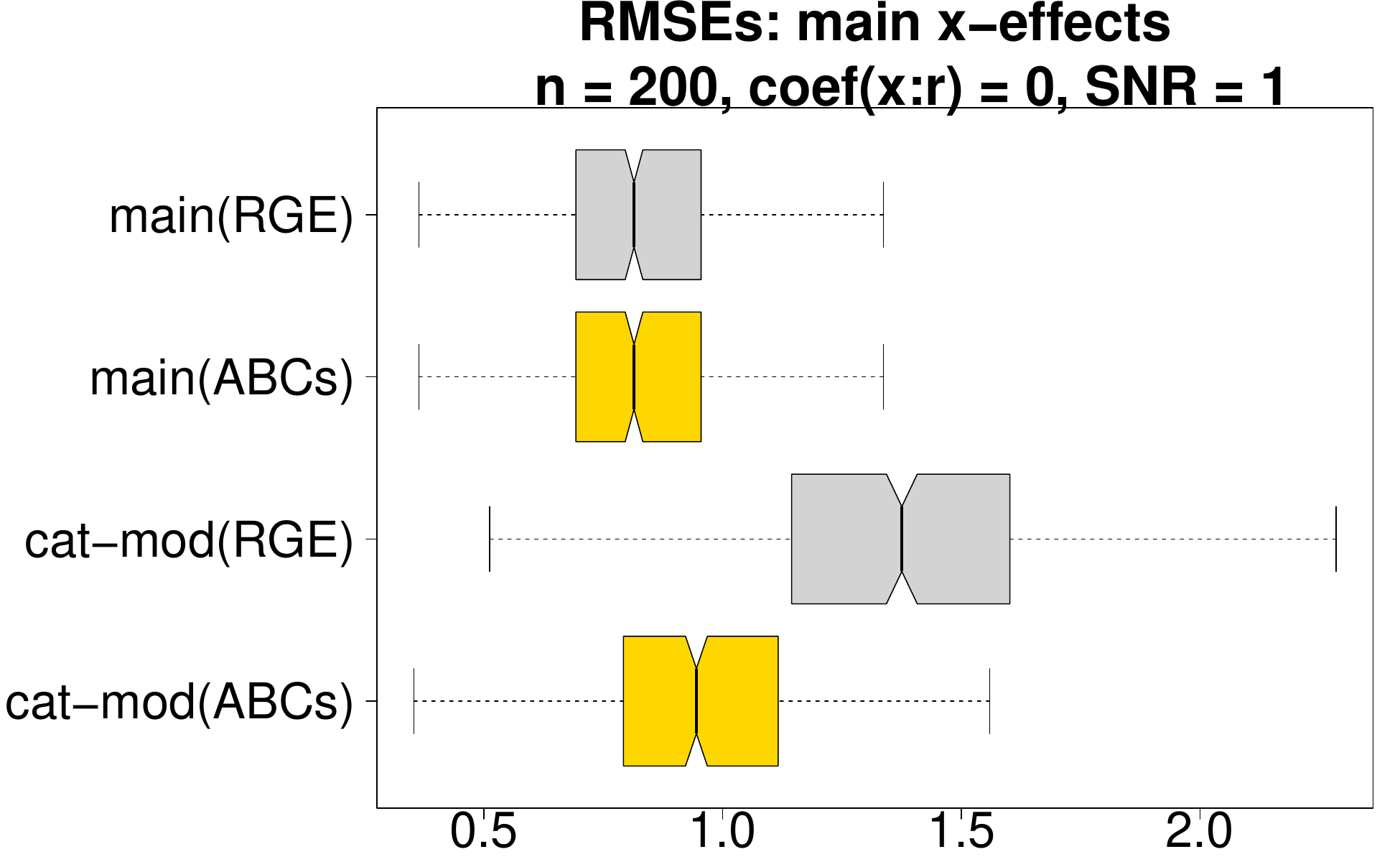}
\includegraphics[width=.49\linewidth]{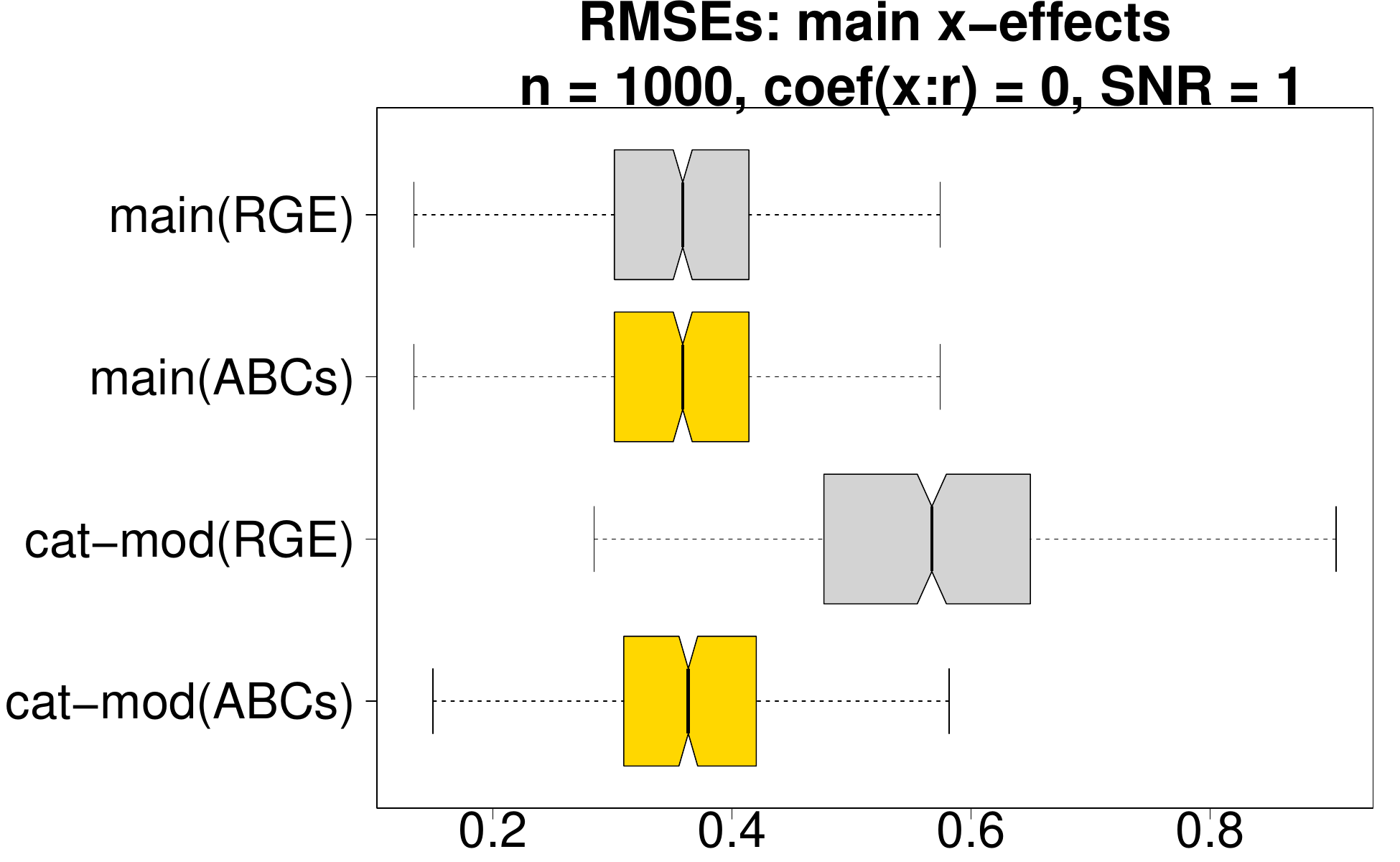}

\includegraphics[width=.49\linewidth]{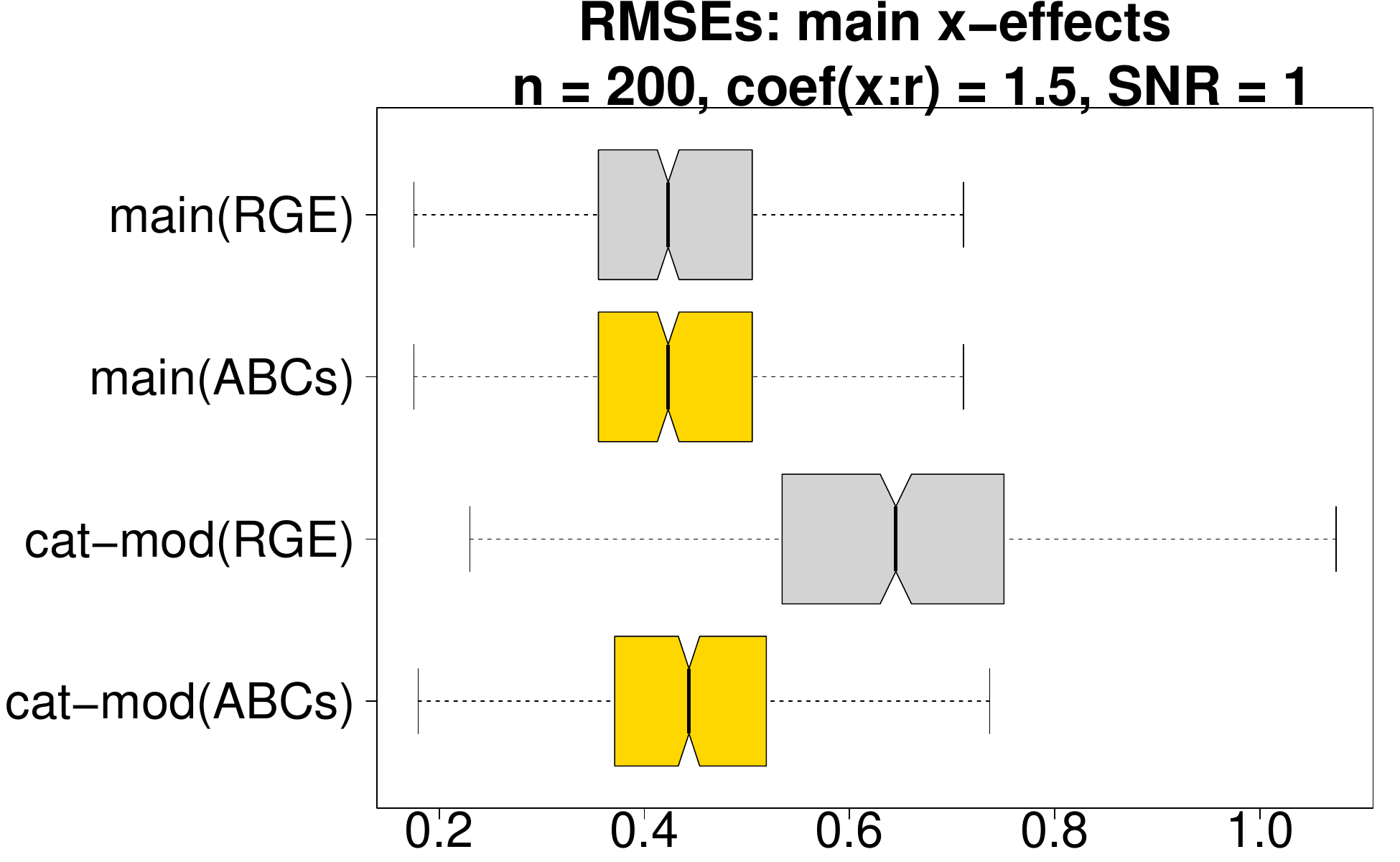}
\includegraphics[width=.49\linewidth]{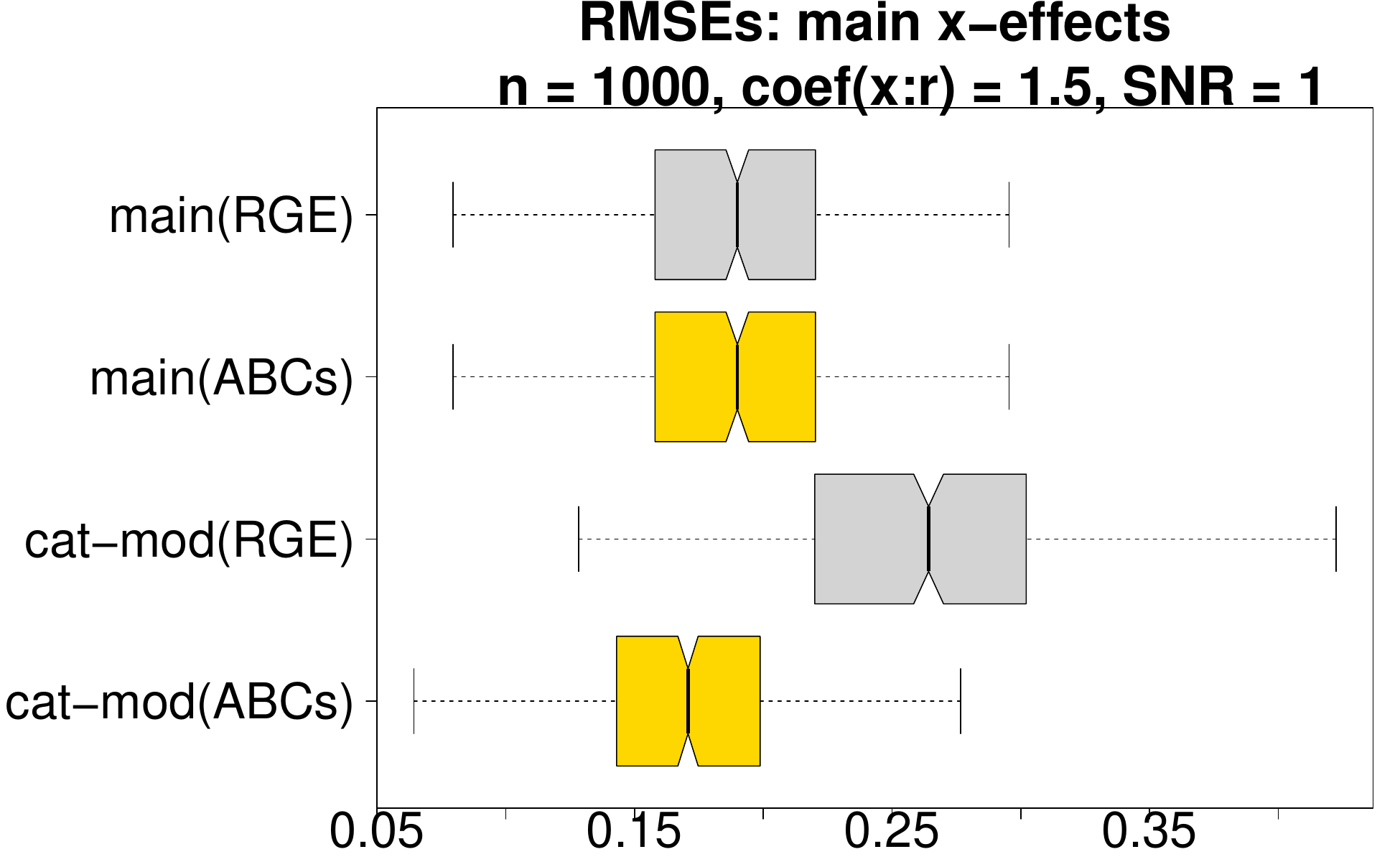}
\caption{\small RMSEs for the main $x$-effects  with extraneous (top) or necessary (bottom) cat-modifier effects for $n=200$ (left) and $n=1000$ (right) under main-only and cat-modified models with ABCs (gold)  and RGE (gray). Boxplots are across 500 simulations; nonoverlapping notches indicate a difference in medians. For $n=200$, the cat-modified models omit the \texttt{race:sex} interaction to avoid rank deficiency. For larger $n$, the cat-modified model with ABCs is better able to match (top right) or improve upon (bottom right) the main $x$-effect estimates compared to the main-only models. 
}
\label{fig:sim-est-n}
\end{figure}

\begin{figure}[h]
\centering
\includegraphics[width=.49\linewidth]{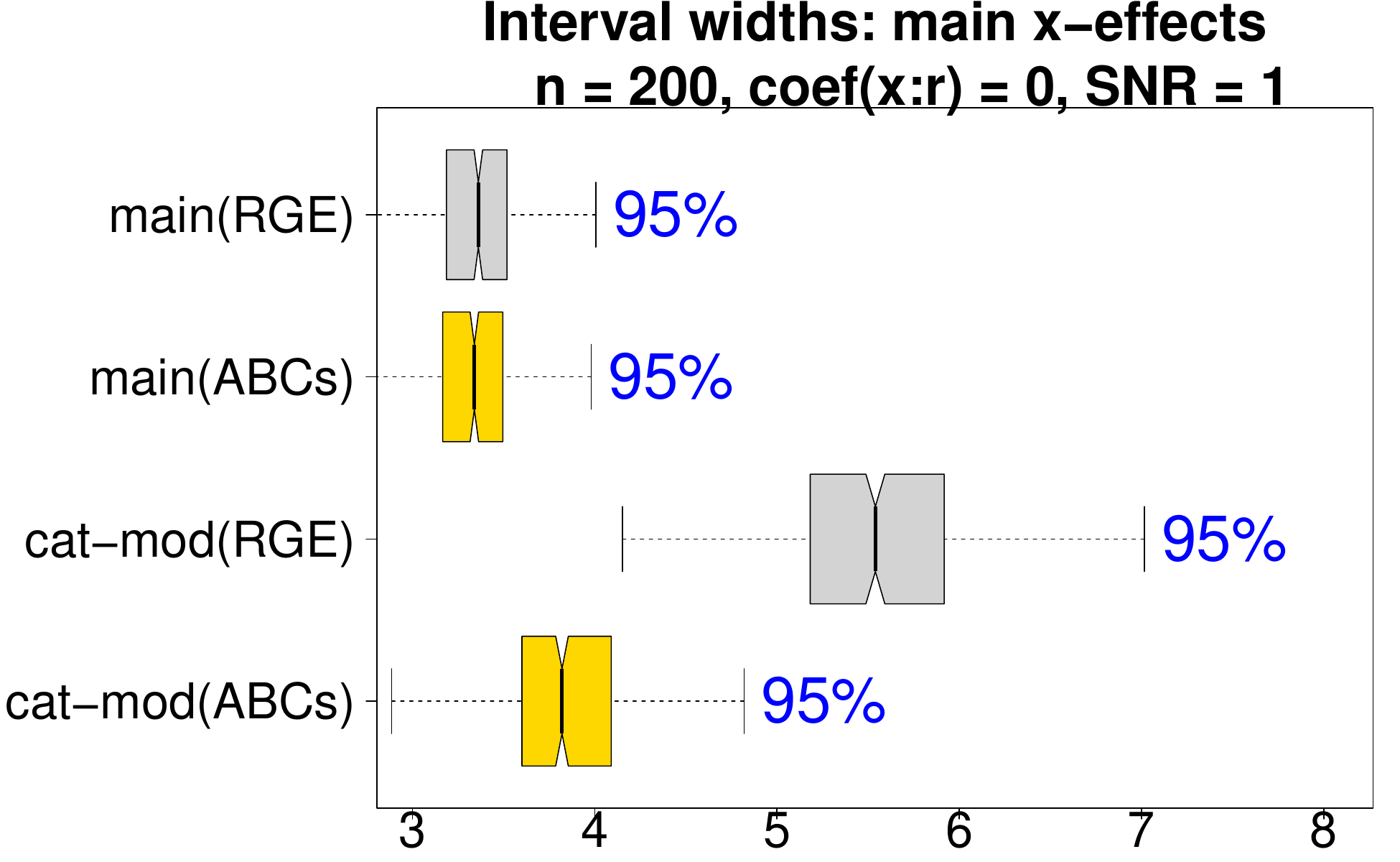}
\includegraphics[width=.49\linewidth]{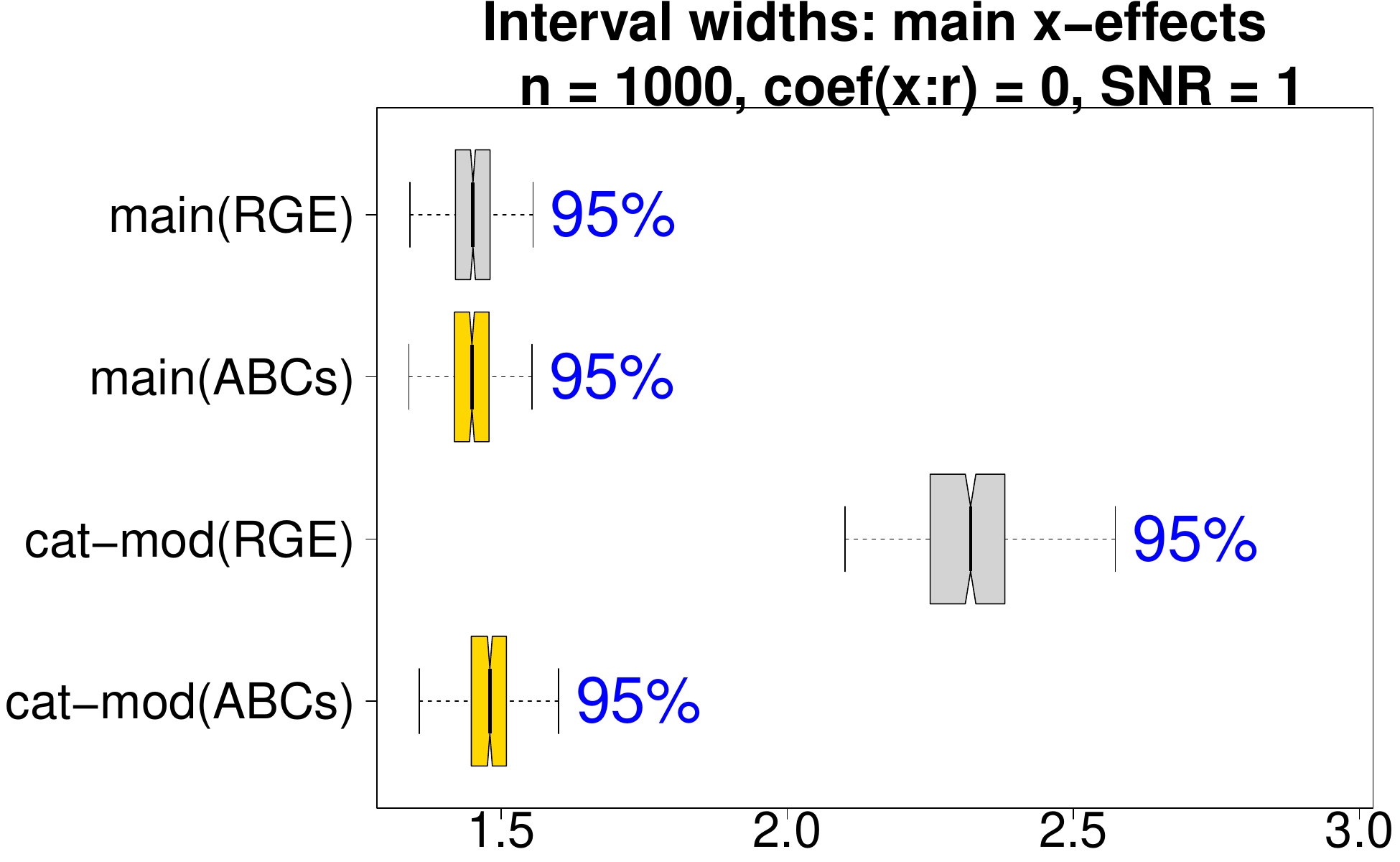}

\includegraphics[width=.49\linewidth]{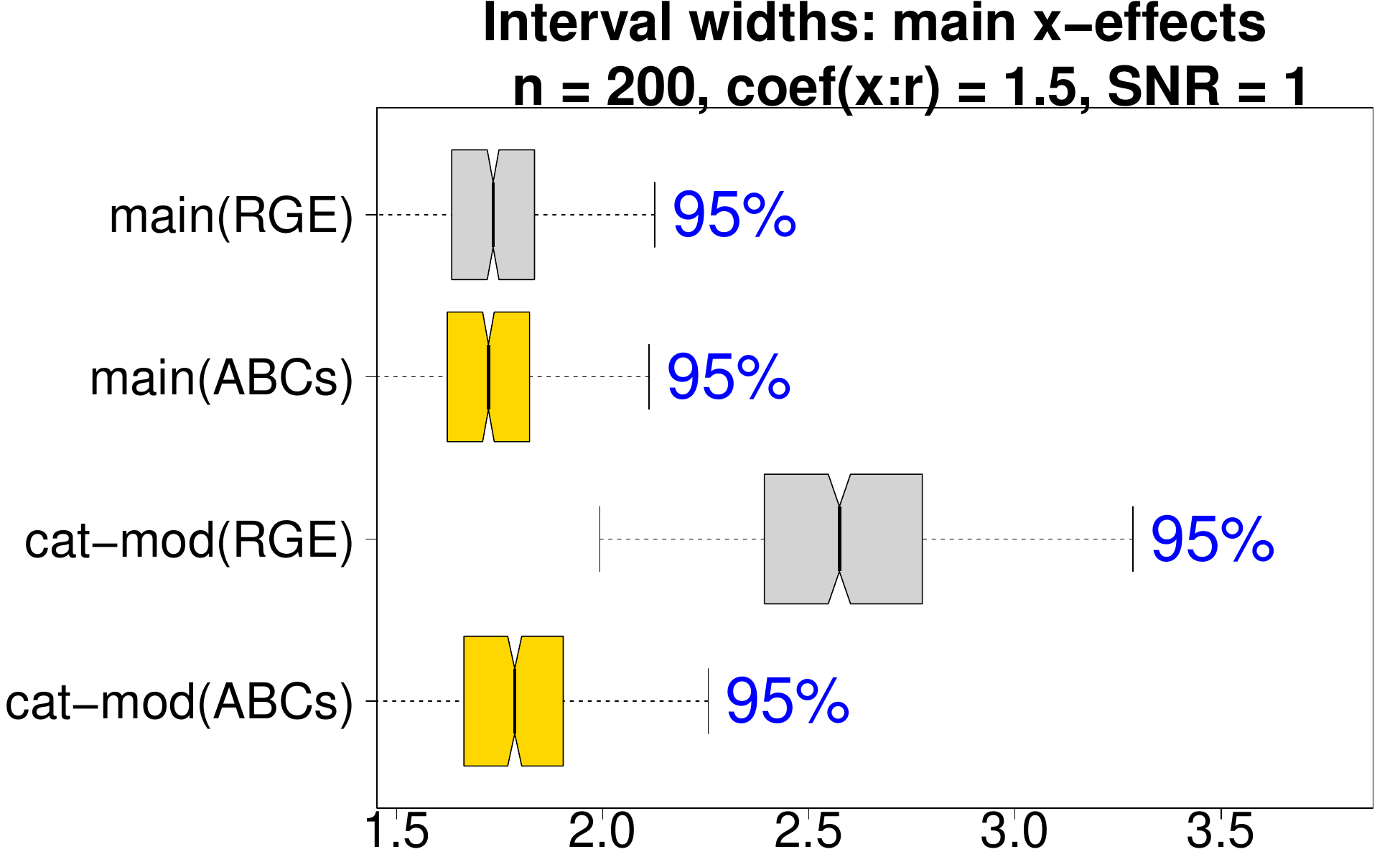}
\includegraphics[width=.49\linewidth]{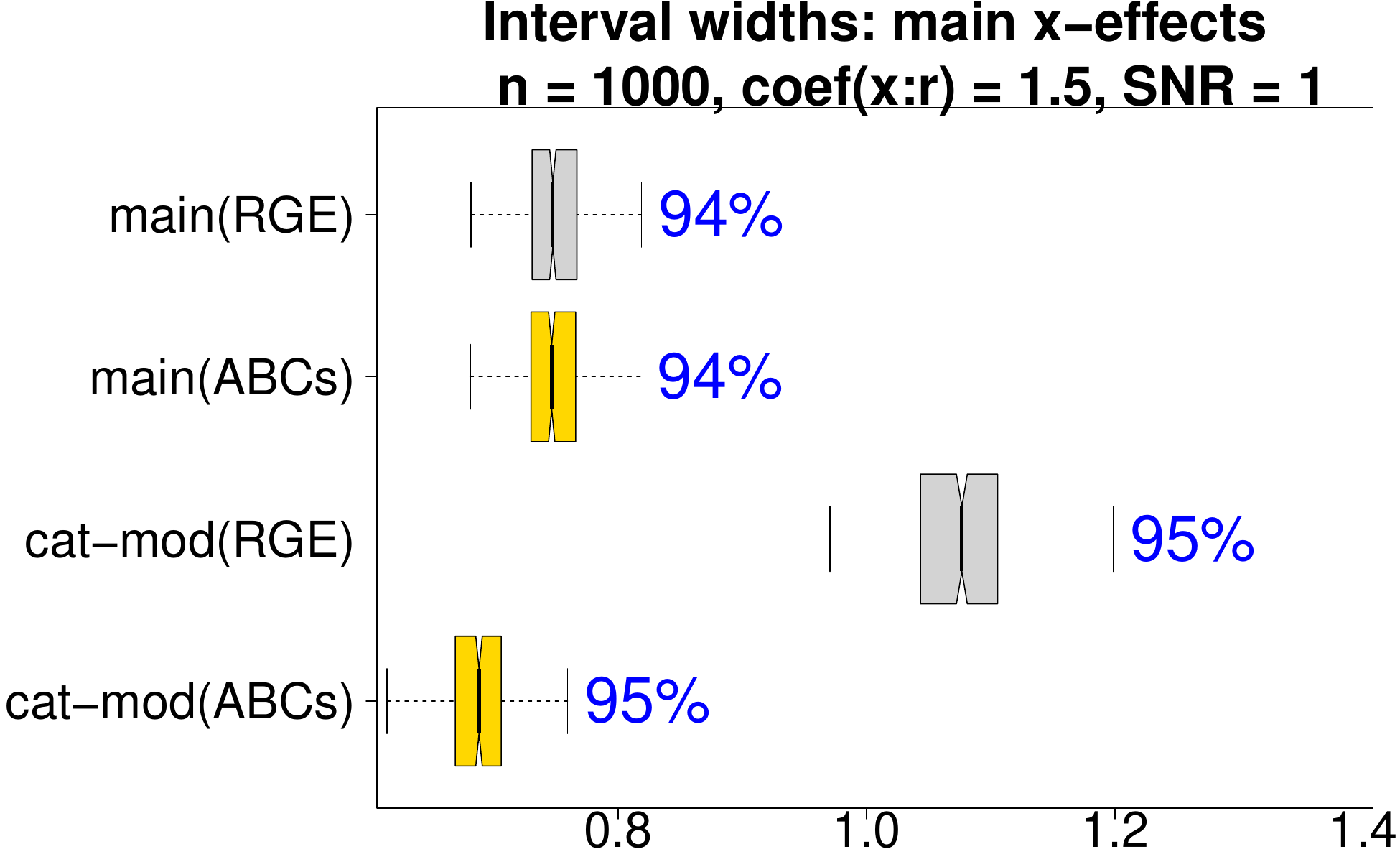}

\caption{\small Interval widths (boxplots) and empirical coverage (annotations) for 95\% confidence intervals for the main $x$-effects with extraneous (top) or necessary (bottom) cat-modifier effects for $n=200$ (left) and $n=1000$ (right) under main-only and cat-modified models with ABCs (gold)  and RGE (gray). For $n=200$, the cat-modified models omit the \texttt{race:sex} interaction to avoid rank deficiency. For larger $n$, the cat-modified model with ABCs is better able to match (top right) or improve upon (bottom right) the statistical power for the main $x$-effects compared to the main-only models.   
}
\label{fig:sim-int-n}
\end{figure}

\begin{figure}[h]
\centering
\includegraphics[width=.49\linewidth]{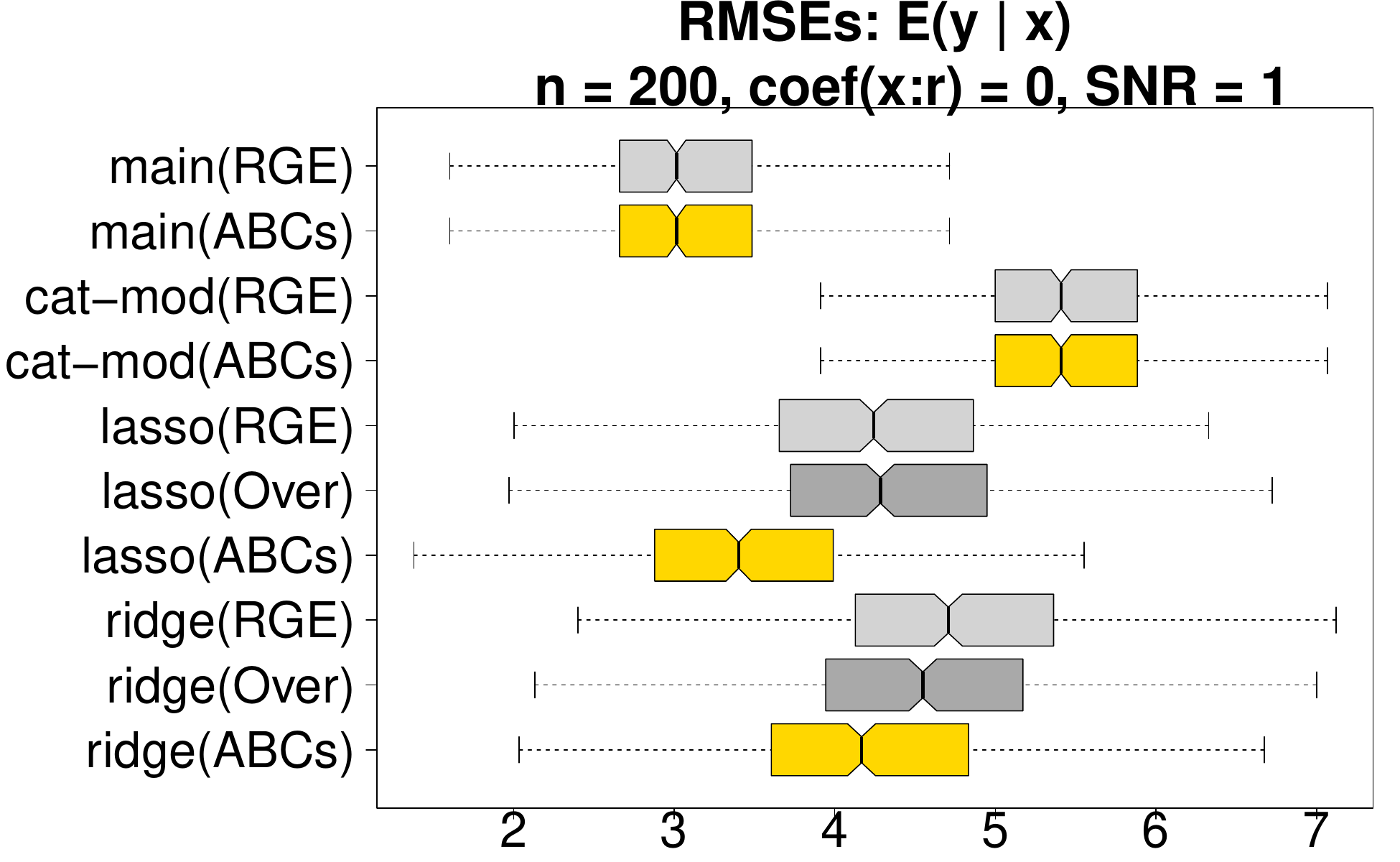}
\includegraphics[width=.49\linewidth]{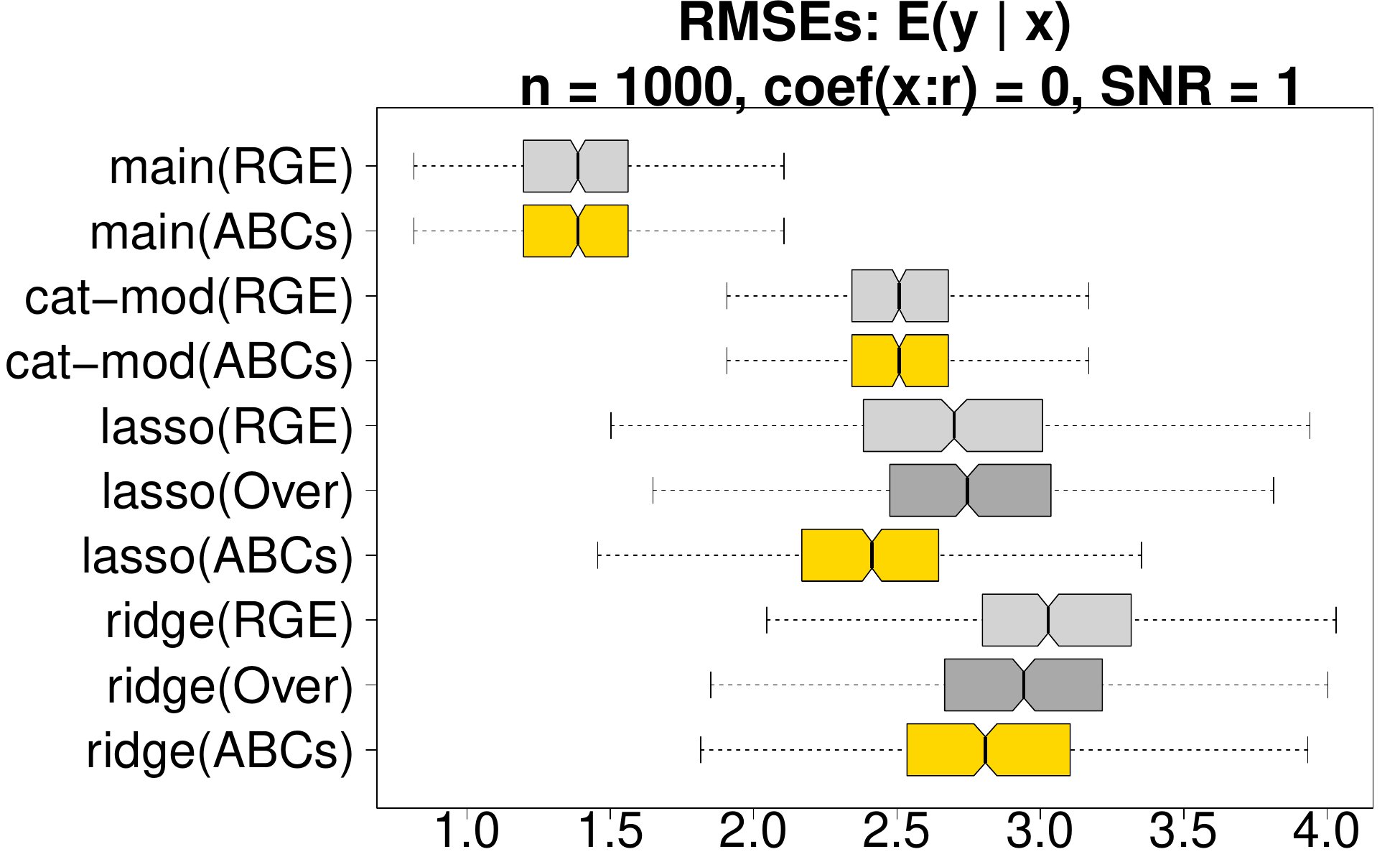}
\includegraphics[width=.49\linewidth]{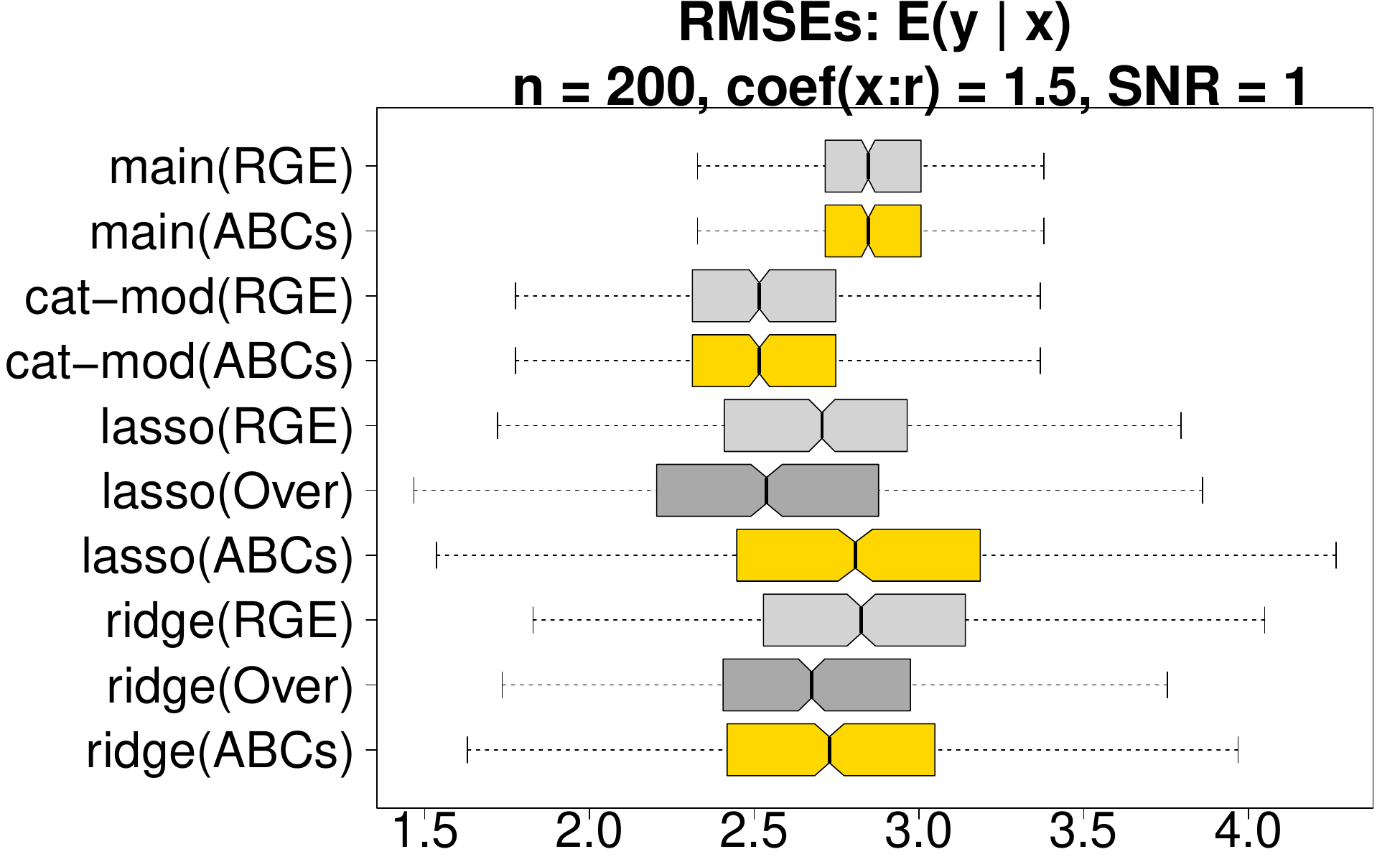}
\includegraphics[width=.49\linewidth]{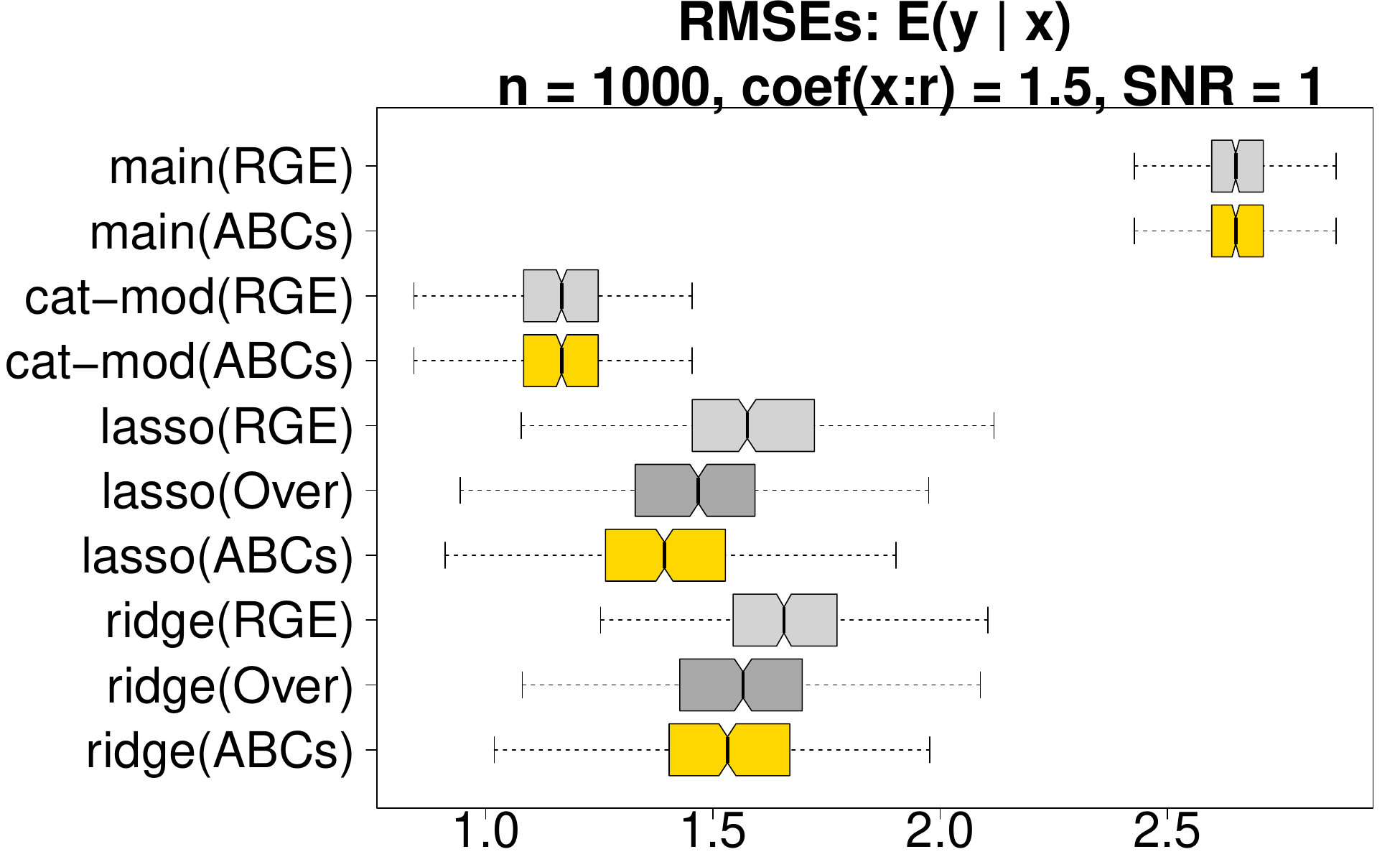}
\caption{\small RMSEs for prediction of $\mu(\bm x, r, s)$  with extraneous (top) or necessary (bottom) cat-modifier effects for $n=200$ (left) and $n=1000$ (right) under main-only and cat-modified models with ABCs (gold)  and RGE (gray). Boxplots are across 500 simulations; nonoverlapping notches indicate a difference in medians. All lasso and ridge estimators use the cat-modified model. Predictions under OLS are identical between RGE and ABCs. For $n=200$, the cat-modified models omit the \texttt{race:sex} interaction to avoid rank deficiency. For each penalized (lasso or ridge) regression, ABCs typically  outperform both RGE and the overparametrized models that omits any constraints. 
}
\label{fig:sim-pred}
\end{figure}

\clearpage
\section{Additional application details}\label{sec-a-app}

\begin{table}[h]
\centering \scriptsize
\begin{tabular}{llrr} 
\multicolumn{4}{c}{{\bf Reference group encoding (RGE)}} \\ 
Variable & Model & Estimate (SE)  & $p$-value \\
\hline
\multirow{2}{*}{\texttt{Intercept}} & Main-only & 0.238 (0.010) &   $<$0.001 \\
& Cat-modified & 0.217 (0.011)  & $<$0.001 \\
\multirow{2}{*}{\texttt{White}} & Main-only & ref& ref\\ 
& Cat-modified & ref& ref\\ 
\multirow{2}{*}{\texttt{Black}} & Main-only & -0.727 (0.013) & $<$0.001 \\
& Cat-modified & -0.664 (0.018) & $<$0.001 \\
\multirow{2}{*}{\texttt{Hispanic}} & Main-only &  -0.016 (0.025) & 0.517  \\
& Cat-modified &  -0.042 (0.035) & 0.228  \\
\multirow{2}{*}{\texttt{Female}} & Main-only & ref& ref\\ 
& Cat-modified & ref& ref\\ 
\multirow{2}{*}{\texttt{Male}} & Main-only & 0.036 (0.012) & 0.003 \\
& Cat-modified & 0.077 (0.015) & $<$0.001 \\
\hline
\texttt{White:Female} & Cat-modified& ref &  ref \\
\texttt{Black:Female} &Cat-modified & ref& ref \\
\texttt{Hisp:Female} & Cat-modified& ref& ref\\
\texttt{White:Male} & Cat-modified& ref&  ref \\
\texttt{Black:Male} &Cat-modified & -0.128 (0.025) & $<$0.001 \\
\texttt{Hisp:Male} & Cat-modified & 0.056 (0.050) & 0.262 \\
\hline
\end{tabular} 
\begin{tabular}{llrr} 
\multicolumn{4}{c}{{\bf Abundance-based constraints (ABCs)}} \\ 
Variable & Model & Estimate (SE)   & $p$-value \\
\hline
\multirow{2}{*}{\texttt{Intercept}} & Main-only &   0.000 (0.006) &   1.000 \\
& Cat-modified &   0.000 (0.006) &   1.000 \\
\multirow{2}{*}{\texttt{White}} & Main-only & 0.256 (0.005) & $<$0.001 \\
& Cat-modified & 0.256 (0.005) & $<$0.001 \\
\multirow{2}{*}{\texttt{Black}} & Main-only & -0.471 (0.008) & $<$0.001 \\
& Cat-modified & -0.471 (0.008) & $<$0.001 \\
\multirow{2}{*}{\texttt{Hispanic}} & Main-only &  0.240 (0.023) & $<$0.001  \\
& Cat-modified &  0.240 (0.023) & $<$0.001  \\
\multirow{2}{*}{\texttt{Female}} & Main-only & -0.018 (0.006) & 0.003 \\
& Cat-modified & -0.018 (0.006) & 0.003 \\
\multirow{2}{*}{\texttt{Male}} & Main-only & 0.018 (0.006) & 0.003 \\
& Cat-modified & 0.018 (0.006) & 0.003 \\
\hline 
\texttt{White:Female} & Cat-modified& -0.021 (0.005) &  $<$0.001 \\
\texttt{Black:Female} &Cat-modified & 0.043 (0.008) & $<$0.001\\
\texttt{Hisp:Female} & Cat-modified& -0.046 (0.022) & 0.034\\
\texttt{White:Male} & Cat-modified& 0.021 (0.005) & $<$0.001\\
\texttt{Black:Male} &Cat-modified & -0.044 (0.008)& $<$0.001\\
\texttt{Hisp:Male} & Cat-modified& 0.051 (0.024) & 0.034\\  
\hline
\end{tabular} 
\caption{\small Linear regression output with RGE (left) and ABCs (right) for the main-only model  \eqref{reg-main-cat} 
and the cat-modified model  \eqref{reg-cm-cat} 
for the North Carolina education data (Section~\ref{sec-app}). The (mother's) race groups are non-Hispanic White (58.7\%), non-Hispanic Black (35.1\%), and Hispanic (6.2\%) and the child's sex are \texttt{Female} (50.1\%) and \texttt{Male} (49.9\%). With RGE (references \texttt{White} and \texttt{Female}), the main effects change dramatically with the addition of cat-modifiers and the standard errors (SEs) uniformly inflate. Yet with ABCs, all main effect estimates \emph{and} SEs are invariant to cat-modifiers (the SEs actually decrease slightly;  this is obscured due to rounding). 
\label{tab:ex-cat}
}
\end{table}

\begin{table}
\centering
\begin{tabular}{lrrr}
Variable $j$ & $\hat \sigma_{x[\texttt{NHW}]}(j)$  & $\hat \sigma_{x[\texttt{NHB}]}(j)$  & $\hat \sigma_{x[\texttt{Hisp}]}(j)$ \\
\hline
Racial isolation (RI) & 
0.691 & 1.071&  0.942 \\ 
Blood lead level & 0.951 & 1.042 &  0.977 \\
Birthweight percentile for
gestational age & 0.994 &  0.963 & 0.979 \\
Mother's age &  0.999  & 0.971 & 0.889 \\
$\mbox{PM}_{2.5}$ exposure & 0.998 &  1.005 & 0.928 \\
\hline
\end{tabular}
\caption{\small The (scaled) sample standard deviations $\hat \sigma_{x[r]}(j)$ by race $r$ for each covariate $j=1,\ldots,p$. \label{tab:emp-cov}
The invariance result for estimators with and without cat-modifiers (Theorem~\ref{thm-int-full}) requires $\hat \sigma_{x[\texttt{NHW}]}(j) = \hat \sigma_{x[\texttt{NHB}]}(j) = \hat \sigma_{x[\texttt{Hisp}]}(j)$ for each covariate $j$ (and similarly for the cross-covariances). Although this condition is clearly violated, the estimates and SEs maintain invariance, which suggests strong empirical robustness for the desirable invariance property of ABCs.   
\label{tab:sigmas}}
\end{table}

\begin{table}[h] \footnotesize
\centering
\begin{tabular}{lrr} 
        Variable (continued) & Estimate (SE)  & $p$-value  \\
        \hline
        Economically disadvantaged \\ \quad (\texttt{EconDisadv})  & \\
        \quad \texttt{No} (39.5\%) & 0.163  (0.009) & $<$0.001 \\
        \quad \texttt{Yes} (60.5\%) & -0.106  (0.006) & $<$0.001 \\
        \hline 
                      \texttt{White:EconDisadvNo} & 0.010  (0.004) & 0.018 \\
        \texttt{Black:EconDisadvNo} & -0.034  (0.023) & 0.138 \\
        \texttt{Hisp:EconDisadvNo} & -0.171  (0.063) & 0.007 \\
        \texttt{White:EconDisadvYes} & -0.013  (0.006) & 0.018 \\
        \texttt{Black:EconDisadvYes} & 0.007  (0.005) & 0.138 \\
        \texttt{Hisp:EconDisadvYes} & 0.025  (0.009) & 0.007 \\
        \texttt{EconDisadvNo:Male} & -0.013  (0.008) & 0.118 \\
        \texttt{EconDisadvYes:Male} & 0.009  (0.006) & 0.118 \\
        \texttt{EconDisadvNo:Female} & 0.014 (0.009) & 0.118 \\
        \texttt{EconDisadvYes:Female} & -0.009 (0.006) & 0.118 \\
        \texttt{EconDisadvNo:mEdu<HS} & -0.056 (0.037) & 0.126 \\
        \texttt{EconDisadvYes:mEdu<HS} & 0.006 (0.004) & 0.126 \\
        \texttt{EconDisadvNo:mEdu=HS} & -0.039 (0.012) & 0.002 \\
        \texttt{EconDisadvYes:mEdu=HS} & 0.016 (0.005) & 0.002 \\
        \texttt{EconDisadvNo:mEdu>HS} & 0.020 (0.005) & $<$0.001 \\
        \texttt{EconDisadvYes:mEdu>HS} & -0.043 (0.011) & $<$0.001 \\
        \texttt{RI:EconDisadvNo} & -0.007  (0.011) & 0.513 \\
\texttt{RI:EconDisadvYes} & 0.005  (0.007) & 0.513 \\
\texttt{BLL:EconDisadvNo} & -0.011  (0.009) & 0.229 \\
\texttt{BLL:EconDisadvYes} & 0.007  (0.006) & 0.229 \\
\texttt{BWTpct:EconDisadvNo} & 0.000  (0.009) & 0.983 \\
\texttt{BWTpct:EconDisadvYes} & 0.000  (0.006) & 0.983 \\
\texttt{mAge:EconDisadvNo} & 0.016  (0.009) & 0.088 \\
\texttt{mAge:EconDisadvYes} & -0.011  (0.006) & 0.088 \\
\texttt{PM2.5:EconDisadvNo} & 0.010  (0.009) & 0.230 \\
\texttt{PM2.5:EconDisadvYes} & -0.007  (0.006) & 0.230 \\
\hline
\end{tabular}

\caption{\small Cat-modified model output under ABCs for NC STEM education outcomes. These results augment Table~\ref{tab:results} to include \texttt{EconDisadv} main and interaction effects, where ``Economically disadvantaged'' is determined by participation in the National Lunch Program. \texttt{EconDisadv} is associated with lower math scores and eliminates the significant positive benefits of higher-educated mothers (\texttt{mEdu>HS}), thus emphasizing the importance of heterogeneous effects. 
\label{tab:results-app}}
\end{table}


\end{document}